%% file: aistats22.tex
\documentclass[twoside]{article}

\usepackage[accepted]{aistats2022}

\usepackage{natbib}

\usepackage{hyperref}       %
\usepackage{url}            %
\usepackage{booktabs}       %
\usepackage{amsfonts}       %
\usepackage{nicefrac}       %
\usepackage{microtype}      %
\usepackage{mdframed}
\usepackage{lipsum}
\usepackage{dsfont}
\usepackage{graphicx}
\usepackage{amsmath}
\usepackage{amsthm}
\usepackage{amssymb}
\usepackage{enumitem}
\usepackage{algorithm}
\usepackage{xspace}
\usepackage{multirow}
\usepackage[table,xcdraw]{xcolor}
\usepackage[noend]{algorithmic}
\usepackage{comment}
\usepackage{multibib}

\newcommand{\R}{\mathbb{R}}
\newcommand{\E}{\mathbb{E}}
\newcommand{\N}{\mathbb{N}}
\newcommand{\U}{\mathbb{U}}
\newcommand{\ours}{\textbf{HadaHeavy}\xspace}
\newcommand{\oracle}{\textbf{HadaOracle}\xspace}
\newcommand{\ttime}{\textit{time}\xspace}
\newcommand{\mem}{\textit{mem}\xspace}
\newcommand{\hrr}{\textbf{HRR}\xspace}
\newcommand{\freqOracle}{\textbf{FreqOracle}\xspace}
\newcommand{\hashtogram}{\textbf{Hashtogram}\xspace}
\newcommand{\ldp}{\textbf{LDP}\xspace}
\def\Var{\mathbb{V}\text{ar}}
\def\Cov{\mathbb{C}\text{ov}}

\def\dsone{\mathds{1}}
\def\mcalA{\mathcal{A}}
\def\mcalD{\mathcal{D}}

\def\mcalU{\mathcal{U}}

\def\mcalY{\mathcal{Y}}

\def\mcalP{\mathcal{P}}
\def\med{ \textbf{Median} }

\def\mrmH{\mathrm{H}}
\def\mbbS{\mathbb{S}}
\def\mbfX{\mathbf{X}}
\def\mbfS{\mathbf{S}}
\def\pmbc{\pmb{c}}
\def\pmbe{\pmb{e}}
\def\pmbx{\pmb{x}}
\def\pmby{\pmb{y}}
\def\pmbs{\pmb{s}}
\def\pmbt{\pmb{t}}

\def\pmbomega{\pmb{\omega}}
\def\gsize{n_{ \text{grp} }}

\def\eps{\varepsilon}
\def\ceps{ C_\eps }

\def\errOracle{ \textit{Err} }
\def\goodset{ \mathbb{G}\textit{d}\,\mathbb{S}\textit{}{et} }
\def\goodsetA{ \mathbb{G}\textit{d}\,\mathbb{S}\textit{et}_1 }
\def\goodsetB{ \mathbb{G}\textit{d}\,\mathbb{S}\textit{et}_2 }
\def\goodsetC{ \mathbb{G}\textit{d}\,\mathbb{S}\textit{et}_3 }

\def\algoOracle{\mathcal{A}_{ \textit{oracle} }}

\def\algoHrrClient{\mathcal{A}_{ \text{HRR-client} }}
\def\algoHrrServer{\mathcal{A}_{ \text{HRR-server} }}

\def\algoTreeHist{\textbf{TreeHist}}
\def\algoProtSHIST{\textbf{PROT-S-Hist}}
\def\algoBitstogram{\textbf{Bitstogram}}
\def\algoPrivateExpanderSketch{\textbf{PrivateExpanderSketch}}

\theoremstyle{plain}
\newtheorem{thm}{Theorem}[section]
\newtheorem{lemma}[thm]{Lemma}
\newtheorem{corollary}[thm]{Corollary}
\newtheorem{fact}[thm]{Fact}

\theoremstyle{definition}
\newtheorem{definition}[thm]{Definition}
\newtheorem*{fact*}{Fact}

\newtheorem{remark}{Remark}[section]

\newcommand{\tony}[1]{}
\newcommand{\hao}[1]{}
\newcommand{\highlight}[1]{}

\definecolor{Mulberry}{rgb}{0.77,0.29,0.55}
\definecolor{CadmiumOrange}{rgb}{0.93,0.53, 0.18}
\definecolor{ForestGreen}{rgb}{0.13, 0.55, 0.13}
\definecolor{WildStrawberry}{rgb}{0.5, 0.7, 0.2}
\renewcommand{\tony}[1]{{\color{Mulberry}{\textbf{T says:}} \emph{#1}}} 
\renewcommand{\hao}[1]{{\color{ForestGreen}{\textbf{H says:}} \emph{#1}}} 
\renewcommand{\highlight}[1]{{\color{Mulberry}{\textbf{}} #1}}

\begin{document}

\twocolumn[

\aistatstitle{Asymptotically Optimal Locally Private Heavy Hitters via Parameterized Sketches}

\aistatsauthor{ Hao Wu \And Anthony Wirth }

\aistatsaddress{School of Computing \& Information Systems\\ The University of Melbourne \And  School of Computing \& Information Systems\\ The University of Melbourne } ]

\begin{abstract}
    \vspace{-2mm}
    We study the frequency estimation problem under the local differential privacy model.
    Frequency estimation is a fundamental computational question, and differential privacy has become the de-facto standard, with the local version (LDP) affording even greater protection.
    On large input domains, sketching methods and hierarchical search methods are commonly and successfully, in practice, applied for reducing the size of the domain, and for identifying frequent elements.
    It is therefore of interest whether the current theoretical analysis of such algorithms is tight, or whether we can obtain algorithms in a similar vein that achieve optimal error guarantee. 
    
    We introduce two algorithms for LDP frequency estimation.
    One solves the fundamental frequency oracle problem; the other solves the well-known heavy hitters identification problem. 
    As a function of failure probability,~$\beta$, the former achieves optimal worst-case estimation error for every~$\beta$; the latter is optimal when~$\beta$ is at least inverse polynomial in~$n$, the number of users.
    In each algorithm, server running time and memory usage are~$\tilde{O}(n)$ and~$\tilde{O}(\sqrt{n})$, respectively,  while user running time and memory usage are both~$\tilde{O}(1)$.
    Our frequency-oracle algorithm achieves lower estimation error than \citeauthor{BNST17} (NeurIPS \citeyear{BNST17}).
    On the other hand, our heavy hitters identification method improves the worst-case error of \algoTreeHist{} (ibid) by a factor of $\Omega(\sqrt{\log n})$; it avoids invoking error-correcting codes, known to be theoretically powerful, but yet to be implemented efficiently.
\end{abstract}

\section{Introduction}

Frequency estimation is a fundamental computation task, widely applied in data mining and machine learning, e.g., learning users’ preferences~\citep{EPK14}, uncovering commonly used phrases~\citep{D.P.Apple17}, and finding popular URLs~\citep{FPE16}.
We expect entities that collect such data to respect their users' privacy, and there are increasing stringent regulations~\citep{voigt2017eu}.
How can we infer and estimate frequency, and thus improve users' experience, without sacrificing personal privacy?

In answering such questions, local differential privacy (\ldp) becomes a popular data collection model for providing user level privacy protection~\citep{EPK14, FPE16, D.P.Apple17, TKBWW17, DKY17}.
In this model, there is a server and a set~$\mcalU{}$ of $n$ users, each holding an element from some domain $\mcalD{}$ of size $d$. 
No user $u \in \mcalU{}$ wants to share their data $v^{ (u) } \in \mcalD{}$ directly with the server. 
To protect sensitive personal information, they run a local randomizer $\mcalA{}^{ (u) }$ to perturb their data.
The server collects only the perturbed data. 
Formally, the algorithm $A^{ (u) }$ is called $\eps{}$-\emph{local differentially private} ($\eps{}$-\ldp) if its output distribution varies little with the input, as defined thus.

\begin{definition}[$\eps{}$-\ldp~\citep{DR14}] \label{def: Differential Privacy}
    Let $\mcalA: \mcalD \rightarrow \mcalY$ be a randomized algorithm mapping an element in $\mcalD{}$ to $\mcalY{}$. 
    We say $\mcalA$ is $\eps{}$-local differentially private if for all $v, v' \in \mcalD$ and all (measurable) $Y \subseteq \mcalY$,
    $$
        \Pr[ \mcalA (v) \in Y ] \le e^\eps \cdot \Pr [ \mcalA (v') \in Y ]\,.
    $$
\end{definition}  
Aligned with prior art~\citep{BS15, BNST17, BNS19, cormode2021frequency}, in this work, we study two closely related, but distinct, functionalities of frequency estimation under the \ldp model: the frequency oracle and the succinct histogram. 
The relevant parameters are the error threshold,~$\lambda$, and the failure probability,~$\beta$:

\begin{definition}[Frequency Oracle] 
     A frequency oracle, denoted as \textbf{FO}, is an algorithm that provides for every $v \in \mcalD{}$, an estimated $\hat f_\mcalU{} [ v ]$ of the frequency of~$v$, denoted as 
     $
        f_\mcalU{} [ v ] \doteq | \{ u \in \mcalU{} : v^{ (u) } = v \} |
     $, 
     such that $\Pr[ | \hat f_\mcalU{} [ v ] - f_\mcalU{} [ v ] | \ge \lambda ] \le \beta$. 
\end{definition}

\begin{definition}[Succinct Histogram]
   A succinct histogram, denoted as \textbf{S-Hist}, is a set of (element, estimate) pairs $\subseteq \mcalD{} \times \R$, of size $O( n / \lambda)$, such that with probability at least $1 - \beta$: i) $\forall v \in \mcalD{}$, if $f_\mcalU{} [v] \ge \lambda$, $(v, \hat f_\mcalU{} [v] )$ belongs to the set; ii) and if $(v, \hat f_\mcalU{} [ v] )$ is in the set, then $| f_\mcalU{} [v] - \hat f_\mcalU{} [v] | \le \lambda$.
\end{definition}

Each element in the Succinct Histogram set is called a \emph{heavy hitter}.
For a fixed $\eps \le 1$, the goal of algorithm design
for both problems in \ldp model is to minimize the error threshold~$\lambda$, while also limiting server/user running time, memory usage, and communication cost.

A number of frequency oracle algorithms~\citep{Warner65, EPK14, BS15, BNST17, WangBLJ17} have been proposed (see~\cite{cormode2021frequency} for a recent survey).
These algorithms achieve error~$O ( (1 / \eps) \cdot \sqrt{ n \ln ( {1} / {\beta} ) } )$.
They have running time, or memory usage that scale linearly with~$d$, the size of the data domain, and work well when it is small. 
The heavy hitters can be discovered by querying the frequencies of all elements in the domain~$\mcalD{}$.
Via union bound, this achieves error~$O ( (1 / \eps) \cdot \sqrt{ n \ln ( d / {\beta} ) } )$.
It can be shown that these error guarantees are optimal~\citep{BNS19}.
However, consider the scale of modern applications, e.g.,
finding popular URLs with length up to~20 characters\footnote{Valid URL characters include digits ($0$-$9$), letters(A-Z, a-z), and a few special characters ("-" , "." , "\_" , "\texttildelow").}~\citep{FPE16}, which results in a domain of size larger than $10^{36}$.

{\it Sketching methods} for reducing the size of the data domain, and {\it hierarchical searching methods} for avoiding inspecting the frequency of each element, are well known.
The former applies hash functions to map elements from the original domain to a smaller one; the latter views elements as strings defined over some alphabet, and identifies the heavy hitters by one or a few characters each time. 
Due to their simplicity, they are widely applied in designing frequency estimation algorithms, and heuristics are proposed to improve their performance~\citep{BNST17, D.P.Apple17, FPE16, BNST20, WangXYHSS018, WangBLJ17,cormode2021frequency}.
These implementations perform well in practice.

The best known error guarantees of frequency estimation algorithms based on the {\it sketching} and {\it hierarchical searching methods} are provided by the seminal work~\citep{BNST20}.
The frequency oracles, \freqOracle and \hashtogram~\citep{BNST20}, guarantee only an error of $O ( (1 / \eps) \cdot \sqrt{ n \ln ( {n} / {\beta} ) } )$.
The succinct histogram algorithm, \algoTreeHist{}~\citep{BNST20}, guarantees an error of~$O ( (1 / \eps ) \cdot  \sqrt{ n \cdot ( \ln d ) \cdot \ln (n / \beta)  } )$.
These algorithms exhibit low time complexity and memory usage: with server running time $\tilde O(n)$ and memory usage $\tilde O(\sqrt{n})$.
But the error guarantees are sub-optimal.
\begin{quote}
\textbf{Research Question:}
Are the theoretical error guarantees of the algorithms based on sketching and hierarchy methods best possible, or can we obtain algorithms of this type that achieve optimal error guarantee?
\end{quote}

There is another line of research for succinct histogram that relies on error-correcting codes. 
\citeauthor{BS15}~\citep{BS15} were the pioneers, with \algoProtSHIST{}, which  \algoBitstogram{}~\citep{BNST20} subsequently improved upon.
This culminates in the work of \algoPrivateExpanderSketch{}~\citep{BNS19} that achieves the optimal error guarantee~$O ( (1 / \eps) \cdot \sqrt{ n \ln ( d / {\beta} ) } )$.
But due to the sophistication of error-correcting codes, none of these methods has been implemented or significantly deployed.
Indeed, the original paper~\citep{BNST20} that proposed both \algoTreeHist{} and \algoBitstogram{} only implemented \algoTreeHist{}. 
Therefore, there is a gap between the error guarantees of the theoretically best algorithm, and the ones deployed in practice.
Determining whether we can bridge the gap answers our Research Question.

\subsection{Our Contributions}

\vspace{-1mm}
Our work provides positive answers to the questions.
In particular, we: (1) design a frequency oracle, \oracle, based on {\it sketching method} with optimal error guarantee; (2) design a succinct histogram algorithm, \ours, based on {\it hierarchical search} that achieves optimal error guarantee under mild assumption of the failure probability. 

We introduce the martingale method into the analysis of the {\it sketching method}.
We prove that, when the proper sketch is chosen, it can be incorporated into a family of frequency oracles to reduce running time and memory, while maintaining the frequency oracles' error guarantee. 
This leads to \oracle with optimal error~$O ( (1 / \eps) \cdot \sqrt{ n \ln ( {1} / {\beta} ) } )$.
Based on the \oracle, we develop a {\it hierarchical search} algorithm, \ours, that explores a large number of elements at each search step, and achieves error~$O ( (1 / \eps ) \cdot  \sqrt{ n \cdot (\ln d ) \cdot ( 1 + ( \ln (1 / \beta)  \ / \ln n ) ) } )$. 
Consistent with the theory community's view of an algorithm that succeeds with high probability, if the failure probability,~$\beta$, is inverse polynomial, i.e., $\beta = 1 / n^{ O(1) }$, the error matches the lower bound~\citep{BNS19}.
All these algorithms have running time~$\tilde{O}(n)$ and memory usage~$\tilde{O} ( \sqrt{n} )$.
Table~\ref{tab:performance_comparison} summarizes the comparisons.

\begin{table*}[!ht]
\centering
\resizebox{\textwidth}{!}{%
\begin{tabular}{cccccl}
\hline
\multicolumn{2}{|c|}{Performance Metric} &
  \multicolumn{1}{c|}{{\color[HTML]{000000} \begin{tabular}[c]{@{}c@{}}Server \\ Time\end{tabular}}} &
  \multicolumn{1}{c|}{{\color[HTML]{000000} \begin{tabular}[c]{@{}c@{}}Server \\ Mem\end{tabular}}} &
  \multicolumn{1}{c|}{{\color[HTML]{000000} Worst-Case Error}} &
  \multicolumn{1}{l|}{Lower Bound} \\ \hline
\vspace{-3mm} &
   &
   &
   &
   &
   \\ \hline
\multicolumn{1}{|c|}{{\color[HTML]{00009B} }} &
  \multicolumn{1}{c|}{{\color[HTML]{00009B} $\oracle$}} &
  \multicolumn{1}{c|}{$\tilde O(n)$} &
  \multicolumn{1}{c|}{$\tilde O( \sqrt{n} )$} &
  \multicolumn{1}{c|}{$O \left( \frac{1}{ \eps } \sqrt{ n \cdot \ln \frac{1}{\beta} } \right)$} &
  \multicolumn{1}{c|}{} \\ \cline{2-5}
\multicolumn{1}{|c|}{{\color[HTML]{00009B} }} &
  \multicolumn{1}{c|}{{\color[HTML]{00009B} $\hrr$ \citep{NXYHSS16, CKS19}}} &
  \multicolumn{1}{c|}{$\tilde O(d)$} &
  \multicolumn{1}{c|}{$\tilde O( d )$} &
  \multicolumn{1}{c|}{$O \left( \frac{1}{ \eps } \sqrt{ n \cdot \ln \frac{1}{\beta} } \right)$} &
  \multicolumn{1}{c|}{} \\ \cline{2-5}
\multicolumn{1}{|c|}{{\color[HTML]{00009B} }} &
  \multicolumn{1}{c|}{{\color[HTML]{00009B} $\freqOracle$ \citep{BNST17}}} &
  \multicolumn{1}{c|}{$\tilde O(n)$} &
  \multicolumn{1}{c|}{$\tilde O( \sqrt{n} )$} &
  \multicolumn{1}{c|}{$O \left( \frac{1}{ \eps } \sqrt{ n \cdot \ln \frac{n}{\beta} } \right)$} &
  \multicolumn{1}{c|}{} \\ \cline{2-5}
\multicolumn{1}{|c|}{\multirow{-4}{*}{{\color[HTML]{00009B} \rotatebox[origin=c]{90}{$\qquad \textbf{FO}$}}}} &
  \multicolumn{1}{c|}{{\color[HTML]{00009B} $\hashtogram$ \citep{BNST17}}} &
  \multicolumn{1}{c|}{$\tilde O(n)$} &
  \multicolumn{1}{c|}{$\tilde O( \sqrt{n} )$} &
  \multicolumn{1}{c|}{$O \left( \frac{1}{ \eps } \sqrt{ n \cdot \ln \frac{n}{\beta} } \right)$} &
  \multicolumn{1}{c|}{\multirow{-4}{*}{$O \left( \frac{1}{ \eps } \sqrt{ n \cdot \ln \frac{1}{\beta} } \right)$}} \\ \hline
\vspace{-3mm} &
   &
   &
   &
   &
   \\ \hline
\multicolumn{1}{|c|}{{\color[HTML]{036400} }} &
  \multicolumn{1}{c|}{{\color[HTML]{036400} $\ours$}} &
  \multicolumn{1}{c|}{$\tilde O(n)$} &
  \multicolumn{1}{c|}{$\tilde O( \sqrt{n} )$} &
  \multicolumn{1}{c|}{$O \left( \frac{1}{ \eps } \sqrt{ n \cdot (\ln d) \cdot \Big( 1 + \frac{\ln ( 1 / \beta ) }{ \ln n } } \Big) \right)$} &
  \multicolumn{1}{l|}{} \\ \cline{2-5}
\multicolumn{1}{|c|}{{\color[HTML]{036400} }} &
  \multicolumn{1}{c|}{{\color[HTML]{036400} $\algoTreeHist{}$ \citep{BNST17}}} &
  \multicolumn{1}{c|}{$\tilde O(n)$} &
  \multicolumn{1}{c|}{$\tilde O( \sqrt{n } )$} &
  \multicolumn{1}{c|}{$O \left( \frac{1}{ \eps } \sqrt{ n \cdot ( \ln d ) \cdot \ln \frac{n}{\beta} }  \right)$} &
  \multicolumn{1}{l|}{} \\ \cline{2-5}
\multicolumn{1}{|c|}{{\color[HTML]{036400} }} &
  \multicolumn{1}{c|}{{\color[HTML]{036400} $\textbf{PrivateExpanderSketch}$ \citep{BNS19}}} &
  \multicolumn{1}{c|}{$\tilde O(n)$} &
  \multicolumn{1}{c|}{$\tilde O( \sqrt{n} )$} &
  \multicolumn{1}{c|}{$O \left( \frac{1}{ \eps } \sqrt{ n \cdot \ln \frac{d}{\beta} } \right)$} &
  \multicolumn{1}{l|}{} \\ \cline{2-5}
\multicolumn{1}{|c|}{\multirow{-4}{*}{{\color[HTML]{036400} \rotatebox[origin=c]{90}{$\qquad \textbf{S-Hist}$}}}} &
  \multicolumn{1}{c|}{{\color[HTML]{036400} $\algoBitstogram{}$ \citep{BNST17}}} &
  \multicolumn{1}{c|}{$\tilde O(n)$} &
  \multicolumn{1}{c|}{$\tilde O( \sqrt{n} )$} &
  \multicolumn{1}{c|}{$O \left( \frac{1}{ \eps } \sqrt{ n \cdot ( \ln \frac{d}{\beta} ) \cdot \ln \frac{1}{\beta} }  \right)$} &
  \multicolumn{1}{l|}{\multirow{-4}{*}{$O \left( \frac{1}{ \eps } \sqrt{ n \cdot \ln \frac{d}{\beta} } \right)$}} \\ \hline
\end{tabular}
}
\caption{
    Comparison of our frequency oracle ($\oracle$) and succinct histogram ($\ours$) algorithms with the state-of-the-art, where `Mem' stands for `Memory'.
    For all algorithms, each user has~$\tilde O(1)$~memory, takes~$\tilde O(1)$ running time, requires~$\tilde O(1)$ public randomness, and reports $O(1)$ bits to the server. 
}
\label{tab:performance_comparison}
\vspace{-3mm}
\end{table*}

\section{Preliminaries} \label{sec: preliminaries}

\vspace{-1mm}
\subsection{Hadamard Randomized Response}

Our algorithms invoke the frequency oracle, named $\hrr$~\citep{NXYHSS16, CKS19}, as a subroutine.

\begin{fact}[Algorithm \hrr~\citep{NXYHSS16, CKS19}] \label{fact: hrr}
    Let $\mcalU{}$ be a set users each holding an element from some finite domain $\mcalD{}$.
    There exists an $\eps$-locally differentially private frequency oracle, \hrr, such that the following holds.
    Fix some query element~$v \in \mcalD{}$ for \hrr. 
    With probability at least~$1 - \beta'$, \hrr returns a frequency estimate $\hat f_\mcalU{} [ v ]$ satisfying
    $$
    {\small
            \left|  \hat f_\mcalU{} [ v ] - f_\mcalU{} [ v ] \right| \in O \left( ( {1} / {\eps} ) \cdot \sqrt{ | \mcalU{} | \cdot \ln ( {1} / {\beta'} )  } \right)\,.
    }
    $$
    Each user in $\mcalU{}$ requires~$\tilde O(1)$~memory, takes~$\tilde O(1)$ running time and reports only~$1$ bit to the server. 
    The server processes the reports in~$\tilde O( |\mcalU{} | + |\mcalD{}| )$ time and~$O(|\mcalD{}|)$ memory, and answers a query in~$\tilde O(1)$ time. 
    The~$\tilde O$ notation hides logarithmic factors in $| \mcalU{} |$, $| \mcalD{} |$ and $1 / \beta'$. 
\end{fact}

The Appendix includes a proof of this fact.

\subsection{Lower Bounds}

\cite{BNS19} 
provide a lower bound for the succinct histogram problem.

\begin{fact}[\cite{BNS19}] \label{fact:lower bound succinct histogram}
    Let~$\eps{} \in O(1)$. 
    Every~$\eps{}$-\ldp algorithm for estimating the frequencies of all %
    elements from $\mcalD{}$,  must have, with probability at least~$1 - \beta$,
    \vspace{-2mm}
    $$
    {\small
            \max_{v \in \mcalD{} } \left|  \hat f_\mcalU{} [ v ] - f_\mcalU{} [ v ] \right| \in \Omega \left( ( {1} / {\eps} ) \cdot \sqrt{ | \mcalU{} | \cdot \ln ( { |\mcalD{} |  } / {\beta} )  } \right).
    }
    $$
\end{fact}

We can obtain a lower bound for the frequency oracles, via a union bound argument, with~${\beta'} = \beta / |\mcalD{}|$.

\begin{corollary}\label{thm:lower bound frequency oracle}
    Let~$\eps{} \in O(1)$. 
    Every $\eps{}$-\ldp frequency oracle algorithm achieving estimation error~$\lambda$ with probability at least~$1 - \beta'$ must have
    \vspace{-2mm}
    $$
        \lambda \in \Omega \left( ( {1} / {\eps} ) \cdot \sqrt{ | \mcalU{} | \cdot \ln ( { 1 } / {\beta'} )  } \right).
    \vspace{-2mm}
    $$
\end{corollary}

The Appendix includes a proof of this corollary.

\section{Frequency Oracle} \label{sec: frequency oracle}

In this section, to reduce running time and memory usage, we show a framework for incorporating sketching methods into \ldp frequency oracles that satisfy relevant conditions.
When combined with \hrr, this gives arise to a frequency oracle with a state-of-the-art theoretical guarantee. 

Suppose $\algoOracle{}$ is an $\eps{}$-\ldp frequency oracle with:
\vspace{-1mm}
\begin{itemize}
    \item \vspace{-2mm} Server running time:~$\tilde{O} \big( \Phi_{\ttime} ( |\mcalU{}|, d ) \big)$;
    \item \vspace{-2mm} Server memory usage:~$\tilde{O} \big( \Phi_{\mem} ( |\mcalU{}|, d ) \big)$;
    \item \vspace{-2mm} Utility guarantee as follows: for every~$\beta' \in (0, 1)$ and each~$v \in \mcalD{}$, with probability at least $1 - \beta'$, $\algoOracle{}$ returns an estimate $\hat f_\mcalU{} [ v ]$ satisfying
    $
        \left|  \hat f_\mcalU{} [ v ] - f_\mcalU{} [ v ] \right| \in O \left( ( {1} / {\eps} ) \cdot \sqrt{ | \mcalU{} | \cdot \ln ( {1} / {\beta'} )  } \right).
    $
     
\end{itemize}

\vspace{-3mm}
For example, when $\algoOracle{}$ is \hrr, then $\Phi_{\ttime} ( |\mcalU{}|, d )$ $= |\mcalU{}| + d$ and $\Phi_{\mem} ( |\mcalU{}|, d ) = d$.
Below we state the key result of this section.
\begin{thm}[Sketching Framework] \label{thm: sketching framework}
    For every~$\beta'$ $\in (0, 1)$, $\algoOracle{}$ can be converted into a new $\eps{}$-\ldp frequency oracle, which has server running time $\tilde{O} ( \Phi_{\ttime} ( |\mcalU{}|, \sqrt{ |\mcalU{}| } ) )$ and memory usage $\tilde{O} ( \Phi_{\mem} ( |\mcalU{}|, \sqrt{ |\mcalU{}| } ) )$.
    Fix some element~$v \in \mcalD{}$ to be given as a query to the new algorithm.
    With probability at least $1 - \beta'$, it returns an estimate $\hat f_\mcalU{} [ v ]$ satisfying
    \vspace{-1mm}
    $$
        \left|  \hat f_\mcalU{} [ v ] - f_\mcalU{} [ v ] \right| \in O \left( ( {1} / {\eps} ) \cdot \sqrt{ | \mcalU{} | \cdot \ln ( {1} / {\beta'} )  } \right).
    \vspace{-2mm}
    $$
\end{thm}

In particular, when $\algoOracle{}$ is \hrr, we obtain a new $\eps{}$-\ldp frequency oracle, which we call \oracle, with the following properties. 

\begin{corollary}[Algorithm \oracle] \label{thm: our oracle}
    Fix an element~$v \in \mcalD{}$ to be given as a query to \oracle. 
    With probability at least~$1 - \beta'$, \oracle returns a frequency estimate $\hat f_\mcalU{} [ v ]$ satisfying
    $$
    {\small
            \left|  \hat f_\mcalU{} [ v ] - f_\mcalU{} [ v ] \right| \in O \left( ( {1} / {\eps} ) \cdot \sqrt{ | \mcalU{} | \cdot \ln ( {1} / {\beta'} )  } \right).
    }
    $$
    Each user in $\mcalU{}$ requires~$\tilde O(1)$~memory, takes~$\tilde O(1)$ running time and reports only~$1$ bit to the server. 
    The server processes the reports in~$\tilde O( |\mcalU{} | )$ time and~$O( \sqrt{ | \mcalU{} | } )$ memory, and answers a query in~$\tilde O(1)$ time. 
    The~$\tilde O$ notation hides logarithmic factors in $| \mcalU{} |$, $| \mcalD{} |$ and $1 / \beta'$.
\end{corollary}

\subsection{A General Framework} \label{subsec: a general framework}

To reduce the size of the data domain, we now show how to incorporate \emph{sketching}  into \ldp frequency oracles.
Choosing parameters carefully, this leads to significant decrease of the server running time and memory usage without degrading estimation error.
The \emph{sketch} we apply is a variant of the Count-Median sketch~\citep{cormode2020small}, with the framework outlined in Algorithm~\ref{algo: hada-oracle-framework}.

\begin{algorithm}[!ht]
    \caption{Sketching Framework}
    \label{algo: hada-oracle-framework}
    {\it Construction}
    \begin{algorithmic}[1]
        \REQUIRE A set of users $\mcalU{}$; $\eps$-\ldp frequency oracle $\algoOracle{}$; element domain $\mcalD{}$;
        \STATE $k \leftarrow C_K \cdot \ln ( 4 / \beta')$, $m \leftarrow 8e^2 \cdot \sqrt{C_K} \cdot \eps \cdot \sqrt{| \mcalU{} |}$;
        \STATE Partition $\mcalU{}$ into $k$ subsets: $\mcalU{}_1, \ldots, \mcalU{}_k$. 
        \STATE Initialize $k$ pairwise independent hash functions $h_1, \ldots, h_k : \mcalD{} \rightarrow [m]$.
        \FOR{$i \in [k]$}
            \STATE The server broadcasts $h_i$ to all users in $\mcalU{}_i$.
            \STATE Each user $u \in \mcalU{}_i$ replaces their element, $v^{ (u) } \in \mcalD{}$, with a new one $h_i ( v^{ (u) } ) \in [m]$.
            \STATE The server runs an independent copy of~$\algoOracle{}$ on $\mcalU{}_i$, denoted as~$\algoOracle{}^{ (i) }$, for new elements $\{ h_i ( v^{(u) } ) : u \in \mcalU{}_i \}$. 
        \ENDFOR
    \end{algorithmic}
    
    \vspace{1mm}
    {\it Query}
    \begin{algorithmic}[1]
        \REQUIRE Element $v \in \mcalD{}$;  
        \FOR{$i \in [k]$}
            \STATE Query~$\algoOracle{}^{ (i) }$ for the frequency of~$h_i(v)$ over $\{ h_i ( v^{(u) } ) : u \in \mcalU{}_i \}$, denote the returned estimate as $\hat f_{\mcalU{}_i, h_i} [ h_i(v) ]$.
        \ENDFOR
        \STATE $\hat f_\mcalU{}[ v ] \leftarrow \textbf{Median} \left( k \cdot \hat f_{\mcalU{}_i, h_i} [ h_i( v ) ] : i \in [k]  \right)$\,.
        \RETURN $\hat f_\mcalU{}[ v ]$
    \end{algorithmic}
\end{algorithm}

{\bf Domain Reduction.}
By mapping the elements from domain~$\mcalD{}$ to~$[m]$, we want to reduce the domain size from~$d$ to some smaller~$m \in \N^+$; we discuss how to set~$m$ later.
The mapping could be performed via a pairwise independent hash function $h: \mcalD{} \rightarrow [m]$, such that: (1) for each~$v \in \mcalD{}$, it is mapped to $[m]$ uniformly at random; (2) for each pair of distinct~$v, v' \in \mcalD{}$, the probability that they are mapped to the same value is~$1 / m$.
The function~$h$ has succinct description of~$O(\log d)$ bits~\citep{MU17}.
Each user~$u \in \mcalU{}$ is then informed of~$h$ and replaces their data~$v^{ (u) } \in \mcalD{}$ with the new element~$h( v^{(u)} ) \in [m]$.

We then invoke \ldp frequency-oracle 
~$\algoOracle{}$ for estimating frequencies in the new dataset,~$\{ h( v^{(u)} ) : u \in \mcalU{} \}$.
For each~$v \in \mcalD{}$, to obtain an estimate of~$f_\mcalU{} [v]$,
we return the frequency estimate of~$h(v)$ in the new dataset, denoted by~$\hat f_{\mcalU{}, h} [ h(v) ]$,  provided by~$\algoOracle{}$.

{\bf Repetition.} 
For $v \in \mcalD{}$, by the triangle inequality,  estimation error $|\hat f_{\mcalU{}, h} [ h(v) ] - f_\mcalU{} [ v ]|$ is at most:
$$
    | \hat f_{\mcalU{}, h} [ h(v) ] - f_{\mcalU{}, h} [ h(v) ]| + | f_{\mcalU{}, h} [ h(v) ] - f_\mcalU{} [ v ]|\,,
$$
where~$f_{\mcalU{}, h} [ h(v) ] \doteq |\{ v \in \mcalU{} : h ( v^{ (u )} ) = h (v) \}|$ is the frequency of~$h(v)$ in the new dataset. 
The first term inherits from the estimation error of~$\algoOracle{}$;
the second term arises from hash collisions of~$h$. 
The assumption of~$\algoOracle{}$ is that the first term is, with probability~$1 - \beta'$, bounded by $O( (1 / \eps) \cdot \sqrt{ | \mcalU{} | \cdot \ln ( {1} / {\beta'} ) } )$. 

For the second term, we could set $m \in \tilde O( \sqrt{| \mcalU{} |}  / \beta')$. 
Via Markov's inequality, with probability $1 - \beta'$, it is is $\tilde O( \sqrt{| \mcalU{} |})$ . 
However, when $\beta'$ is small, e.g., $1 / |\mcalU{} |^c$, for some constant~$c$, this would be unacceptable.

Alternatively, to bound the second term, we could set $m \in \tilde O( \sqrt{| \mcalU{} |} )$, bounding with constant probability the second term by~$\tilde O( \sqrt{| \mcalU{} |})$.
We independently select~$k$
pairwise independent hash functions~$h_1, \ldots, h_k$.
To run the standard Count-Median Sketch, we would apply each $h_i, i \in [k]$ to each of the users, and for each~$i \in [k]$ invoke $\algoOracle{}$ independently to estimate the elements' frequencies over~$\{ h_i ( v^{ (u) } ) : u \in \mcalU{} \}$.
For $v \in \mcalD{}$, the median over~$i \in [k]$ of $\hat f_{\mcalU{}, h_i } [ h_i(v) ]$ is returned as its frequency estimate.

But this scheme requires each user to participate in~$k$ frequency oracles, which would degrade privacy according to the Composition Theorem~\citep[Theorem~3.14]{DR14}.
This motivates partitioning~$\mcalU{}$ into~$k$ subsets, denoted as $\mcalU{}_1, \ldots, \mcalU{}_k$.
For each $i \in [k]$, the users in subset $\mcalU{}_i$ map their elements with hash function~$h_i$, and an independent copy of $\algoOracle{}$ is applied to $\mcalU{}_i$, to estimate the frequencies of $\{ h_i ( v^{ (u) } ) : u \in \mcalU{}_i \}$.

Here we study two partitioning schemes.
\vspace{-2mm}
\begin{itemize}
    \item {\it Independent partitioning.} Here, each user is put into one of the subsets $\{ \mcalU{}_i : i \in [k] \}$ uniformly and independently at random.
    
    \item {\it Permutation partitioning.} Here we randomly permute~$\mcalU{}$: the first~$|\mcalU{}| / k$ users in the permutation become~$\mcalU{}_1$, the next~$|\mcalU{}| / k$ become~$\mcalU{}_2$, etc.
    \vspace{-2mm}
\end{itemize}

Compared to {\it independent partitioning}, {permutation partitioning} creates subsets of equal size. 
Since each user participates in only one copy of~$\algoOracle{}$, Algorithm~\ref{algo: hada-oracle-framework} is $\eps{}$-\ldp. 
It remains to prove the utility guarantee of Algorithm~\ref{algo: hada-oracle-framework}.

\subsection{Utility Analysis}

{\it Two Kinds of Frequencies.} 
First, consider the frequencies of elements before hashing in each subset.
For each~$i \in [k]$, and~$v \in \mcalD{}$, define $f_{\mcalU{}_i} [ v ] \doteq | \{ u \in \mcalU{}_i : v^{ (u) } = v \} |$ to be the frequency of $v$ in the set $\{ v^{ (u) } : u \in \mcalU{}_i \}$.
It is a random variable, whose randomness inherits from the partitioning. 
But for each partitioning scheme, $\E \big[ f_{\mcalU{}_i} [ v ] \big] = f_{\mcalU{} } [ v ] / k$.

Second, consider the frequencies of the hashed elements in each subset. 
For each~$i \in [k]$ and $w \in [m]$, let $f_{ \mcalU{}_i , h_i} [ w ] \doteq | \{ u \in \mcalU{}_i : h_i( v^{ (u) } ) = w \} |$ be the number of users in $\mcalU{}_i$ whose item are hashed to~$w$.
It is also a random variable, whose randomness arises from the partitioning, and from the hash function~$h_i$. 
It holds that $\E [ f_{ \mcalU{}_i , h_i} [ w ] ] = | \mcalU{}_i | / m$.

{\it Three Kinds of Errors.} For each~$i \in [k]$ and~$v \in \mcalD{}$, let $\hat f_{\mcalU{}_i, h_i} [ h_i(v) ]$ be the estimate of $f_{\mcalU{}_i, h_i} [ h_i (v) ]$ by~$\algoOracle{}$. 
We are interested in the deviation of~$k \cdot \hat f_{\mcalU{}_i, h_i} [ h_i(v) ]$ from $f_\mcalU{} [ v]$, which can be decomposed into three parts:

\begin{enumerate}
    \item $\lambda_1 (i, v) \doteq k f_\mcalU{}_i [ v ] -  f_\mcalU{} [ v ]$\,.
    \item $\lambda_2 (i, v) \doteq k f_{\mcalU{}_i, h_i} [ h_i(v) ] - k f_\mcalU{}_i [ v ] $\,.
    \item $\lambda_3 (i, v) \doteq k \hat f_{\mcalU{}_i, h_i} [ h_i (v) ] - k f_{\mcalU{}_i, h_i} [ h_i(v) ]$\,. 
\end{enumerate}

Define $\errOracle{} (i, v) \doteq k \hat f_{\mcalU{}_i, h_i} [ h_i (v) ] - f_{\mcalU{} } [ v ]$.
Its absolute value is bounded by a triangle inequality
{
    \begin{equation} \label{ineq: error decomposition}
        | \errOracle{} (i, v) | 
        \le 
        |\lambda_1(i, v)| + |\lambda_2(i, v)| + |\lambda_3(i, v)|\,.
    \end{equation}
}
{\it Four Kinds of Good Sets.} We are interested in the following four kinds of \emph{good sets}, that play important roles in bounding the estimation error of Algorithm~\ref{algo: hada-oracle-framework}.
First, define 
$$
    \goodset{}_0 \doteq \{ i \in [k] : k |\mcalU{}_i | \in \Theta( |\mcalU{} |  ) \}\,.
$$
\noindent
Then, for each $v \in \mcalD{}$, define
\vspace{-2mm}
{\small
    $$
        \goodsetA{} (v)
        \doteq \Big\{ i \in [k] : | \lambda_1(i, v) | \in O \Big( \sqrt{ |\mcalU{} | \ln \frac{1}{\beta'} } \Big) \Big\}\,.
    \vspace{-2mm}
    $$
}
Finally, for $j = 2, 3$, define 
\vspace{-2mm}
{\small
    $$
        \goodset{}_j (v) \doteq \Big\{ i \in [k] : | \lambda_j (i, v) | \in O \Big( \frac{1}{ \eps{} } \sqrt{ k  |\mcalU{}_i | \ln \frac{1}{\beta'} } \Big) \Big\}\,.
    \vspace{-2mm}
    $$
}

The following result is the key to the utility guarantee.

\begin{thm} \label{thm: size of the good sets}
    With probability at least~$1 - \beta' / 4$, it holds that $|\goodset{}_0| > (1 - 1 / 8) k$.
    And for each $v \in \mcalD{}$ and each $j \in [3]$, with probability at least $1 - \beta' / 4$, it holds  that $|\goodset{}_j (v) | > (1 - 1 / 8) k$. 
\end{thm}

Theorem~\ref{thm: size of the good sets} holds for both independent partitioning and permutation partitioning. 
The proof is technical, so we defer to the end of this subsection.
For now, we prove the utility guarantee of Algorithm~\ref{algo: hada-oracle-framework}. 

\begin{corollary} \label{corollary: median error guarantee}
    For each $v \in \mcalD{}$, with probability at least~$1 - \beta'$, the~$\hat f_\mcalU{} [v]$ returned by Algorithm~\ref{algo: hada-oracle-framework} satisfies 
    $$
        \left|  \hat f_\mcalU{} [ v ] - f_\mcalU{} [ v ] \right| \in O \left( (1 / \eps{} ) \cdot \sqrt{ | \mcalU{} | \cdot \ln (1 / \beta') } \right).
    $$
\end{corollary}
\begin{proof}[Proof of Corollary~\ref{corollary: median error guarantee}]
    By Theorem~\ref{thm: size of the good sets}, and a union bound, with probability at least~$1 - \beta'$, we have 
    \begin{equation}
        \Big| \big( \cap_{ j \in [3] } \goodset{}_j (v) \big) \cap \goodset{}_0 \Big| > k / 2\,.
        \label{eqn:median}
    \end{equation}
    For each $i \in \big( \cap_{ j \in [3] } \goodset{}_j (v) \big) \cap \goodset{}_0$, since $k |\mcalU{}_i| \in \Theta( |\mcalU{}| )$,  for $j = 2$ or $3$, it holds that 
    {\small
        $$ 
            | \lambda_j (i, v) | \in O \Big( \frac{1}{ \eps{} } \sqrt{ k  |\mcalU{}_i | \ln \frac{1}{\beta'} } \Big) \subset O \Big( \frac{1}{ \eps{} } \sqrt{ |\mcalU{} | \ln \frac{1}{\beta'} } \Big).
        $$
    }
    Therefore, via Inequality~(\ref{ineq: error decomposition}), it holds that 
    $$ \small
        | \errOracle{}(i, v) | \in O \Big( (1 / \eps{} ) \cdot \sqrt{ |\mcalU{} | \ln (1 / \beta') } \Big)\,.
    $$
    We finish by combining this error bound with $\hat f_\mcalU{}[ v ] = \med{}_{i \in [k] }\, k \cdot \hat f_{\mcalU{}_i, h_i} [ h( v ) ]$ and inequality~\eqref{eqn:median}. 
\end{proof}

\begin{proof}[Proof Outline for Theorem~\ref{thm: size of the good sets}]
    This is a sketch of a full proof that appears in the~Appendix. 
    As they are easier, we first bound the sizes of~$\goodsetB{} (v)$ and~$\goodsetC{} (v)$.

    {\it Bounding the Size of} $\goodsetB{}(v)$.
    Observe that for each $i \in [k]$, the (scaled) errors $|\lambda_2 (i, v)| / k = |  f_{\mcalU{}_i, h_i} [ h(v) ] - f_\mcalU{}_i [ v ] |$ result from hash collisions. 
    Since~$h_i$ is pair-wise independent, the expected size of collision is at most~$|\mcalU{}_i| / m$, i.e., 
    $
        \E \left[ \left| \lambda_2 (i, v) \right| / k \right] \le {| \mcalU{}_i | } / { m }. 
    $
    Recall that Algorithm~\ref{algo: hada-oracle-framework} initializes $k = C_K \cdot \ln ( 4 / \beta')$ and $m = 8e^2 \cdot \sqrt{C_K} \cdot \eps \cdot \sqrt{| \mcalU{} |}$ for some constant $C_K$.
    Via Markov's inequality, and that $|\mcalU{}| \ge |\mcalU{}_j|$, we have 
    $$
        \Pr \left[ \left| \lambda_2 (i, v) \right| \ge (1 / \eps{} ) \sqrt{ k |\mcalU{}_i | \ln ( 4 / \beta') } \right] \le 1 / (8e^2)\,. 
    $$
    Therefore, $\Pr[ i \neq \goodsetB{} (v) ]$ is upper bounded by $1 / (8e^2)$.
    As the $h_1,\ldots , h_k$ are chosen independently, the indicators of not being in $\goodsetB{} (v)$ are independent over~$i \in [k]$. 
    By a Chernoff bound, we can prove that, if $k = C_K \cdot \ln (4 / \beta')$ for some large enough constant~$C_K$, then %
    $$
        \Pr[ |\{ i \in [k] : i \neq \goodsetB{} (v)\} | \ge k / 8 ] \le \beta' / 4\,.
    $$   

    {\it Bounding the Size of} $\goodset{}_3 (v)$.
    Via the assumption of~$\algoOracle{}$, for $i \in [k]$, with probability at most $1 / ( {8e^2} )$,
    {\small
    $$
        | \hat f_{\mcalU{}_i, h_i} [ h_i (v) ] - f_{\mcalU{}_i, h_i} [ h_i (v) ] | \notin O \left( \frac{1}{\eps} \sqrt{ | \mcalU{}_i |  \cdot \ln (8e^2) } \right)\,.
    $$
    }
    Scaling both sides by a factor of~$k$, we get $|\lambda_3 (i, v) | \notin O ( ( {1} / {\eps} )  \sqrt{  k^2 | \mcalU{}_j | } )$.
    Replacing one factor~$k$ with $C_K \cdot \ln ({4} / {\beta'} )$, we have 
    $
        |\lambda_3 (i, v) | \notin O ( ( {1} / {\eps} ) \sqrt{ k  | \mcalU{}_j | \cdot \ln ( {1} / {\beta'} ) } ).
    $
    Since the indicators of being in $\goodsetC{} (v)$ are independent, there exists some constant~$C_K$, s.t., 
    $$
        \Pr[ |\{ i \in [k] : i \neq \goodsetC{} (v)\} | \ge k / 8 ] \le \beta' / 4\,.
    $$   

    {\it Bounding the Size of} $\goodset{}_0$. 
    For permutation partition, it holds that $|\mcalU{}_i| = |\mcalU{}| / k$, $\forall i \in [k]$.
    Therefore, $|\goodset{}_0 (v)| = k$.
    For independent partitioning, analyzing $\goodset{}_0$ is not trivial, as the $|\mcalU{}_i|$ are not independent. 
    Therefore, we consider the deviations of the subset sizes as a whole: define 
    $\Delta_0 \doteq \sum_{i \in [k] } \big\vert | \mcalU{}_i | - | \mcalU{} | / k \big\vert$,
    which measures the distance between vector $( |\mcalU{}_1 |, \ldots, |\mcalU{}_k |) \in \R^{ k }$ and its expectation.
    Via the McDiarmid inequality~\citep{MU17}, we prove that, there exists some constant $C_0$, s.t., for every choice of positive integer~$k$, with probability at least~$1 - \beta' / 4$: $\Delta_0 \le C_0 \sqrt{ | \mcalU{} | \ln ( {4} / {\beta'} ) }$. 
    
    It follows that 
    $\small{
            \sum_{i \in [k] } \big\vert k | \mcalU{}_i | - | \mcalU{} | \big\vert \le k C_0 \sqrt{ | \mcalU{} | \ln ( {4} / {\beta'} ) }.
        }
    $
    By a Markov inequality-like argument, the number of $i \in [k]$, such that $\small \big\vert k | \mcalU{}_i | - | \mcalU{} | \big\vert \ge 8 C_0 \sqrt{ | \mcalU{} | \ln ( 4 / \beta' ) }$ is bounded by~$(1 / 8)k$, which finishes the proof. 
    
    {\it Bounding the Size of} $\goodset{}_1 (v)$. 
    Similarly, we consider the deviations as a whole, and define $
            \Delta_1 \doteq \sum_{i \in [k] } \Vert  f_\mcalU{}_i - f_\mcalU{} / k \Vert_2$,
    where $\left\Vert f_\mcalU{}_i - { f_\mcalU{} } / k \right\Vert_2 \doteq \sqrt{ \sum_{ v' \in \mcalD{} } ( f_\mcalU{}_i [ v' ] -   f_\mcalU{} [ v' ] / k )^2 }$.
    Note that $\forall i \in [k]$, $|\lambda_1(i, v)| \le k \left\Vert  f_\mcalU{}_i - { f_\mcalU{} } / k \right\Vert_2$.
    
    Via the martingale concentration inequalities (the McDiarmid inequality and the Azuma–Hoeffding inequality~\citep{MU17, FL06}), we prove for both independent partitioning and permutation, that there exists some constant~$C_1$, s.t., for every choice of positive integer~$k$, with probability at least~$1 - \beta' / 4$: $\Delta_1 \le C_1 \sqrt{ | \mcalU{} | \ln ( 4 / \beta' ) }$. 
    
    Hence,  
    $k \sum_{i \in [k] } \Vert  f_\mcalU{}_i - f_\mcalU{} / k \Vert_2 \le k C_1 \sqrt{ | \mcalU{} | \ln ( 4 / \beta' ) }$. 
    By a counting argument, the number of~$i \in [k]$, such that $k \left\Vert  f_\mcalU{}_i - { f_\mcalU{} } / k \right\Vert_2 \ge 8 C_1 \sqrt{ | \mcalU{} | \ln ( 4 / \beta' ) }$ is bounded by~$(1 / 8)k$, which finishes the proof.
\end{proof}

\subsection{Comparison With Previous Approaches}

The seminal work of~\citet{BNST17} was the first to provide rigorous analysis for $\eps{}$-\ldp frequency oracles with sketching methods.
We differ from their approach thus: 
(1) They apply Count-Sketch instead of Count-Median Sketch.
(2) Similar to our bounding~$|\goodset{}_0|$, they invoke a technique called {\it Poisson Approximation}~\citep{MU17}, which approximates the distribution of~$|\mcalU{}_i|, i \in [k]$ by a set of~$k$ independent Poisson random variables. 
This results in a setting of~$k \in \Theta ( \ln (|\mcalU{}| / \beta' ) )$, and subsequently a sub-optimal utility guarantee of $O \big( ( {1} / {\eps} ) \cdot \sqrt{ | \mcalU{} | \cdot \ln ( { | \mcalU{} | } / {\beta'} ) } \big)$.
(3) Finally, though invoking \hrr as a sub-routine, their work does not exploit the fast Hadamard transform.
Hence their oracle answers a frequency query in $\tilde O( \sqrt{ |\mcalU{}| })$ time, instead of~$\tilde O(1)$.

The recent experimental study by
\cite{cormode2021frequency} 
provides inspiring insights.
They propose to view existing algorithms as different combinations of {\it sketching methods} (for domain reduction) and frequency oracle algorithms (constructed over the reduced domain).  
Their work empirically evaluates performance of the sketches based on different values of~$k$ and~$m$, without corresponding theoretical analysis.

Lastly, we discuss briefly the Count-Median Sketch applied by our algorithm. 
Compared to the standard Count-Median Sketch, which applies each of the~$k$ hash functions to the data of all users, our version partitions the users into~$k$ subsets, and applies each hash function only to a particular subset.
The primary reason for doing so in our work is to protect the privacy of the users.
But this technique can also be applied to applications where there is no requirement for privacy protection.
It reduces the time for processing a user's report from~$O(k)$ to~$O(1)$.
Based on our technique for analysing the size of~$\goodset{}_1$, this introduces an additional error of~$O( \sqrt{|\mcalU{}| \ln (1 / \beta') } )$ to an element's frequency estimate. 
In cases where the allowed error is~$\Omega( \sqrt{|\mcalU{}| \ln (1 / \beta') } )$, this is practical. 

\section{Succinct Histogram} \label{sec: succinct histogram}

We show how to construct a succinct histogram efficiently, based on the \oracle discussed in previous section. 
Our new algorithm, \ours improves the \algoTreeHist{} algorithm by~\cite{BNST17}. 

\begin{thm} \label{thm: hadaheavy}
    Let $\mcalU{}$ be a set of $n$ users each holding an element from some finite domain $\mcalD{}$ of size $d$, $\eps \in (0, 1)$ be the privacy parameter and $\beta \in (0, 1)$ be the specified failure probability. \ours is an $\eps$-\ldp algorithm, based on hierarchical search method, that returns an \textbf{S-Hist} set of (element, estimate) pairs of size~$\tilde{O}( \sqrt{n} )$, %
    where 
    $$
        \lambda \in O ( (1 / \eps ) \cdot  \sqrt{ n \cdot (\ln d ) \cdot ( 1 + ( \ln (1 / \beta)  \ / \ln n ) ) } )\,.
    $$
    Each user in $\mcalU{}$ requires~$\tilde O(1)$~memory, takes~$\tilde O(1)$ running time and reports only~$1$ bit to the server. 
    The server processes the reports in~$\tilde O( n )$ time and~$\tilde{O}( \sqrt{ n } )$ memory. 
    The~$\tilde O$ notation hides logarithmic factors in~$n$, $d$ and~$1 / \beta$.
\end{thm}

\subsection{\ours} \label{subsec: our succinct histogram}

\ours represents element in~$\mcalD{}$ with an alphabet of size~$\sqrt{n}$.

\textbf{Base-$\sqrt{n}$ representation.} Let $\Lambda \doteq \{0, 1, \ldots, \sqrt{n} - 1 \}$ be an alphabet of size $\sqrt{n}$ (for simplicity, we assume that $\sqrt{n}$ is an integer), and $L \doteq 2 \cdot (\log d) / \log n$. 
Each element in $\mcalD{}$ can be encoded as a unique string in $\Lambda^L = \{0, 1, \ldots, \sqrt{n} - 1 \}^L$.  

\textbf{Prefix.} For each~$v \in \mcalD{}, \tau \in [L]$, let $v[1 : \tau]$ be the first~$\tau$ characters in~$v$'s base-$\sqrt{n}$ representation, which is called a \emph{prefix} of $v$. 
Let $\Lambda^0 \doteq \{ \bot \}$ be the set consisting of the empty string. 
For each $\tau \in [L]$, $\Lambda^\tau$ the set of all possible strings of length $\tau$. 
Further, for each string~$\pmbs{} \in \Lambda^\tau$, define the frequency of~$\pmbs{}$ to be~$f_\mcalU{} [ \pmbs{} ] \doteq | \{ u \in \mcalU{} : v^{ (u) }[1 : \tau] = \pmbs{} \} |$.

\textbf{Child Set.} For each $0 \le \tau < L$, and each string $\pmbs{} \in \Lambda^\tau$, the child set of~$\pmbs{}$, denoted as~$\pmbs{} \times \Lambda$, is defined as $\pmbs{} \times \Lambda  \doteq \{ \pmbs{} \circ \pmbt{} : \pmbt{} \in \Lambda \} \subset \Lambda^{\tau + 1}$, where $\pmbs{} \circ \pmbt{}$ is the concatenation of the strings $\pmbs{}$ and $\pmbt{}$. 

The key motivation for the {\it hierarchical searching method} is that, if an element~$v \in \mcalD{}$ is frequent, so is each of its prefixes. 

\textbf{Overview of \ours}. 
The goal of the algorithm is to identify a set of elements in~$\mcalD{}$, called heavy hitters, whose frequencies are no less than some threshold,~$\lambda$ (to be determined later). Clearly, for each~$\tau \in [L]$,
$$
    f_\mcalU{} [ v[1:\tau] ] \ge f_\mcalU{} [v] \ge \lambda\,.
$$
Assuming that we know the exact values of frequencies of the strings.
We can search for heavy hitters as follows. 
First, we initialize a sequence of empty sets $\mcalP{}_1, \ldots, \mcalP{}_L$, which we call {\it search sets}.
Then, we examine all strings in $\Lambda$.
If $\pmbs{} \in \Lambda$ has frequency~$f_\mcalU{} [ \pmbs{}] \ge \lambda$, then we put it into $\mcalP{}_1$.
Next, for each string~$\pmbs{} \in \mcalP{}_1$, we check each string $\pmbs'$ from its child set~$\pmbs{} \times \Lambda$.
If $\pmbs{}'$ has frequency~$f_\mcalU{} [ \pmbs{}' ] \ge \lambda$, then we put it into $\mcalP{}_2$.
In general, for $\tau \ge 2$, we can construct $\mcalP{}_\tau$ after $\mcalP{}_{\tau - 1}$ is constructed.
Finally~$\mcalP{}_L$ should contain all heavy hitters in $\mcalD{}$.

{\bf Frequency Oracles.}
As the exact values of the frequencies of the strings are not available, we want to learn their estimates. 
For each~$\tau \in [L]$, we construct a frequency oracle (\oracle) to estimate the frequencies of the strings in~$\Lambda^\tau$. 
We want to avoid each user participating in every frequency oracle: a user reporting~$L$ times to the server would degenerate the privacy guarantee of the algorithm. 
Therefore, we partition the set of users $\mcalU{}$ into~$L$ subsets~$\U_1, \U_2, \ldots, \U_L$. 
As before, this is performed by either {\emph{independent partitioning}} or {\emph{permutation partitioning}}, introduced in Section~\ref{subsec: a general framework}.
As $L \in O(\log d)$, by Corollary~\ref{thm: our oracle}, the server uses in total~$\tilde O( n )$ processing time and~$\tilde O(\sqrt{ n } )$ processing memory to construct these oracles.  

For each $\tau \in [L]$ and each $\pmbs{} \in \Lambda^\tau$, let~$f_{\U_\tau} [ \pmbs{} ] \doteq | \{ u \in \U_\tau : v^{ (u) }[1 : \tau] = \pmbs{} \} |$ be the frequency of~$\pmbs{}$ in $\U_\tau$, and $\hat f_{\U_\tau} [ \pmbs{} ]$ be its estimate by \oracle. 
As $\E [ L \cdot f_{\U_\tau} [\pmbs{}] ] = f_\mcalU{} [ \pmbs{} ]$, we use~$\hat f_\mcalU{} [ \pmbs{} ] \doteq L \cdot \hat f_{\U_\tau} [ \pmbs{} ]$ as an estimate of $f_\mcalU{} [ \pmbs{} ]$: its estimation error is established by the following theorem, proven in the Appendix.

\begin{thm} \label{thm: hadaheavy estiamtion error of single element}
    For each $\tau \in [L]$, fix some query string~$\pmbs{} \in \Lambda^\tau$ for the frequency estimate. 
    It holds that, with probability~$1 - \beta'$,
    $$
        | \hat f_\mcalU{} [ \pmbs{} ] - f_\mcalU{} [ \pmbs{} ] | \in O( ( 1 / \eps) \sqrt{ n \cdot (\log d) \cdot (\ln (1 / \beta') ) / \ln n } )\,.
    $$
    
\end{thm}
There are two sources of error. First, for each $\tau \in [L]$, the frequency distribution in $\U_\tau$ deviates from its expectation. We bound the~$\ell_2$ distance between the distribution in $\U_\tau$ and its expectation by martingale methods. The second kind of error inherits from \oracle.

{\bf Modified Search Strategy.} Only having access to estimates of the frequencies of the prefixes, we need to modify the criterion for adding elements to $\mcalP{}_\tau$. 

Let $\lambda' \doteq ( C_{\lambda} / \eps) \sqrt{ n \cdot (\log d) \cdot (\ln (n / \beta) ) / \ln n }$, where~$C_\lambda$ is some constant we will determine later.
Define $\mcalP{}_0 \doteq \{ \bot \}$. 
The $\mcalP{}_\tau$ are iteratively constructed according to the following criterion:
$$
    \mcalP{}_\tau \leftarrow \{ \pmbs{} \in \mcalP{}_{\tau - 1} \times \Lambda : \hat f_\mcalU{} [ \pmbs{} ] \ge 2 \lambda' \},\, \forall \tau \in [L]\,,
$$
where 
$
    \mcalP{}_{\tau - 1} \times \Lambda \doteq \{ \pmbs{} = \pmbs{}_1 \circ \pmbs{}_2 : \pmbs{}_1 \in \mcalP{}_{\tau - 1}, \pmbs{}_2 \in \Lambda \}.
$
Finally, after $\mcalP{}_L$ is constructed, its elements are returned as heavy hitters.

\subsection{Analysis}
Since each user participates in only one frequency oracle, the algorithm is $\eps{}$-\ldp.   
It remains to analyze the utility guarantee, running time and memory usage of the {\it modified search strategy}.

\begin{thm} \label{thm: elements in P_tau} 
    Let $\lambda \doteq 3 \cdot \lambda'$.
    With probability at least~$1 - \beta$, it is guaranteed that for each $\tau \in [L]$, and each $ \pmb{s} \in \Lambda^\tau$:\, (1) if $f_\mcalU{} [ \pmb{s} ] \ge \lambda$, then $\pmbs{} \in \mcalP{}_\tau$;\, (2) and for each $\pmbs{} \in \mcalP{}_\tau$, the frequency estimate~$\hat f_\mcalU{} [ \pmb{s} ]$ satisfies $|\hat f_\mcalU{} [ \pmb{s} ] - f_\mcalU{} [ \pmb{s} ] | \le \lambda'$.
    Constructing the $\mcalP{}_\tau$ for all $\tau \in [L]$ has~$\tilde{O}(n)$ running time and~$\tilde{O}(\sqrt{n})$ memory usage.
\end{thm}
We sketch the proof here, details are in the Appendix.

\begin{proof}[Proof Outline for Theorem~\ref{thm: elements in P_tau}]
    We focus on the estimation errors of prefixes from a fixed set. 
    
    \begin{definition}[Candidate Set] \label{def: candidate set}
    Define $\Gamma_0 \doteq \{ \bot \}$ to be the set of the empty string, and for $\tau \in [L]$, $\Gamma_\tau \doteq \{ \pmbs{} \in \Lambda^\tau : f_\mcalU{} [ \pmbs{} ] \ge \lambda' \}$, the set of prefixes of length~$\tau$ whose frequency is at least~$\lambda'$. 
    For $\tau < L$, the child set of~$\Gamma_\tau$ is defined as $ \Gamma{}_\tau \times \Lambda  \doteq \{ \pmbs{} = \pmbs{}_1 \circ \pmbs{}_2 : \pmbs{}_1 \in \Gamma{}_\tau, \pmbs{}_2 \in \Lambda \}$, where $\pmbs{}_1 \circ \pmbs{}_2$ is the concatenation of $\pmbs{}_1$ and $\pmbs{}_2$. 
    The \emph{candidate} set is defined as $\Gamma \doteq \cup_{0 \le \tau < L } \left( \Gamma{}_\tau \times \Lambda \right)$. 
    \end{definition}
    
    Note that for each $\tau \in [L]$, we have $|\Gamma_\tau | \le n / \lambda' \le \sqrt{n}$. 
    Hence, $|\Gamma| = \sum_{0 \le \tau < L } \left| \Gamma{}_\tau \times \Lambda \right| \le L\sqrt{n} \cdot \sqrt{n} \in \tilde O (n)$.  %
    By applying Theorem~\ref{thm: hadaheavy estiamtion error of single element} with $\beta' = \beta / (nL)$ and  the union bound over all~$\pmbs{} \in \Gamma$, we have:
    
    \vspace{-1mm}
    \begin{corollary} \label{cor: error of candidate set}
        There exists some constant $C_\lambda$, such that with probability at least $1 - \beta$, it holds that
        {\small
            $
                \max_{ \pmbs{} \in \Gamma } | \hat f_\mcalU{} [ \pmbs{} ] - f_\mcalU{} [ \pmbs{} ] | 
                \le \lambda'
            $
        }, where 
        \vspace{-2mm}
        $$
            \lambda' = ( C_{\lambda} / \eps) \sqrt{ n \cdot (\log d) \cdot (\ln (n / \beta) ) / \ln n }\,.
        $$
        
    \end{corollary}
    \vspace{-3mm}
    Conditioned all strings in~$\Gamma$ having estimation error~$\lambda'$, we can prove by induction that the {\it modified search strategy} only inspects the frequencies of the strings from~$\Gamma$, in order to construct the $\mcalP{}_\tau, \tau \in [L]$.
    Therefore, the strings added to $\mcalP{}_\tau, \tau \in [L]$ have estimation errors bounded by~$\lambda'$.
    Moreover, for each $\tau \in [L]$, and each $ \pmb{s} \in \Lambda^\tau$, if $f_\mcalU{} [ \pmb{s} ] \ge \lambda = 3 \lambda'$, then $\pmbs{} \in \Gamma$, and we can prove that $\pmbs{}$ will be added to~$\mcalP{}_\tau$.
    The details are included in the of the complete proof in Appendix.
    
    To analyze the running time and memory usage, observe that the {\it modified search strategy} invokes~$L$ frequency oracles.
    By Theorem~\ref{thm: our oracle}, they have total construction time~$\tilde{O}(n)$ and memory usage~$\tilde{O}( \sqrt{n} )$.
    Since each frequency query takes~$\tilde O(1)$ time, and all strings queried belong to~$\Gamma$, the total query time is bounded by $|\Gamma | \in \tilde O(n)$, which finishes the proof.
\end{proof}

\subsection{Comparison With Previous Approaches}

Among the previous algorithms that identify heavy hitters based on hierarchical search,
\algoTreeHist{}~\citep{BNST17} provides the best known error guarantee of~$O ( (1 / \eps ) \cdot  \sqrt{ n \cdot ( \ln d ) \cdot \ln (n / \beta)  } )$.
Our algorithm reduces this error to~$O ( (1 / \eps ) \cdot  \sqrt{ n \cdot (\ln d ) \cdot ( 1 + ( \ln (1 / \beta)  \ / \ln n ) ) } )$.
There are two major differences between our algorithm and \algoTreeHist{}. 
First, the algorithm \algoTreeHist{} considers base-$2$ representation of elements in $\mcalD{}$, instead of base-$\sqrt{n}$ representation. 
Each element in~$\mcalD{}$ is encoded as a binary string of length~$\log d$.
This requires the algorithm to partition the user set $\mcalU{}$ into $\log d$ subsets, which results in smaller subset sizes and larger estimation error than our algorithm. 
Second, the frequency oracle used by \algoTreeHist{} does not exploit the fast Hadamard transform. 
Its frequency oracle answers a frequency query in $\tilde O( \sqrt{n} )$ time. 
In comparison, \oracle answers a query in~$\tilde{O} (1)$ time.

Recent works~\citep{WLJ21, cormode2021frequency} observe that, instead of identifying prefixes of the heavy hitters with increasing lengths, one character at a time, we can identify such prefixes by several characters at each step.
This reduces the number of steps required to reach the full length strings. 
Indeed, using a large alphabet to represent the elements in~$\mcalD{}$ (e.g., an alphabet of size~$\sqrt{n}$, as we proposed) achieves the same effect.
This strategy is effective empirically~\citep{WLJ21, cormode2021frequency}, but there are no theoretical guarantees. 
We believe our work improves understanding of these high-quality experimental results.

\section{Related Work} \label{sec: related work}

\vspace{-1mm}
{\bf Frequency Oracle.}
We briefly document the development of frequency oracles in recent years.
\cite{BS15}
described a frequency oracle that achieves error
$O( (1 / \eps) \cdot \sqrt{ n \log ( 1 / {\beta} ) } )$.
However, it needs $\tilde O(n^2)$ random bits to describe a random matrix,  and answers a query in $O(n)$ time. 
In \citeyear{BNST17}, a similar version with simplified analysis, called \textbf{ExplicitHist}, was studied by 
~\cite{BNST17}. 
\textbf{ExplicitHist} achieves the same estimation error, but requires only $\tilde O(1)$ random bits to describe the random matrix. 
It answers a query in $O(n)$ time. 
Other optimizations have been proposed.
First documented in \citep{NXYHSS16}, and widely used in \ldp literature \citep{BNST17, D.P.Apple17, CKS19}, the \hrr algorithm uses the Hadamard matrix to replace the random matrix without increasing the estimation error.
The matrix does not need to be generated explicitly and each of its entries can be computed in $\tilde O(1)$ time when needed. 
The \hrr answers a query in $O( \min \{n, d \} )$ time, or with pre-processing time~$\tilde O(d)$, answers each query in $\tilde O(1)$ time. 
In \citeyear{AcharyaSZ19}, \citeauthor{AcharyaSZ19} proposed a variant of \hrr which does not rely on public randomness.
The protocol shares similar performance guarantees to the original, but can be modified to support frequency estimation in low privacy regime ($\varepsilon > 1$)~\citep{GhaziG0PV21}.
Finally, \cite{BNST17} also applied Count-Sketch~\citep{CCF02} to reduce the domain size of the elements. Their frequency oracles, \freqOracle and \hashtogram, have server running time $\tilde O(n)$ and memory usage $\tilde O( \sqrt{n} )$. 
But these algorithms have sub-optimal estimation error of $O( (1 / \eps) \cdot \sqrt{ n \log ( {n} / {\beta} ) } )$.

{\bf Succinct Histogram.}
For the succinct histogram problem, \citet{BS15} proposed the first polynomial-time algorithm that has worst-case error $O( (\log^{1.5} ( {1} / {\beta} ) ) \cdot (1 / \eps ) \cdot  \sqrt{ n \log d})$. 
However, it has server running time $\tilde O( n^{2.5} )$ and user time $\tilde O(n^{1.5} )$, which is not practical. 
\cite{BNST17}
proposed two improved algorithms, \algoTreeHist{} and \algoBitstogram{}, which
involve different techniques. \algoTreeHist{} searches for the heavy hitters via a prefix tree;
 \algoBitstogram{} hashes elements into a smaller domain and identifies a noisy version of the heavy hitters.
The recovery of the true heavy hitters relies on error-correcting codes. 
The former algorithm has error $O ( (1 / \eps ) \cdot  \sqrt{ n \cdot ( \log d ) \cdot \log (n / \beta)  } )$, while the latter has error $O ( (1 / \eps ) \cdot  \sqrt{ n \cdot ( \log (d / \beta) ) \cdot \log (1 / \beta)  } )$; each achieves almost-optimal error, but \algoTreeHist{} is inferior to \algoBitstogram{} by a factor of $\sqrt{ \log n }$. Importantly, each algorithm has server time~$\tilde O(n)$ and user time~$\tilde O(1)$.

Due to the sophistication of error-correcting codes,
of the two algorithms presented in~\citep{BNST17}, only \algoTreeHist{} was implemented and experimented.
\cite{BNS19}
further refined \algoBitstogram{} based on the list-recoverable code, which involves identifying spectral clusters in a derived graph. 
Their new algorithm, \algoPrivateExpanderSketch{}~\citep{BNS19}, achieved an optimal error of $O( (1 / \eps ) \cdot  \sqrt{ n \cdot \log (d / \beta) } )$.
Again, this style of algorithm has yet to be implemented.

We observe finally that {\it Sketching methods} and {\it hierarchical searching methods} are not only applied to the \ldp model, but also to other models of DP for frequency estimation, for example the shuffle model~\citep{LWY21, BalleBGN19, GhaziG0PV21}.

\subsubsection*{Acknowledgements}
This research is supported by an Australian Government Research Training Program (RTP) Scholarship.

\bibliographystyle{IEEEtranN}
\bibliography{ref}

\input{supplement}

\end{document}

%% file: supplement.tex
\onecolumn
\aistatstitle{Supplementary Material} 

\begin{samepage}
    The supplementary material is organized as follows:
    \begin{enumerate}
        \item In Section~\ref{sec: concentration inequalities}, we list the concentration inequalities applied in our proofs. 
        \item In Section~\ref{appendix: proof for sec: preliminaries}, we provide the detailed proofs for Section~\ref{sec: preliminaries}.
        \item In Section~\ref{appendix: proof for sec: frequency oracle}, we provide the detailed proofs for Section~\ref{sec: frequency oracle}.
        \item In Section~\ref{appendix: proof for sec: succinct histogram}, we provide the detailed proofs for Section~\ref{sec: succinct histogram}. 
    \end{enumerate}

    \section{Concentration Inequalities} \label{sec: concentration inequalities}
    
    \begin{fact}[Chernoff bound \citep{MU17}] \label{fact: prototype chernoff bound} Let $X_1, \ldots, X_n$ be independent 0-1 random variables. Let $X = \sum_{i \in [n]} X_i$ and $\mu = \E[X]$. Then for every $\delta > 0$,
    $$
        \Pr[ X \ge (1 + \delta) \mu ] \le  \left( \frac{ e^\delta}{ (1 + \delta)^{1 + \delta} } \right)^\mu\,.
    $$
    Similarly, for every $\delta \in (0, 1)$
    $$
        \Pr[ X \le (1 - \delta) \mu ] \le  \left( \frac{ e^{-\delta} }{ (1 - \delta)^{1 - \delta} } \right)^\mu\,.
    $$
    \end{fact}

    \begin{fact}[Bernstein’s Inequality \citep{AMS09}] \label{fact: bernstein} Let $X_1, \ldots, X_n$ be independent real-valued random variables such that $|X_i| \le c$ with probability one. Let $S_n = \sum_{i \in [n]} X_i$ and $\Var{} [ S_n ]  = \sum_{ i \in [n] } \Var [X_i^2]$. Then for all $\beta \in (0, 1)$, 
    $$
        \left| S_n - \E[ S_n ] \right| \le \sqrt{ 2 \Var{} [ S_n ] \ln \frac{2}{\beta} } + \frac{2 c \ln \frac{2}{\beta}}{ 3}\,,
    $$
    with probability at least $1 - \beta$. 
    \end{fact}
    
    \begin{fact}[Hoeffding's Inequality \citep{DL01}] \label{fact: hoeffding} Let $X_1, \ldots, X_n$ be independent real-valued random variables such that that $|X_i| \in [a_i, b_i], \forall i \in [n]$ with probability one. Let $S_n = \sum_{i \in [n]} X_i$, then for every~$\eta \ge 0$:
    \begin{align*}
        \Pr[ S_n - \E[ S_n] \ge \eta ] \le \exp \left( - \frac{2 \eta^2 }{ \sum_{i \in [n] } (b_i - a_i)^2 } \right)\,, \text{ and} \\
        \Pr[ \E [ S_n ] - S_n \ge \eta ] \le \exp \left( - \frac{2 \eta^2 }{ \sum_{i \in [n] } (b_i - a_i)^2 } \right)\,.
    \end{align*}
    \end{fact}
    
    \begin{definition}[Martingale \citep{MR95, MU17}] \label{def: martingale}
        A sequence of random variables $Y_0, \ldots, Y_n$ is a martingale with respect to the sequence $X_0, \ldots, X_n$ if, for all $i \ge 0$, the following conditions hold: i) $Y_i$ is a function of $X_0, \ldots, X_i$; ii) $\E[ |Y_i| ] < \infty$; and iii) $\E[ Y_{i + 1} \mid X_0, \ldots, X_i] = Y_i.$
    \end{definition}

    \begin{fact}[Azuma's Inequality \citep{MU17}] \label{fact:Azuma} Let $Y_0, \ldots, Y_n$ be a martingale such that 
    $$
        A_i \le Y_i - Y_{i - 1}  \le A_i + c_i\,, 
    $$
    for some constants~$\{c_i\}$ and for some random variables~$\{A_i\}$ that may be functions of~$Y_0, Y_1, \ldots Y_{i - 1}$. Then for all $t \ge 0$ and every $\eta > 0$, 
    $$
        \Pr[ |Y_t - Y_0 | \ge \eta ] \le 2 \exp \left( - \frac{ 2 \eta^2 }{   \sum_{i \in [t] } c_i^2 }\right)\,.
    $$
    \end{fact}
    
    \begin{fact}[\citep{FL06}] \label{fact: martingale-bernstein}
        Let the sequence of random variables $Y_0, \ldots, Y_n$ be a martingale with respect to the sequence of random variables $X_0, \ldots, X_n$ such that 
        \begin{enumerate}
            \item $\Var{}[ Y_i \mid X_0, \ldots, X_{i - 1} ] \le \sigma_i^2, \forall i \in [n]$\,; \text{ and}
            \item $| Y_i - Y_{i - 1} | \le c, \forall i \in [n]$\,.
        \end{enumerate}
        Then, we have 
        $$
            \Pr[ Y_n - Y_0 \ge \eta ] \le \exp \left( - \frac{\eta^2}{ 2 \left( \sum_{i \in [n] } \sigma_i^2 + c \eta / 3 \right) } \right)\,.
        $$
    \end{fact}

    \begin{definition}[Lipschitz Condition] \label{def:Lipschitz-Condition} A function $\Delta: \R^n \rightarrow \R$ satisfies the Lipschitz condition with bound $c \in \R$ if, for every $i \in [n]$ and for every sequence of values $x_1, \ldots, x_n \in \R$ and $y_i \in \R$, 
    $$
        | \Delta( x_1, x_2, \ldots, x_{i - 1}, x_i, x_{i + 1}, \ldots, x_n ) - \Delta( x_1, x_2, \ldots, x_{i - 1}, y_i, x_{i + 1}, \ldots, x_n ) | \le c.
    $$
    \end{definition}
    
    \begin{fact}[McDiarmid’s Inequality \citep{MU17, TIO17}] \label{fact:McDiarmid-Inequality}
    Let $\Delta: \R^n \rightarrow \R$ be a function that satisfies the Lipschitz condition with bound~$c \in \R$. Let~$X_1, \ldots, X_n$ be independent random variables such that $\Delta(X_1, \ldots, X_n)$ is in the domain of~$\Delta$. Then for all $\eta \ge 0$,
    $$
        \Pr \left[ \Delta(X_1, \ldots, X_n) - \E \left[ \Delta(X_1, \ldots, X_n) \right] \ge \eta \right] \le \exp \left( - \frac{ 2 \eta^2 } { n c^2 } \right)\,.   
    $$
    \end{fact}
    
    \begin{definition}[\citep{TIO17}] \label{def: symmetric function}
        Let $\mbbS{}_n$ be the symmetric group of $[n]$ (i.e., the set of all possible permutations of $[n]$). A function $\Delta : \mbbS{}_n \rightarrow \R$, is called $(n_1, n_2)$-symmetric with respect to permutations if, for each permutation $\pmbx{} \in \mbbS{}_n$, $\Delta( \pmbx{} )$ does not change its value under the change of order of the first~$n_1$ and/or last~$n_2 = n - n_1$ coordinates of $\pmbx{}$. For brevity, we call these functions $(n_1, n_2)$-symmetric functions.
    \end{definition}
    
    \begin{fact}[McDiarmid’s Inequality with respect to permutations~\citep{TIO17}] \label{fact: McDiarmid Inequality Permutation}
    Let $\Delta: \mbbS{}_n \rightarrow \R$ be an $(n_1, n_2)$-symmetric function for which there exists a constant $c > 0$ such that $| \Delta( \pmbx{} ) - \Delta( \pmbx{}_{i, j} ) | \le c$ for all $\pmbx{} \in \mbbS{}_n, i \in \{1, \ldots, n_1 \}, j \in \{ n_1 + 1, \ldots, n \}$, where the permutation~$\pmbx{}_{i,j}$ is obtained from~$\pmbx{}$ by transposition of its~$i^{\text{th}}$ and~$j^{\text{th}}$ coordinates. Let~$\mbfX{}$ be a vector of random permutation chosen uniformly from a symmetric permutation group of the set~$[n]$. Then for every $\eta > 0$, 
    $$
        \Pr[ \Delta( \mbfX{} ) - \E [ \Delta( \mbfX{} ) ] \ge \eta ] \le \exp \left( - \frac{2 \eta^2 }{ n_1 c^2} \left( \frac{n - 1 / 2 }{n - n_1} \right) \left( 1 - \frac{1}{2 \max \{n_1, n_2 \} } \right) \right)\,. 
    $$
    \end{fact}

\end{samepage}

\newpage
\section{Proofs For Section~\ref{sec: preliminaries}} \label{appendix: proof for sec: preliminaries}

This section is organized as follows: 
\begin{enumerate}
    \item In Section~\ref{appendix: proof of fact: hrr}, we provide the detailed proof for Fact~\ref{fact: hrr}. 
    \item In Section~\ref{appendix: proof of thm:lower bound frequency oracle}, we provide the detailed proof for Corollary~\ref{thm:lower bound frequency oracle}. 
\end{enumerate}

\subsection{Hadamard Randomized Response (\hrr) } \label{sec: hrr}
\label{appendix: proof of fact: hrr}

{\bf Fact~\ref{fact: hrr}}~(Algorithm \hrr~\citep{NXYHSS16, CKS19}). 
{\it 
    Let $\mcalU{}$ be a set users each holding an element from some finite domain $\mcalD{}$.
    There exists an $\eps$-locally differentially private frequency oracle, \hrr, such that the following holds.
    Fix some query element~$v \in \mcalD{}$ for \hrr. 
    With probability at least~$1 - \beta'$, \hrr returns a frequency estimate $\hat f_\mcalU{} [ v ]$ satisfying
    $$
    {\small
            \left|  \hat f_\mcalU{} [ v ] - f_\mcalU{} [ v ] \right| \in O \left( ( {1} / {\eps} ) \cdot \sqrt{ | \mcalU{} | \cdot \ln ( {1} / {\beta'} )  } \right)\,.
    }
    $$
    Each user in $\mcalU{}$ requires~$\tilde O(1)$~memory, takes~$\tilde O(1)$ running time and reports only~$1$ bit to the server. 
    The server processes the reports in~$\tilde O( |\mcalU{} | + |\mcalD{}| )$ time and~$O(|\mcalD{}|)$ memory, and answers a query in~$\tilde O(1)$ time. 
    The~$\tilde O$ notation hides logarithmic factors in $| \mcalU{} |$, $| \mcalD{} |$ and $1 / \beta'$. 
}

\subsubsection{The Hadamard Matrix}

The main vehicle for the \hrr algorithm is the Hadamard matrix.
In this section, we provide its definition, and prove some of its important properties. 

\begin{definition}[Hadamard Matrix] \label{def: Hadamard matrix}
    The Hadamard matrix is defined recursively for a parameter,~$m$, that is a power of two: $\small \mrmH_1 = [1]$ and
    $\small
        \mrmH_m = 
            \begin{bmatrix}
                \mrmH_{ m / 2}     &   \mrmH_{ m / 2} \\
                \mrmH_{ m / 2}     &   -\mrmH_{ m / 2}
            \end{bmatrix}
    $.
    For example,
    {\small
            $
                \mrmH_2 = 
                \begin{bmatrix}
                    1 &1 \\
                    1 &-1
                \end{bmatrix}, 
                \text{ and }
                \mrmH_4 = 
                \begin{bmatrix}
                    1 &1  &1     &1 \\
                    1 &-1 &1     &-1 \\
                    1 &1  &-1    &-1 \\
                    1 &-1 &-1    &1 
                \end{bmatrix}.   
            $
    }
    \end{definition}

Observe that a Hadamard matrix is symmetric. 
We list here some other important properties of Hadamard matrix.

\begin{fact}\label{fact:hada-orthogonal-columns}
    The columns of the Hadamard matrix are mutually orthogonal.
\end{fact}
\begin{proof}[\textbf{Proof of Fact \ref{fact:hada-orthogonal-columns} }]
    We prove this by induction on the size of $m$. For $\mrmH_1$, this is trivially true. Suppose this holds for $\mrmH_{ m / 2 }$. For every $i \in [m / 2]$, let $\pmbx{}_i \in \R^{m  / 2}$ be the $i^\text{th}$ column of $\mrmH_{m / 2}$. By the induction hypothesis, for all $i, j \in [m / 2]$, if $i \neq j$, then
    $
        \left< \pmbx{}_{i }, \pmbx{}_{j } \right> = 0.
    $
    
    Consider~$\mrmH_m$. For each $i \in [m]$, define $\pmby{}_i \in \R^m$ to be the $i^{\text{th}}$ column of $\mrmH_m$. Further, define
    $$
        c(i) \doteq \begin{cases}
            i,  & \text{ if } i \le m / 2\,, \\
            i - m / 2 & \text{ if } i > m / 2\,.
        \end{cases} 
        \quad \text{and} \qquad
        s(i) \doteq \begin{cases}
            1,  & \text{ if } i \le (m / 2)\,, \\
            -1  & \text{ if } i > (m / 2)\,.
        \end{cases} 
    $$
    Note that for each $i \in [m]$, we have $c(i) \in [m / 2]$.
    By the definition of $\mrmH_m$, it holds that for all $i, j \in [m]$, 
    $$
        \pmby{}_i = \begin{bmatrix}
            \pmbx{}_{c(i)} \\
            s( i ) \cdot \pmbx{}_{c(i)}
        \end{bmatrix}
        \quad \text{ and }
        \pmby{}_j = \begin{bmatrix}
            \pmbx{}_{c(j)} \\
            s( j ) \cdot \pmbx{}_{c(j)}
        \end{bmatrix}.
    $$
    Hence, 
    $$
        \left< \pmby{}_i, \pmby{}_j \right>  = 
        \left< \pmbx{}_{c(i)}, \pmbx{}_{c(j)} \right>  + s( i ) s( j ) \left< \pmbx{}_{c(i)}, \pmbx{}_{c(j)} \right>.
    $$
    If $i \neq j$, then there are two possible cases: 1) $i \neq j + m / 2$ and $j \neq i + m / 2$, then by the induction hypothesis, $\left< \pmbx{}_{c(i)}, \pmbx{}_{c(j)} \right> = 0$; and 2) $i = j + m / 2$ or $j = i + m / 2$, then $s( i ) s( j ) = -1$. In both cases, $\left< \pmby{}_i, \pmby{}_j \right> = 0$ holds.
\end{proof}

\begin{fact} \label{fact:hada-entry-computation}
    For $0 \le i, j < m$, the $(i + 1,j + 1)$-th entry $\mrmH_m[i + 1, j + 1]$ of the Hadamard matrix can be computed in $O( \log m )$ time (both the rows and columns of $\mrmH_m$ are indexed from $1$ to $m$). In particular, let the vectors $(i)_{\log m}, (j)_{\log m} \in \{0, 1\}^{\log m}$ be the $\log m$-bit binary representation of $i$ and $j$, respectively. Then 
        $
            \mrmH_m[i + 1, j + 1] = (-1)^{ \left< (i)_{\log m}, (j)_{\log m} \right> },
        $
    where $\left< (i)_{\log m}, (j)_{\log m} \right>$ is the dot product between $(i)_{\log m}$ and $(j)_{\log m}$.
\end{fact}  
\begin{proof}[\textbf{Proof of Fact \ref{fact:hada-entry-computation} }]
    We prove the Theorem by induction. The claim can be verified manually for~$\mrmH_1$ and~$\mrmH_2$. Suppose this holds for~$\mrmH_{m / 2}$; to prove it true for~$\mrmH_m$, recall that the recursive definition     
    $\small
        \mrmH_m = 
            \begin{bmatrix}
                \mrmH_{ m / 2}     &   \mrmH_{ m / 2} \\
                \mrmH_{ m / 2}     &   -\mrmH_{ m / 2}
            \end{bmatrix}
    $ 
    divides~$\mrmH_m$ into four sub-matrices. For $0 \le i, j < m$, let the vectors $(i)_{\log m}, (j)_{\log m} \in \{0, 1\}^{\log m}$ be the $\log m$-bit binary representations of~$i$ and~$j$, respectively. Let $b_i = (i)_{\log m}[1]$ be the highest bit and $\pmbs{}_i = (i)_{\log m}[2: \log m]$ be the last $(\log m) - 1$ bits of $(i)_{\log m}$, respectively. Similarly, we can define $b_j = (j)_{\log m}[1]$ and $\pmbs{}_j = (j)_{\log m}[2: \log m]$. Now, 
    $$\small    
        (i)_{\log m} = \begin{bmatrix}
            b_i,
            \pmbs{}_i
        \end{bmatrix},
        \qquad
        (j)_{\log m} = \begin{bmatrix}
            b_j,
            \pmbs{}_j
        \end{bmatrix}.
    $$
    Consider the $(i + 1, j + 1)$-th entry of $\mrmH_m$. Our goal is prove that 
    $$
        \mrmH_m [i + 1, j + 1] = (-1)^{ \left< (i)_{\log m}, (j)_{\log m}\right> } = (-1)^{ b_i b_j + \left< \pmbs{}_i , \pmbs{}_j \right> }. 
    $$
    Observe that $b_i b_j = 1$ if the $(i + 1, j + 1)$-th entry belongs to the lower right sub-matrix and $b_i b_j = 0$ otherwise. By definition of $\mrmH_m$, the sub-matrix this entry belongs to can be written as $(-1)^{b_i b_j } \mrmH_{m / 2}$. If we also view $\pmbs{}_i, \pmbs{}_j \in \{0, 1\}^{ ( \log m ) - 1 }$ as integers in $[0, m / 2)$, then $(\pmbs{}_i + 1, \pmbs{}_j + 1)$ is the pair of indexes of the entry inside the sub-matrix. By the induction hypothesis, the value of the $(\pmbs{}_i + 1, \pmbs{}_j + 1)$ entry of the matrix $(-1)^{b_i b_j } \mrmH_{m / 2}$ is given by 
    $(-1)^{ b_i b_j + \left< \pmbs{}_i , \pmbs{}_j \right> }$, which finishes the proof. 
\end{proof}

\begin{fact}[Fast Hadamard Transform] \label{fact:fast-hada-transform}
    For all $\pmbx{} \in \R^m$, there is a standard divide-and-conquer algorithm that computes the multiplication $\mrmH_m \,\pmbx{}$ (\emph{equivalently, $\mrmH_m^T \pmbx{}$, as~$\mrmH_m$ is symmetric}) in~$O(m \log m)$ time and~$O(m)$ memory. 
\end{fact}
\begin{proof}[\textbf{Proof of Fact \ref{fact:fast-hada-transform} }]
    Let $\pmbx{}_1\in \R^{m / 2}$ be the first $m / 2$ entries, and $\pmbx{}_2 \in \R^{m / 2}$ be the second $m / 2$ entries of $\pmbx{}$ respectively. Define $\pmby{}_1 = \mrmH_{ m / 2} \ \pmbx{}_1$ and $\pmby{}_2 = \mrmH_{ m / 2 } \ \pmbx{}_2$. Then
    $$
        \mrmH_m  \ \pmbx{} = 
            \begin{bmatrix}
                \mrmH_{ m / 2}     &   \mrmH_{ m / 2} \\
                \mrmH_{ m / 2}     &   -\mrmH_{ m / 2} \\
            \end{bmatrix} 
            \begin{bmatrix}
                \pmbx{}_1 \\
                \pmbx{}_2
            \end{bmatrix}
        =
            \begin{bmatrix}
                \mrmH_{ m / 2} \ \pmbx{}_1 + \mrmH_{ m / 2} \ \pmbx{}_2 \\
                \mrmH_{ m / 2} \ \pmbx{}_1 - \mrmH_{ m / 2} \ \pmbx{}_2
            \end{bmatrix}
        = 
            \begin{bmatrix}
                \pmby{}_1 + \pmby{}_2 \\
                \pmby{}_1 - \pmby{}_2
            \end{bmatrix}.
    $$
    Let $T(m)$ be the time to compute $\mrmH_m \ \pmbx{}$. Computing $\pmby{}_1$ and $\pmby{}_2$ takes time $2 \cdot T(m / 2)$. Computing $\pmby{}_1 + \pmby{}_2$ and $\pmby{}_1 - \pmby{}_2$ takes time $O(m)$. Therefore,
    $
        T(m) = 2 \cdot T( m / 2) + O(m). 
    $
    Solving the recursion gives $T(m) = O(m \log m)$. 
\end{proof}

\subsubsection{The Algorithm}

The algorithm relies on a Hadamard matrix $\mrmH_m$ with $m = 2^{ \lceil \log | \mcalD{} | \rceil }$ (hence $| \mcalD{} | \le m < 2 | \mcalD{} |$ ), and assigns each element $v \in \mcalD{}$ the $v^{\text{th}}$ column of $\mrmH_m$. 
For a user $u \in \mcalU{}$, it is said to be assigned to the $v^{\text{th}}$ column, if its element $v^{(u)} = v$. The problem of estimating the frequency of a element $v \in \mcalD{}$ reduces to estimating the number of users in~$\mcalU{}$ assigned to the $v^{\text{th}}$ column. 
When the value of~$m$ is clear from context, we omit subscript~$m$ and write Hadamard matrix~$\mrmH_m$ as~$\mrmH$.

\subsubsection{Client Side} 
The client-side algorithm is described in Algorithm~\ref{algo:hrr-Client}. 
Each user receives (from the server) a row index~$r$ of the Hadamard matrix, and privacy parameter~$\eps$. Its own element~$v$ is the column index. 
It returns the value of $\mrmH{}[r,v]$, but flipped with probability~$1 / (e^\eps + 1)$.  
Multiplying $\mrmH{}[r,v]$ by a Rademacher random variable~$b$ that equals~$1$ with probability $e^\eps / (e^\eps + 1)$ and~$-1$ with probability $1 / (e^\eps + 1)$ achieves the flip.
It returns the one-bit result to the server. 

\begin{algorithm}[H]
    \caption{\hrr-Client $\algoHrrClient$ }
    \label{algo:hrr-Client}
        \begin{algorithmic}[1]
            \REQUIRE Row index $r \in [m]$; privacy parameter $\eps$.
            \STATE Let $v \in \mcalD{}$ be the user's element.
            \STATE Sample $b \in \{-1, 1 \}$, which is $+1$ with probability $e^\eps / (e^\eps + 1)$.  
            \RETURN $\omega \leftarrow b \cdot \mrmH[r,v]$.
        \end{algorithmic}
\end{algorithm}

\begin{fact}[Running Time and Memory Usage] \label{fact: time of HRR client}
        Algorithm~\ref{algo:hrr-Client} has running time~$\tilde{O} ( 1 )$ and memory usage~$\tilde{O} ( 1 )$.
\end{fact}
\vspace{-4mm}
\begin{proof}[\textbf{Proof of Fact~\ref{fact: time of HRR client}}]
    According to Fact~\ref{fact:hada-entry-computation}, the entry~$\mrmH[r, v]$ can be computed in~$O(\log m) \subseteq O(\log | \mcalD{} |) \subseteq \tilde{O} (1)$ time and~$O(\log m) \subseteq O(\log | \mcalD{} |) \subseteq \tilde{O} (1)$ memory. 
\end{proof}

\begin{fact}[Privacy Guarantee] \label{fact: hrr-Client is private}
    Algorithm~\ref{algo:hrr-Client} is $\eps$-locally differentially private.
\end{fact}
\vspace{-4mm}
\begin{proof}[\textbf{Proof of Fact~\ref{fact: hrr-Client is private}}]
    We need to prove that the output distribution of $\algoHrrClient{}$ deviates little with the value~$v \in \mcalD{}$, the user's element. 
    To explicitly state the dependence of $\algoHrrClient$ on~$v$, we write its output as $\algoHrrClient{}(r, \eps; v)$. 
    It suffices to prove that~$\forall v, v' \in \mcalD{}$, the output distributions of $\algoHrrClient{}(r, \eps; v)$ and $\algoHrrClient{}(r, \eps; v')$ are similar. 
    There are only two possible outputs, namely,~$\{ - 1, 1\}$. Let~$b$ and~$b'$ be the Rademacher random variables generated by $\algoHrrClient{}(r, \eps; v)$ and $\algoHrrClient{}(r, \eps; v')$ respectively. 
    Then 
    \begin{align*}
        &\Pr[ \algoHrrClient{}(r, \eps; v) = 1 ] = \Pr[ b \cdot \mrmH[ r, v ] = 1] \le e^\eps / (e^\eps + 1)\,, \\
        &\Pr[ \algoHrrClient{}(r, \eps; v') = 1 ] = \Pr[ b' \cdot \mrmH[ r, v' ] = 1] \ge 1 / (e^\eps + 1)\,.
    \end{align*}
    Hence, 
    $
        \Pr[ \algoHrrClient{}(r, \eps; v) = 1 ] \le e^\eps \cdot \Pr[ \algoHrrClient{}(r, \eps; v') = 1 ].
    $
    By Definition~\ref{def: Differential Privacy}, the algorithm $\algoHrrClient{}$ is $\eps$-differentially private. 
\end{proof}

\subsubsection{Server Side}
The server-side algorithm is described in Algorithm~\ref{algo:hrr-server}. 
Its input comprises the set of users,~$\mcalU{}$, their elements' domain,~$\mcalD{}$, and privacy parameter~$\eps$. The server maintains a vector $\pmbomega{} \in \R^m$, which is initialized with all zeros. For each user $u \in \mcalU{}$, the server samples an integer $r^{ (u) } \in [m]$ independently and uniformly at random. 
Then it invokes $\algoHrrClient(r^{ (u) }, \eps)$ by sending $r^{(u)}$ and $\eps$ to user $u$. 
On receiving user $u$'s response,~$\omega^{ (u)}$, the server increases the $( r^{ (u) } )^{\text{th}}$ entry of $\pmbomega$ by $(e^\eps + 1) / (e^\eps - 1) \cdot \omega^{ (u) }$. 
Finally, it returns a vector $\hat f_\mcalU{} = \mrmH^T \,\pmbomega{} \in \R^m$. 
Note that the dimension of $\hat f_\mcalU{}$ is $m$, which could be larger than $|\mcalD{}|$: we use only the first $|\mcalD{}|$ entries of $\hat f_\mcalU{}$. 

\begin{algorithm}[!ht]
    \caption{\hrr-Server $\algoHrrServer{}$}
    \label{algo:hrr-server}
    \begin{algorithmic}[1]
        \REQUIRE A set of users $\mcalU{}$; element domain $\mcalD{}$; privacy parameter $\eps$. 
        \STATE Set $m \leftarrow 2^{ \lceil \log | \mcalD{} | \rceil }$,  $\pmbomega \leftarrow \{0 \}^m$.  
        \FOR{$u \in \mcalU{}$ }
            \STATE $r^{ (u)} \leftarrow$ uniform random integer from $[m]$.
            \STATE $\omega^{ (u)} \leftarrow \algoHrrClient (r^{ (u)}, \eps)$.
            \STATE $\pmbomega [ r^{ (u) } ] \leftarrow \pmbomega [ r^{ (u) } ] + (e^\eps + 1) / (e^\eps - 1) \cdot \omega^{ (u) }$.
        \ENDFOR
        \RETURN $\hat f_\mcalU{} \leftarrow \mrmH^T \,\pmbomega$.
    \end{algorithmic}
\end{algorithm}

\begin{fact}[Running Time and Memory Usage] \label{fact:hrr server time and memory}
    Algorithm~\ref{algo:hrr-server} has running time~$\tilde O(| \mcalD{} | + | \mcalU{} |)$ and memory usage~$O(| \mcalD{} |)$. 
\end{fact}
\begin{proof}[{\bf Proof of Fact~\ref{fact:hrr server time and memory}}]
    The server needs memory of size $O(m) \subseteq O(| \mcalD{} |)$ to store the vector $\pmbomega$. 
    Processing responses from the users in $\mcalU{}$ takes time $O(| \mcalU{} |)$. By Fact~\ref{fact:fast-hada-transform}, $\mrmH{}^T \pmbomega{}$ can be computed in $O(m \log m) \subseteq \tilde O( | \mcalD{} |)$ time, with memory usage $O(m) \subseteq O(| \mcalD{} |)$. 
    Hence the overall running time is $\tilde O(| \mcalD{} | + | \mcalU{} |)$ and memory usage is $O(| \mcalD{} |)$. 
\end{proof}

\begin{remark} \label{remark: differnt time of hrr}
    Via the fast Hadamard transform (Fact~\ref{fact:fast-hada-transform}), the server computes $\hat f_\mcalU{} \leftarrow \mrmH^T \,\pmbomega$ (Algorithm~\ref{algo:hrr-server}, line 6) in $\tilde O( |\mcalD{} | )$ time. 
    Then for each $v \in \mcalD{}$, if its frequency is queried, the server can return $\hat f_\mcalU{} [v]$ in $O(1)$ time. 
    There is another version of \hrr that omits line~6. 
    It has server running time $\tilde O( |\mcalU{}| )$. However, when it answers a frequency query for some $v \in \mcalD{}$, it needs to compute $\hat f_\mcalU{} [v] = (\mrmH^T \,\pmbomega ) [v]$ on the fly, which requires $\tilde O( \min \{ m, | \mcalU{} | \} )$ time.
\end{remark}

\subsubsection{Utility Guarantee}

We have proven that the client-side algorithm is~$\eps{}$-locally differentially private, and analyzed the running time and memory usage of both client-side and server-side algorithms. 
In this section, we discuss their utility guarantees. 

\begin{fact}[Expectation] \label{fact: unbiased estimator of hrr server}
    Let $\hat f_\mcalU{}$ be the estimate vector returned by Algorithm~\ref{algo:hrr-server}. Then for all~$v \in \mcalD{}$, $\hat f_\mcalU{} [v ]$ is an unbiased estimator of~$f_\mcalU{} [v ]$. 
\end{fact}

\begin{proof}[\textbf{Proof of Fact~\ref{fact: unbiased estimator of hrr server}}]
    For each $i \in [m]$, let $\pmbe{}_i$ be the $i^{\text{th}}$ standard basis vector. 
    For each user $u \in \mcalU{}$, let $\pmbc^{ (u)} \doteq \mrmH \ \pmbe{}_{ v^{ (u) } }$ be the column assigned to $u$, i.e., the $(v^{(u)})^\text{th}$ column of the Hadamard matrix $\mrmH$, and let $b^{(u)}$ be the Rademacher random variable generated when algorithm $\algoHrrClient{}$ is invoked for user $u$. 
    Let $\ceps{} \doteq (e^\eps + 1) / (e^\eps - 1)$. 
    By algorithm $\algoHrrClient$, the response from user $u$ can be expressed as $\omega^{ (u) } = b^{ (u) } \cdot \pmbc^{ (u) } [ r^{ (u) } ]$. 
    When the server receives the response~$\omega^{(u)}$, the update of~$\pmbomega$ can be rewritten as 
    $$
        \pmbomega{} \leftarrow \pmbomega{} +  \ceps{} \cdot  \omega^{ (u) } \cdot \pmbe{}_{ r^{ (u) } }\,. 
    $$
    Hence $\pmbomega{} = \sum_{u \in \mcalU{} } \ceps{} \cdot  \omega^{ (u) } \cdot \pmbe{}_{ r^{ (u) } } = \sum_{u \in \mcalU{} } \ceps{} \cdot b^{ (u) } \cdot \pmbc^{ (u) } [ r^{ (u) } ] \cdot \pmbe{}_{ r^{ (u) } }$, where $\pmbe{}_{ r^{ (u) } }$ is the $(r^{(u)})^\text{th}$ standard basis vector. 
    Let $\pmbc_{ v } \doteq \mrmH \ \pmbe{}_{ v }$ be the $v^{\text{th}}$ column of $\mrmH$. 
    Since $\hat f_\mcalU{} = \mrmH^T \ \pmbomega{}$, we have
    $$
        \hat f_\mcalU{} [ v ] = \left< \pmbc{}_{v },  \pmbomega{} \right> 
        = \sum_{u \in \mcalU{} } \ceps{} \cdot b^{ (u) } \cdot \pmbc^{ (u) } [ r^{ (u) } ] \cdot \left< \pmbc{}_{v }, \pmbe{}_{ r^{ (u) } } \right> 
        = \sum_{u \in \mcalU{} } \ceps{} \cdot b^{ (u) } \cdot \pmbc^{ (u) } [ r^{ (u) } ] \cdot \pmbc{}_{v } [ r^{ (u) } ]\,.
    $$
    By the independence of~$b^{ (u) }$ and~$r^{ (u) }$, and by linearity of expectation, we have 
    $$
        \E \left[ \hat f_\mcalU{} [ v ] \right] 
        = \sum_{u \in \mcalU{} } \ceps{} \cdot \E[ b^{ (u) } ] \cdot \E \left[ \pmbc^{ (u) } [ r^{ (u) } ] \cdot \pmbc{}_{v } [ r^{ (u) } ] \right] 
        = \sum_{u \in \mcalU{} }  \E \left[ \pmbc^{ (u) } [ r^{ (u) } ] \cdot \pmbc{}_{v } [ r^{ (u) } ] \right].
    $$
    The second equality follows from
    $\E[ b^{ (u) } ] = 1 \cdot e^\eps / ( e^\eps + 1 ) + (-1) \cdot 1 / ( e^\eps + 1 ) = 1 / \ceps{}\,$.
    
    As~$r^{ (u) }$ is sampled uniformly from~$[m]$, it holds that 
    $$
        \sum_{u \in \mcalU{} }  \E \left[ \pmbc^{ (u) } [ r^{ (u) } ] \cdot \pmbc{}_{v } [ r^{ (u) } ] \right] 
        = \sum_{u \in \mcalU{} } \frac{1}{m } \cdot  \sum_{j = 1}^m \left( \pmbc{}^{ (u) } [j] \cdot \pmbc{}_{v } [j] \right) 
        = \sum_{u \in \mcalU{} } \frac{1}{m }  \cdot  \left< \pmbc{}^{ (u) }, \pmbc{}_{v } \right> 
        = \sum_{u \in \mcalU{} } \dsone{} [ v = v^{ (u) } ]\,.
    $$
    
    The final equality follows from the orthogonality of columns of $\mrmH$, and that $\left< \pmbc{}_{v }, \pmbc{}_{v } \right> = m$. 
    We conclude that 
    $$
        \E \left[ \hat f_\mcalU{} [ v ] \right] = \sum_{u \in \mcalU{} } \dsone{} [ v = v^{ (u) } ] = f_\mcalU{} [ v ].
    $$ 
\end{proof}

\begin{fact}[Confidence Interval] \label{fact: confidence interval of hrr}
    For a fixed $v \in \mcalD{}$ and for all $\beta' \in (0, 1)$, with probability at least $1 - \beta'$, it holds that 
    \begin{align*}
        | \hat f_\mcalU{} [v ] - f_\mcalU{} [v ] | \in O\left( (1 / \eps) \cdot \sqrt{ | \mcalU{} | \cdot \ln ( {1} / {\beta'} ) } \right)\,.
    \end{align*}
    
\end{fact}
    
\begin{proof}[\textbf{Proof of Fact~\ref{fact: confidence interval of hrr}}]
    For each $u \in \mcalU{}$, define $Z^{ (u) } \doteq \ceps{} \cdot b^{ (u) } \cdot \pmbc^{ (u) } [ r^{ (u) } ] \cdot \pmbc{}_{v } [ r^{ (u) } ]$. The $\{Z^{ (u) }\}$ are independent random variables in the range of $[-\ceps{}, \ceps{}]$. As $\hat f_\mcalU{} [v ] = \sum_{ u \in \mcalU{} } Z^{ (u) }$ and $\E \left[ \hat f_\mcalU{} [ v ] \right] = f_\mcalU{} [ v ]$, by Hoeffding's inequality (Fact~\ref{fact: hoeffding}), for all $\eta > 0$, 
    $$
        \Pr \left[ \left| \hat f_\mcalU{} [ v ] -  f_\mcalU{} [ v ]\right| \ge  \eta \right] \le 2 \exp \left( - \frac{ 2 \eta^2 }{ \sum_{u \in \mcalU{} } (\ceps{} - (-\ceps{}) )^2 } \right)\,.
    $$
    If we upper bound the failure probability with $\beta'$, we obtain that $\eta \le \ceps{} \cdot \sqrt{ 2 | \mcalU{} | \ln ( {2} / {\beta'} ) }$.
    Noting that $C_\eps{} \in O( 1 / \eps{} )$ for $\eps{} \in O(1)$ finishes the proof.
\end{proof}

\subsection{Proof of Corollary~\ref{thm:lower bound frequency oracle}}
\label{appendix: proof of thm:lower bound frequency oracle}

{\bf Corollary~\ref{thm:lower bound frequency oracle}.}
{\it 
    Let~$\eps{} \in O(1)$. 
    Every $\eps{}$-\ldp frequency oracle algorithm achieving estimation error~$\lambda$ with probability at least~$1 - \beta'$ must have
    $$
        \lambda \in \Omega \left( ( {1} / {\eps} ) \cdot \sqrt{ | \mcalU{} | \cdot \ln ( { 1 } / {\beta'} )  } \right).
    $$
}
\vspace{-5mm}
\begin{proof}[Proof of Corollary~\ref{thm:lower bound frequency oracle}]
    Let~${\beta'} = \beta / d$ be the failure probability of the frequency oracle algorithm, and
    suppose by contradiction that $\lambda \in  o( (1 / \eps) \cdot \sqrt{ | \mcalU{} | \ln ( 1 / {\beta'} ) } )$ for this algorithm.
    If so, we could query it for the frequency of all elements in the domain~$\mcalD{}$. 
    By a union bound, we could thus construct a succinct histogram with failure probability at most~$\beta$ and with error for every element in $o( (1 / \eps) \cdot \sqrt{ | \mcalU{} | \ln ( 1 / {\beta'} ) } )$.
    Since the latter bound is in, $o( (1 / \eps) \cdot \sqrt{ | \mcalU{} | \ln ( |\mcalD{}| / \beta ) } )$, this contradicts Fact~\ref{fact:lower bound succinct histogram}.
\end{proof}

\newpage
\section{Proofs For Section~\ref{sec: frequency oracle}} \label{appendix: proof for sec: frequency oracle}

This section is organized as follows: 
\begin{enumerate}
    \item In Section~\ref{appendix: proof of thm: sketching framework}, we provide the detailed proof for Theorem~\ref{thm: sketching framework}. 
    \item In Section~\ref{appendix: proof of thm: size of the good sets}, we provide the detailed proof for Theorem~\ref{thm: size of the good sets}. 
    \item In Section~\ref{appendix: sec proof of Delta 0}, we provide the proof for Lemma~\ref{lem: bounds of Delta 0}, which we rely on to prove Theorem~\ref{thm: size of the good sets}.
    \item In Section~\ref{appendix: sec proof of Delta 1}, we provide the proof for Lemma~\ref{lem: bounds of Delta 1}, which we rely on to prove Theorem~\ref{thm: size of the good sets}.
\end{enumerate}

\subsection{Theorem~\ref{thm: sketching framework}} \label{appendix: proof of thm: sketching framework}

The properties of Theorem~\ref{thm: sketching framework} have been established implicitly in Section~\ref{sec: frequency oracle}.
Here we show how to put the pieces together explicitly. 

{\bf Theorem~\ref{thm: sketching framework}}~(Sketching Framework).
{\it
    For every~$\beta' \in (0, 1)$, $\algoOracle{}$ can be converted into an~$\eps{}$-\ldp frequency oracle, with server running time~$\tilde{O} ( \Phi_{\ttime} ( |\mcalU{}|, \sqrt{ |\mcalU{}| } ) )$ and memory usage~$\tilde{O} ( \Phi_{\mem} ( |\mcalU{}|, \sqrt{ |\mcalU{}| } ) )$.
    Fix an element~$v \in \mcalD{}$ to be given as a query to the new algorithm.
    With probability at least $1 - \beta'$, it returns an estimate $\hat f_\mcalU{} [ v ]$ satisfying
    \vspace{-3mm}
    $$
        \left|  \hat f_\mcalU{} [ v ] - f_\mcalU{} [ v ] \right| \in O \left( ( {1} / {\eps} ) \cdot \sqrt{ | \mcalU{} | \cdot \ln ( {1} / {\beta'} )  } \right).
    \vspace{-2mm}
    $$
}

\begin{proof}[{\bf Proof of Theorem~\ref{thm: sketching framework}}]
    Since Algorithm~\ref{algo: hada-oracle-framework} (Sketching Framework) invokes hash functions to reduce the domain size from~$|\mcalD{}|$ to $m \in O( \sqrt{|\mcalU{}|} )$, it follows that it has server running time~$\tilde{O} ( \Phi_{\ttime} ( |\mcalU{}|, \sqrt{ |\mcalU{}| } ) )$ and memory usage~$\tilde{O} ( \Phi_{\mem} ( |\mcalU{}|, \sqrt{ |\mcalU{}| } ) )$.
    The privacy guarantee follows from that Algorithm~\ref{algo: hada-oracle-framework} partitions the set of users~$\mcalU{}$ into subsets, and invokes~$\algoOracle{}$ for each subset. 
    Therefore, each user participates in only one copy of~$\algoOracle{}$. 
    As~$\algoOracle{}$ is~$\eps{}$ differentially private, so is Algorithm~\ref{algo: hada-oracle-framework}. 
    Finally, the utility guarantee follows from Corollary~\ref{corollary: median error guarantee}. 
\end{proof}

This finishes the proof of Theorem~\ref{thm: sketching framework}. 
In the next section, we discuss Theorem~\ref{thm: size of the good sets}. 

\subsection{Theorem~\ref{thm: size of the good sets}} \label{appendix: proof of thm: size of the good sets}

{\bf Theorem~\ref{thm: size of the good sets}}
{\it
    With probability at least~$1 - \beta' / 4$, it holds that $|\goodset{}_0| > (1 - 1 / 8) k$.
    And for each $v \in \mcalD{}$ and each $j \in [3]$, with probability at least $1 - \beta' / 4$, it holds  that $|\goodset{}_j (v) | > (1 - 1 / 8) k$. 
}

\begin{proof}[{\bf Proof of Theorem~\ref{thm: size of the good sets}}]
    We need to prove the Theorem for $\goodset{}_0, \goodset{}_1 (v), \goodset{}_2 (v), \goodset{}_3 (v)$, separately.
    As they are easier, we first bound the sizes of~$\goodsetB{} (v)$ and~$\goodsetC{} (v)$.
    
    \vspace{4mm}
    {\bf Bounding the Size of} $\goodset{}_2 (v)$.
    
        Fix some $v \in \mcalD{}$. 
        Recall that $\lambda_2 (j, v) \doteq | k f_{\mcalU{}_i, h_i } [ h(v) ] - k f_\mcalU{}_i [ v ] |$. 
        Via the definition of $\goodset{}_2 (v)$, it can be rewritten as 
        $$
            \goodset{}_2 (v)
                = \left \{ i \in [k] : | k f_{\mcalU{}_i, h_i} [ h(v) ] - k f_\mcalU{}_i [ v ] | \in O \left( \frac{1}{\eps} \cdot \sqrt{ k \cdot |\mcalU{}_i | \ln \frac{1}{\beta'} } \right) \right\}\,.
        $$
        Observe that for each $i \in [k]$, the (scaled) errors $|\lambda_2 (i, v)| / k = |  f_{\mcalU{}_i, h_i} [ h(v) ] - f_\mcalU{}_i [ v ] |$ result from hash collisions. 
        Consider an fixed $i \in [k]$. 
        For each user $u \in \mcalU{}_i$, define the indicator random variable $X_u = \dsone{}[ h_i ( v^{ (u) } ) = h_i ( v ) ]$ for the event~$h_i ( v^{ (u) } ) = h_i ( v) $.
        If $v^{ (u) } = v$, it always holds that $X_u = 1$.
        Otherwise, as $h_i$ is a pairwise-independent hash function, $\Pr[ X_u = 1 ] = 1 / m$. 
        
        By definition, $f_{ \mcalU{}_i, h_i} [ h_i(v) ] = \sum_{u \in \mcalU{}_i } X_u$. Therefore, 
        $$
            \E \left[ \left| f_{ \mcalU{}_i, h_i} [ h_i (v) ] -  f_\mcalU{}_i [ v ] \right| \right] 
            = \E \left[ \left| \sum_{u \in \mcalU{}_i } X_u - f_\mcalU{}_i [ v ] \right| \right]
            = \E \left[ \sum_{u \in \mcalU{}_i, v^{ (u) } \neq v } X_u  \right] \,.
        $$
        By linearity of expectation, we have 
        $$
            \E \left[ \left| f_{ \mcalU{}_i, h_i} [ h_i (v) ] -  f_\mcalU{}_i [ v ] \right| \right] = \frac{ | u \in \mcalU{}_i, v^{ (u) } \neq v | }{ m } \le \frac{| \mcalU{}_i | }{ m }\,. 
        $$
        By Markov's inequality, 
        \begin{equation*}
            \Pr \left[ \left| f_{ \mcalU{}_i, h_i} [ h_i (v) ] -  f_\mcalU{}_i [ v ] \right| \ge \frac{1}{\eps} \sqrt{ \frac{1}{k} \cdot |\mcalU{}_i | \ln \frac{4}{\beta'} } \right] 
            \le \frac{ | \mcalU{}_i | / m }{  ( 1 / \eps )  \sqrt{ (1 / k) \cdot |\mcalU{}_i | \ln ( 4 / \beta') } } 
            = \frac{ \eps{} \sqrt{ k | \mcalU{}_i | } }{ m \sqrt{ \ln ( 4 / \beta') } }.
        \end{equation*}
        Recall that Algorithm~\ref{algo: hada-oracle-framework} initializes $k = C_K \cdot \ln ( 4 / \beta')$ and $m = 8e^2 \cdot \sqrt{C_K} \cdot \eps \cdot \sqrt{| \mcalU{} |}$ for some constant $C_K$.
        The upper bound on the probability simplifies to $\sqrt{ |\mcalU{}_i| } \, / \, ( 8 e^2 \sqrt{ |\mcalU{}| } )$.
        Using that $|\mcalU{}| \ge |\mcalU{}_i|$, this upper bound further simplifies to $1 / (8e^2)$.

        For each $i \in [k]$, define the indicator random variable 
        $$
            Y_i \doteq \dsone{} \left[ | k f_{\mcalU{}_i, h_i} [ h(v) ] - k f_\mcalU{}_i [ v ] | \ge \frac{1}{\eps} \sqrt{ k \cdot |\mcalU{}_i | \ln \frac{4}{\beta'} } \right]\,
        $$
        for the event~$| k f_{\mcalU{}_i, h_i} [ h(v) ] - k f_\mcalU{}_i [ v ] | \ge ( 1 / \eps ) \sqrt{ k \cdot |\mcalU{}_i | \ln ( 4 / \beta' ) }$.
        Let~$Y \doteq \sum_{i \in [k] } Y_i$.
        We have~$\E[ Y_i ] \le 1 / (8e^2)$, and~$\mu \doteq \E[ Y ] \le k / ( 8 e^2)$.
        As the $h_1, .., h_k$ are chosen independently, the $\{ Y_i \}$ are independent.
        Via Chernoff bound~(Fact~\ref{fact: prototype chernoff bound}), 
        \begin{align*}
            \Pr \left[ Y \ge \frac{k}{8} \right] 
            = \Pr \left[ Y \ge \left( 1 + \left( \frac{k}{ 8 \mu} - 1 \right) \right) \mu \right] 
            \le \left(  \frac{ \exp \left( k /  (8 \mu) - 1 \right)  }{  \left( k /  (8 \mu) \right)^{ k /  (8 \mu) }  } \right)^\mu
            = \exp \left( \frac{k}{8} - \mu - \frac{k}{8} \ln \frac{k}{8 \mu} \right).
        \end{align*}
        
        For~$\mu \le k / ( 8 e^2)$, the function~$-\mu - (k / 8) \ln (k / (8 \mu) ) = -\mu - (k / 8) \ln (k / 8) + (k / 8) \ln \mu$ is maximized when~$\mu = k / ( 8 e^2)$.
        Therefore, 
        $$
            \Pr \left[ Y \ge \frac{k}{8} \right] \le \exp \left( \frac{k}{8} - \frac{k}{8 e^2} - \frac{k}{8} \ln e^2 \right) = \exp \left( - \frac{k}{8} \left( 1 + \frac{1}{e^2} \right) \right).
        $$
        
        Recall that $k = C_K \cdot \ln ( 4 / \beta')$. If we set $C_K = 8$, then we get 
        $
             \Pr[ Y \ge k / 8 ] \le \beta' / 4. 
        $
        
    $\blacksquare$

    \vspace{4mm}
    {\bf Bounding the Size of} $\goodset{}_3 (v)$.
    
        Via the assumption of~$\algoOracle{}$, for $i \in [k]$, with probability at most $1 / ( {8e^2} )$,
        $$
            | \hat f_{\mcalU{}_i, h_i} [ h_i (v) ] - f_{\mcalU{}_i, h_i} [ h_i (v) ] | \notin O \left( \frac{1}{\eps} \sqrt{ | \mcalU{}_i |  \cdot \ln (8e^2) } \right)\,.
        $$
        Scaling both sides by a factor of~$k$, we get 
        $$
            |\lambda_3 (i, v) | 
            = | k \cdot \hat f_{\mcalU{}_i, h_i} [ h_i (v) ] - k \cdot f_{\mcalU{}_i, h_i} [ h_i (v) ] | 
            \notin O \left( \frac{1}{\eps} \sqrt{ k^2 | \mcalU{}_i |  \cdot \ln (8e^2) } \right).
        $$
        Replacing one factor~$k$ with $C_K \cdot \ln ({4} / {\beta'} )$, we have 
        $
            |\lambda_3 (i, v) | \notin O ( ( {1} / {\eps} ) \sqrt{ k  | \mcalU{}_j | \cdot \ln ( {1} / {\beta'} ) } ).
        $
        Since for each $i \in [k]$, the event happens independently, the probability that there are more than $k / 8$ choices of $i \in [k]$ for which this event happens is at most
        \begin{align*}
                \binom{k }{k / 8} \left( \frac{1}{ 8e^2 } \right)^{ k / 8 } 
                \le \left( \frac{ e k }{ k / 8  } \right)^{k / 8} \left(  \frac{1}{ 8e^2 } \right)^{ k / 8 }
                = \left( \frac{1}{e} \right)^{ k / 8 }, 
        \end{align*}
        where the first inequality follows from that $\binom{k}{k / 8} \le \frac{ k^{k / 8} }{ (k / 8)! }$ and that $\frac{ (k / 8)^{k / 8} }{ (k / 8)! } \le e^{k / 8}$.
        Recall that~$k = C_K \cdot \ln ( 4 / \beta')$. If we set $C_K = 8$, then we get 
        $
             \left( 1 / e \right)^{ k / 8 } \le \beta' / 4.
        $
        
    $\blacksquare$
    
    \vspace{4mm}
    To bound the sizes of~$\goodset{}_0$ and~$\goodsetA{} (v)$, we need the following lemmas.
    
    \begin{lemma} \label{lem: bounds of Delta 0}
        Let $\Delta_0 \doteq \sum_{i \in [k] } \big\vert | \mcalU{}_i | - | \mcalU{} | / k \big\vert$.
        $\exists$ $C_0 > 0$, s.t., with probability $1 - \beta' / 4$: $\Delta_0 \le C_0 \sqrt{ | \mcalU{} | \ln ( 4 / \beta' ) }$.
    \end{lemma}
    
    \begin{lemma} \label{lem: bounds of Delta 1}
        Let 
        $
            \Delta_1 \doteq \sum_{i \in [k] } \Vert  f_\mcalU{}_i - f_\mcalU{} / k \Vert_2
        $,
        where 
        $
            \left\Vert f_\mcalU{}_i - { f_\mcalU{} } / k \right\Vert_2 
            \doteq 
            \sqrt{ 
                \sum_{ v' \in \mcalD{} } ( f_\mcalU{}_i [ v' ] -   f_\mcalU{} [ v' ] / k )^2 
            }.
        $
        There exists some constant~$C_1 > 0$, s.t., with probability $1 - \beta' / 4$: $\Delta_1 \le C_1 \sqrt{ | \mcalU{} | \ln ( 4 / \beta' ) }$.
    \end{lemma}
    
    We need to prove the lemmas for both {\it independent partitioning} and {\it permutation partitioning}.
    The proofs are technical, so we defer them to the end of the proof.
    For now, we show how to put them together to bound the sizes of~$\goodset{}_0$ and~$\goodsetA{} (v)$.
    
    \vspace{4mm}
    {\bf Bounding the Size of} $\goodset{}_0$.
    
        By Lemma~\ref{lem: bounds of Delta 0}, with probability at least $1 - \beta' / 4$, it holds that $\Delta_0 \le C_0 \sqrt{ | \mcalU{} | \ln ( 4 / \beta' ) }$ for some constant $C_0$. 
        Therefore, 
        $$
            k \cdot \Delta_0 = \sum_{i \in [k] } \big\vert k \cdot | \mcalU{}_i | - | \mcalU{} | \big\vert \le k C_0 \sqrt{ | \mcalU{} | \ln \frac{4}{\beta'} }\,.
        $$
        By a counting argument, the number of $i \in [k]$, such that $\big\vert k \cdot | \mcalU{}_i | - | \mcalU{} | \big\vert \ge 8 C_0 \sqrt{ | \mcalU{} | \ln ( 4 / \beta' ) }$ is bounded by $(1 / 8)k$. 
        This implies that for at least $(1 - 1 / 8)k$ of the $i \in [k]$, we have 
        $$
            \big\vert k \cdot | \mcalU{}_i | - | \mcalU{} | \big\vert \le 8 C_0 \sqrt{  | \mcalU{} | \ln \frac{4}{\beta'} }\,.
        $$
        
        By the assumption that $|\mcalU{}| \ge \ln ( 4 / \beta' )$, we get $\big\vert k \cdot | \mcalU{}_i | - | \mcalU{} | \big\vert \in \Theta (|\mcalU{}|)$.
        
    $\blacksquare$

    {\bf Bounding the Size of} $\goodset{}_1(v)$.
    
        By Lemma~\ref{lem: bounds of Delta 1}, with probability at least $1 - \beta' / 4$, it holds that $\Delta_1 \le C_1 \sqrt{ | \mcalU{} | \ln (4 / \beta' ) }$ for some constant $C_1$. 
        Hence,  
        $$
            k \cdot \Delta_1 = k \sum_{i \in [k] } \Vert  f_\mcalU{}_i - f_\mcalU{} / k \Vert_2 \le k C_1 \sqrt{  | \mcalU{} | \ln \frac{4}{\beta'} }.
        $$ 
        By a counting argument, the number of~$i \in [k]$, such that $k \left\Vert  f_\mcalU{}_i - { f_\mcalU{} } / k \right\Vert_2 \ge 8 C_1 \sqrt{ | \mcalU{} | \ln ( 4 / \beta' ) }$ is bounded by~$(1 / 8)k$. 
        This implies that for at least $(1 - 1 / 8)k$ of the $i \in [k]$, we have 
        $$
            \big\vert \lambda_1 (i, v) \big\vert 
            = 
            \big\vert k f_\mcalU{}_i [ v ] -  f_\mcalU{} [ v ] \big\vert 
            \le 
            k \left\Vert  f_\mcalU{}_i - { f_\mcalU{} } / k \right\Vert_2 
            \le 
            8 C_1 \sqrt{ | \mcalU{} | \ln ( 4 / \beta' ) }\,.
        $$
        which finishes the proof.
        
    $\blacksquare$
    
\end{proof}

In the following two sections, we prove Lemma~\ref{lem: bounds of Delta 0} and Lemma~\ref{lem: bounds of Delta 1} respectively. 

\subsection{Bounding~$\Delta_0$} \label{appendix: sec proof of Delta 0}

{\bf Lemma~\ref{lem: bounds of Delta 0}.}
{\it 
    Let $\Delta_0 \doteq \sum_{i \in [k] } \big\vert | \mcalU{}_i | - | \mcalU{} | / k \big\vert$.
    $\exists$ $C_0 > 0$, s.t., with probability $1 - \beta' / 4$: $\Delta_0 \le C_0 \sqrt{ | \mcalU{} | \ln ( 4 / \beta' ) }$.
}

The lemma holds trivially for {\bf permutation partitioning}, as in such case it holds that~$|\mcalU{}_i| = |\mcalU{}| / k$ and~$\Delta_0 = 0$. 
We need to prove the lemma for \textbf{independent partitioning}. 

\subsubsection{Proof of Lemma~\ref{lem: bounds of Delta 0} for {\bf Independent Partitioning}}

Without loss of generality, assume that $\mcalU{} = \{ 1, 2, .., | \mcalU{} | \}$. 
For each $u \in \mcalU{} $, let $X_u \in [k]$ be the index of the subset that user~$u$ belongs to.
Let $\mbfX{} \doteq (X_1,  \ldots, X_{ | \mcalU{} | })$: by definition, for each $i \in [k]$, $\Pr[ X_u = i ] = 1 / k$. 
The set $\mcalU{}_i$ can be represented as
$$
    \mcalU{}_i \doteq \{ u \in \mcalU{}  : X_u = i \}\,.
$$
Now,~$\Delta_0$ can be rewritten as 
$$ 
    \Delta_0 = \sum_{i \in [k] } \sqrt{ \left( \sum_{ u \in \mcalU } \dsone{} \left[ X_u = i \right] - |\mcalU{}| / k \right)^2 }\,, 
$$
where $\dsone{} \left[ X_u = i \right]$ is the indicator random variable for the event~$X_u = i$.
Hence, $\Delta_0$ is a random variables that depends on~$\mbfX{}$. We write $\Delta_0$ explicitly as $\Delta_0(X_1, \ldots, X_{ | \mcalU{} | })$ or $\Delta_0(\mbfX{})$ when necessary. 
For a sequence of values $\pmbx{} = \{ x_1, \ldots, x_{ | \mcalU{} | } \} \in [k]^{ | \mcalU{} | }$, we use $\Delta_0 (x_1, \ldots, x_{ | \mcalU{} | })$ or $\Delta_0 (\pmbx{})$ to denote the value of $\Delta_0$, when $\mbfX{} = \pmbx{}$. 
Observe that 
$$
    \Delta_0 = \Delta_0 - \E[ \Delta_0 ] + \E[ \Delta_0 ].
$$
In order to upper bound $\Delta_0$, we can upper bound both $\Delta_0 - \E[ \Delta_0]$ and $\E[ \Delta_0]$ superlatively. 
In particular, we will prove that 1) with probability at least~$1 - \beta' / 4$, it holds that~$\Delta_0 - \E[ \Delta_0 ] \le \sqrt{ 2 | \mcalU{} | \ln (4 / \beta') }$; 2) $\E[ \Delta_0 ] \le \sqrt{ k  | \mcalU{} | }$.
Substituting $k = C_K \cdot \ln ( 4 / \beta')$, we get that, with probability at least~$1 - \beta' / 4$, 
$$
    \Delta_0 \le \sqrt{ 2 | \mcalU{} | \ln (4 / \beta') } + \sqrt{ k  | \mcalU{} | } = \sqrt{ 2 | \mcalU{} | \ln (4 / \beta') } + \sqrt{ C_K | \mcalU{} | \ln (4 / \beta') }.
$$
As we set~$C_K = 8$ in the proof of Theorem~\ref{thm: size of the good sets} in Section~\ref{appendix: proof of thm: size of the good sets}, the RHS simplifies to~$\Delta_0 \le 3 \sqrt{ 2 | \mcalU{} | \ln (4 / \beta') }$.

\textbf{Step 1: Bounding $\Delta_0 - \E[ \Delta_0]$.}
    
    We upper bound it by McDiarmid’s Inequality (Fact~\ref{fact:McDiarmid-Inequality}). 
    We will prove that $\Delta_0$ satisfies Lipschitz condition (Definition~\ref{def:Lipschitz-Condition}) with bound $2$, i.e., for all $u \in \mcalU{} $, and every sequence of values $\pmbx{} = \{ x_1, \ldots,x_u, \ldots, x_{ | \mcalU{} | } \} \in [k]^{| \mcalU{} |}$ and $x_u' \in [k]$, 
    \begin{equation} \label{ineq: delta0-lipschitz}
        | \Delta_0(x_1, \ldots, x_u, \ldots, x_{ | \mcalU{} | }) - \Delta_0(x_1, \ldots, x_u', \ldots, x_{ | \mcalU{} | }) | \le 2\,. 
    \end{equation}
    Then by McDiarmid’s Inequality~(Fact~\ref{fact:McDiarmid-Inequality}), 
    $$
        \Pr\left[ \Delta_0 - \E[ \Delta_0] \ge \sqrt{ 2 | \mcalU{} | \ln \frac{4}{\beta'} } \right] \le \exp \left( - \frac{  2 \left( \sqrt{ 2 | \mcalU{} | \ln ( 4 / \beta') } \right)^2  }{ | \mcalU{} | \cdot 4 } \right) \le \beta' / 4\,.
    $$
    
    \noindent {\bf Proof of Inequality~(\ref{ineq: delta0-lipschitz})}.
    Define a random vector in~$\R^k$ that depends on~$\mbfX{}$ as
    $$
        \pmbomega (\mbfX{} ) \doteq \left( \sum_{ u \in \mcalU } \dsone{} \left[ X_u = 1 \right] - \frac{|\mcalU{}|} {k}, \ \ldots, \ \sum_{ u \in \mcalU } \dsone{} \left[ X_u = k \right] - \frac{|\mcalU{}|} { k} \right).
    $$
    For a sequence of values $\pmbx{} = \{ x_1, \ldots,x_u, \ldots, x_{ | \mcalU{} | } \} \in [k]^{| \mcalU{} |}$ , let $\pmbomega (\pmbx{} )$ be the vector of $\pmbomega (\mbfX{} )$ when $\mbfX{} = \pmbx{}$.
    By its definition, $\Delta_0 (\pmbx{} )$ equals $\Vert \pmbomega ( \pmbx{} ) \Vert_1$, the~$\ell_1$ norm of~$\pmbomega ( \pmbx{} )$.
    
    Consider a fixed $u \in \mcalU{} $. 
    Let $\pmbx{}' = \{ x_1, \ldots, x_u', \ldots, x_{ | \mcalU{} | } \}$ be the sequence obtained by replacing $x_u$ with $x_u'$. 
    The inequality~(\ref{ineq: delta0-lipschitz}) clearly holds when $x_u = x_u'$. 
    Now, suppose that $x_u \neq x_u$. 
    Then $\pmbomega( \pmbx{} )$ and $\pmbomega( \pmbx{}' )$ differ in only two coordinates, each by~$1$.
    Specifically, $\pmbomega( \pmbx{} ) - \pmbomega( \pmbx{}' ) = \pmbe{}_{ x_u } - \pmbe{}_{ x_u' }$, where~$\pmbe{}_{ x_u }$ and~$\pmbe{}_{ x_u' }$ are the~$(x_u)^{ (th) }$ and the~$(x_u')^{ (th) }$ standard basis vectors in~$\R^{ k }$, respectively.
    By the triangle inequality, 
    $$
        \Big\vert \left\Vert \pmbomega( \pmbx{} ) \right\Vert_1 - \left\Vert \pmbomega( \pmbx{}' ) \right\Vert_1 \Big\vert \le 
        \Vert \pmbe{}_{ x_u } - \pmbe{}_{ x_u' } \Vert_1 
        = 2. 
    $$

    \vspace{1mm}
    \noindent \textbf{Step 2: Bounding $\E[ \Delta_0 ]$.} 
    
    By Jensen's inequality, it holds  that 
    \begin{align*}
        \E\left[ \Delta_0 \right]
            =  \sum_{i \in [k] } \E \left[ \sqrt{ \left( \sum_{ u \in \mcalU } \dsone{} \left[ X_u = i \right] - |\mcalU{}| / k \right)^2 } \right]
            \le \sum_{i \in [k] } \sqrt{ \E\left[ \left( \sum_{ u \in \mcalU } \dsone{} \left[ X_u = i \right] - |\mcalU{}| / k \right)^2  \right] }\,.
    \end{align*}

    Fix an $i \in [k]$. For each~$u \in \mcalU{}$, define the indicator random variable $Z^{ (u) } \doteq \dsone{} [ X_u = i ]$ for the event $u \in \mcalU{}_i$. 
    Then $\Pr[ Z^{ (u) } = 1 ] = 1 / k$ and $\Pr[ Z^{ (u) } = 0 ] = 1 - 1 / k$. 
    Further
    $
        | \mcalU{}_i | = \sum_{ u \in \mcalU{} } Z^{ (u) },
    $
    is a sum of $|\mcalU{}|$ independent random variables and has expectation~$|\mcalU{}| / k$. 
    Hence, 
    \begin{align*}
       \E\left[  \left( \sum_{ u \in \mcalU{} } Z^{ (u) } - |\mcalU{}| / k \right)^2 \right] 
            = \Var{} \left[  | \mcalU{}_i | \right] 
            =  \sum_{ u \in \mcalU{} } \Var{} \left[ Z^{ (u) } \right] 
            \le \frac{ | \mcalU{} | }{ k }\,.
    \end{align*}
    Therefore, 
    $$
         \E\left[ \Delta_0 \right] 
         \le \sum_{i \in [k] } \sqrt{  \E\left[ \left( \sum_{ u \in \mcalU } \dsone{} \left[ X_u = i \right] - |\mcalU{}| / k \right)^2  \right] } 
         \le \sum_{i \in [k] } \sqrt{ \frac{ | \mcalU{} | }{ k } } 
         = \sqrt{ k | \mcalU{} |  }\,.
    $$
    $\square{}$    

\vspace{8mm}

This finishes the proof of Lemma~\ref{lem: bounds of Delta 0}.
Next, we prove Lemma~\ref{lem: bounds of Delta 1}.

\subsection{Bounding~$\Delta_1$} \label{appendix: sec proof of Delta 1}

{\bf Lemma~\ref{lem: bounds of Delta 1}.}
{\it 
    Let 
    $
        \Delta_1 \doteq \sum_{i \in [k] } \Vert  f_\mcalU{}_i - f_\mcalU{} / k \Vert_2
    $,
    where 
    $
        \left\Vert f_\mcalU{}_i - { f_\mcalU{} } / k \right\Vert_2 
        \doteq 
        \sqrt{ 
            \sum_{ v' \in \mcalD{} } ( f_\mcalU{}_i [ v' ] -   f_\mcalU{} [ v' ] / k )^2 
        }.
    $
    There exists some constant~$C_1 > 0$, s.t., with probability $1 - \beta' / 4$: $\Delta_1 \le C_1 \sqrt{ | \mcalU{} | \ln ( 4 / \beta' ) }$.
}

We need to prove the lemma for both {\bf independent partitioning} and {\bf permutation partitioning}.

\subsubsection{Proof of Lemma~\ref{lem: bounds of Delta 1} for {\bf Independent Partitioning}}

Without loss of generality, assume that $\mcalU{} = \{ 1, 2, .., | \mcalU{} | \}$. 
For each $u \in \mcalU{} $, let $X_u \in [k]$ be the index of the subset that user~$u$ belongs to.
Let $\mbfX{} \doteq (X_1,  \ldots, X_{ | \mcalU{} | })$: by definition, for each $i \in [k]$, $\Pr[ X_u = i ] = 1 / k$. 
The set $\mcalU{}_i$ can be represented as
$$
    \mcalU{}_i \doteq \{ u \in \mcalU{}  : X_u = i \}\,.
$$
For each $v \in \mcalD{}$, we have 
$$
    f_\mcalU{}_i [ v ] = \sum_{ u \in \mcalU{}_i } \dsone{} \left[ v^{ ( u ) } = v \right], 
$$
where $\dsone{} \left[ v^{ ( u ) } = v \right]$ is the indicator random variable for the event~$v^{ ( u ) } = v$. 
Therefore, 
$$
    \left\Vert f_\mcalU{}_i - \frac{ f_\mcalU{} }{ k }  \right\Vert_2 
    = 
    \sqrt{ \sum_{ v \in \mcalD{} } \left( f_\mcalU{}_i [ v ] -  \frac{ f_\mcalU{} [ v ] }{ k} \right)^2 }
$$ 
is a random variables that depends on~$\mbfX{}$. 
We write $\left\Vert f_\mcalU{}_i - f_\mcalU{} / k   \right\Vert_2$ explicitly as $\left\Vert f_\mcalU{}_i - f_\mcalU{} / k   \right\Vert_2 (X_1, \ldots, X_{ | \mcalU{} | })$ or $\left\Vert f_\mcalU{}_i - f_\mcalU{} / k   \right\Vert_2 (\mbfX{})$ when necessary. 
For a sequence of values $\pmbx{} = \{ x_1, \ldots, x_{ | \mcalU{} | } \} \in [k]^{ | \mcalU{} | }$, we use $\left\Vert f_\mcalU{}_i - f_\mcalU{} / k   \right\Vert_2 (x_1, \ldots, x_{ | \mcalU{} | })$ or $\left\Vert f_\mcalU{}_i - f_\mcalU{} / k   \right\Vert_2 (\pmbx{})$ to denote the value of $\left\Vert f_\mcalU{}_i - f_\mcalU{} / k   \right\Vert_2 ( \mbfX{} )$, when $\mbfX{} = \pmbx{}$. 

Moreover, as
$
    \Delta_1 
    = \sum_{i \in [k] }\left\Vert f_\mcalU{}_i - f_\mcalU{} / k   \right\Vert_2, 
$
it is also a random variables that depends on~$\mbfX{}$. 
We write $\Delta_1$ explicitly as $\Delta_1(X_1, \ldots, X_{ | \mcalU{} | })$ or $\Delta_1(\mbfX{})$ when necessary. 
For a sequence of values $\pmbx{} = \{ x_1, \ldots, x_{ | \mcalU{} | } \} \in [k]^{ | \mcalU{} | }$, we use $\Delta_1 (x_1, \ldots, x_{ | \mcalU{} | })$ or $\Delta_1 (\pmbx{})$ to denote the value of $\Delta_1$, when $\mbfX{} = \pmbx{}$. 

Observe that 
$$
    \Delta_1 = \Delta_1 - \E[ \Delta_1 ] + \E[ \Delta_1 ].
$$
In order to upper bound $\Delta_1$, we can upper bound both $\Delta_1 - \E[ \Delta_1]$ and $\E[ \Delta_1]$ superlatively. 
In particular, we will prove that 1) with probability at least~$1 - \beta' / 4$, it holds that~$\Delta_1 - \E[ \Delta_1 ] \le \sqrt{ 2 | \mcalU{} | \ln (4 / \beta') }$; 2) $\E[ \Delta_1 ] \le \sqrt{ k  | \mcalU{} | }$.
Substituting $k = C_K \cdot \ln ( 4 / \beta')$, we get that, with probability at least~$1 - \beta' / 4$, 
$$
    \Delta_1 \le \sqrt{ 2 | \mcalU{} | \ln (4 / \beta') } + \sqrt{ k  | \mcalU{} | } = \sqrt{ 2 | \mcalU{} | \ln (4 / \beta') } + \sqrt{ C_K | \mcalU{} | \ln (4 / \beta') }.
$$
As we set~$C_K = 8$ in the proof of Theorem~\ref{thm: size of the good sets} in Section~\ref{appendix: proof of thm: size of the good sets}, the RHS simplifies to~$\Delta_1 \le 3 \sqrt{ 2 | \mcalU{} | \ln (4 / \beta') }$.

\noindent \textbf{Step 1: Bounding $\Delta_1  - \E[ \Delta_1 ]$.}

    We upper bound it by McDiarmid’s Inequality (Fact~\ref{fact:McDiarmid-Inequality}). 
    We will prove that $\Delta_1$ satisfies Lipschitz condition (Definition~\ref{def:Lipschitz-Condition}) with bound $2$, i.e., for all $u \in \mcalU{} $, and every sequence of values $\pmbx{} = \{ x_1, \ldots,x_u, \ldots, x_{ | \mcalU{} | } \} \in [k]^{| \mcalU{} |}$ and $x_u' \in [k]$, 
    \begin{equation} \label{ineq: delta1-lipschitz}
        | \Delta_1(x_1, \ldots, x_u, \ldots, x_{ | \mcalU{} | }) - \Delta_1(x_1, \ldots, x_u', \ldots, x_{ | \mcalU{} | }) | \le 2\,. 
    \end{equation}
    Then by McDiarmid’s Inequality~(Fact~\ref{fact:McDiarmid-Inequality}), 
    $$
        \Pr\left[ \Delta_1 - \E[ \Delta_1] \ge \sqrt{ 2 | \mcalU{} | \ln \frac{4}{\beta'} } \right] \le \exp \left( - \frac{  2 \left( \sqrt{ 2 | \mcalU{} | \ln ( 4 / \beta') } \right)^2  }{ | \mcalU{} | \cdot 4 } \right) \le \beta' / 4\,.
    $$

    \noindent {\bf Proof of Inequality~(\ref{ineq: delta1-lipschitz})}. 
    
    Consider a fixed $u \in \mcalU{} $. Let $\pmbx{}' = \{ x_1, \ldots, x_u', \ldots, x_{ | \mcalU{} | } \}$ be the sequence obtained by replacing $x_u$ with $x_u'$. 
    The inequality~(\ref{ineq: delta1-lipschitz}) clearly holds when $x_u = x_u'$. 
    It is left to consider the case when $x_u \neq x_u'$.
    To simplify the notation, denote $j = x_u$ and $\ell = x_u'$.
    Via the definition that $\Delta_1  = \sum_{i \in [k] } \left\Vert f_\mcalU{}_i - \frac{ f_\mcalU{} }{ k }  \right\Vert_2$,
    the $\Delta_1 ( \pmbx{} )$ and $\Delta_1 ( \pmbx{}' )$ differ only in two terms.
    Specifically, 
    $$
        \Delta_1 ( \pmbx{} ) - \Delta_1 ( \pmbx{}' ) 
        = 
        \left\Vert f_\mcalU{}_j - \frac{ f_\mcalU{} }{ k }  \right\Vert_2 ( \pmbx{} ) - \left\Vert f_\mcalU{}_j - \frac{ f_\mcalU{} }{ k }  \right\Vert_2 ( \pmbx{}' ) 
        + 
        \left\Vert f_\mcalU{}_\ell - \frac{ f_\mcalU{} }{ k }  \right\Vert_2 ( \pmbx{} ) - \left\Vert f_\mcalU{}_\ell - \frac{ f_\mcalU{} }{ k }  \right\Vert_2 ( \pmbx{}' ) 
    $$
    For each~$v \in \mcalD{}$, define $f_{\mcalU{}_j, \pmbx{} } [ v ] \doteq | \{ u \in \mcalU{}_j : v^{ (u) } = v \} |$ to be the frequency of $v$ in the set $\{ v^{ (u) } : u \in \mcalU{}_j \}$, when~$\mbfX{} = \pmbx{}$.
    Let $f_{\mcalU{}_j, \pmbx{} } \doteq \big( f_{\mcalU{}_j, \pmbx{} } [ v ] : v \in \mcalD{} \big)$ be the frequency vector when~$\mbfX{} = \pmbx{}$.
    Similarly, define~$f_{\mcalU{}_j, \pmbx{}' }$ to be the frequency vector when~$\mbfX{} = \pmbx{}'$.
    Further, let $f_{\mcalU{}_{ \ell }, \pmbx{} }$ and $f_{\mcalU{}_{ \ell }, \pmbx{}' }$ be the frequency vectors defined on~$\mcalU{}_\ell$, when $\mbfX{} = \pmbx{}$ and $\mbfX{} = \pmbx{}'$ respectively.
    Recall that~$f_\mcalU{} = \big( f_\mcalU{} [v] : v \in \mcalD{} \big)$ denotes the frequency vector defined on the entire user set~$\mcalU{}$.
    
    Let~$v^{ (u) }$ be the data of user~$u$. 
    When the value of~$\mbfX{}$ changes from~$\pmbx{}$ to~$\pmbx{}'$, the subset that user~$u$ belongs to switches from~$\mcalU{}_j$ to~$\mcalU{}_\ell$.
    The frequency of~$v^{ (u) }$ in~$\mcalU{}_j$ decreases by~$1$, and such frequency in~$\mcalU{}_\ell$ increases by~$1$.
    Therefore, 
    $$
        f_{\mcalU{}_j, \pmbx{} } - f_{\mcalU{}_j, \pmbx{} ' } =  \pmbe{}_{ v^{ (u) } }\,,
        \quad
        f_{\mcalU{}_\ell, \pmbx{} } - f_{\mcalU{}_\ell, \pmbx{} ' } =  - \pmbe{}_{ v^{ (u) } }\,.
    $$
    where $\pmbe{}_{ v^{ (u) } }$ is the $v^{ (u) }$-th standard basis vector in $\R^{ |\mcalD{} | }$. 
    As the norm~$\Vert \cdot \Vert_2$ satisfies the triangle inequality, we obtain
    $$
        \left\Vert  f_{\mcalU{}_j, \pmbx{} } - \frac{ f_\mcalU{} }{ k } \right\Vert_2 - \left\Vert f_{\mcalU{}_j, \pmbx{} ' } - \frac{ f_\mcalU{} }{ k } \right\Vert_2 \le \Vert \pmbe{}_{ v^{ (u) } } \Vert_2 = 1\,,
        \quad
        \left\Vert  f_{\mcalU{}_{ \ell }, \pmbx{} } - \frac{ f_\mcalU{} }{ k } \right\Vert_2 - \left\Vert f_{\mcalU{}_{ \ell }, \pmbx{} ' } - \frac{ f_\mcalU{} }{ k } \right\Vert_2 
        \le 
        \Vert - \pmbe{}_{ v^{ (u) } } \Vert_2 = 1\,. 
    $$
    Therefore, $| \Delta_1 ( \pmbx{} ) - \Delta_1 ( \pmbx{}' ) | \le 2$.
    
\vspace{1mm}
\noindent \textbf{Step 2: Bounding $\E[ \Delta_1 ]$.} 
    By linearity of expectation, 
    $$
        \E[ \Delta_1  ] = \sum_{i \in [k] } 
        \E\left[ \left\Vert f_\mcalU{}_i - \frac{ f_\mcalU{} }{ k }  \right\Vert_2 \right].
    $$
    For a fixed~$i \in [k]$, by Jensen's inequality, it holds  that 
    \begin{align*}
        \E\left[ \left\Vert f_\mcalU{}_i - \frac{ f_\mcalU{} }{ k }  \right\Vert_2 \right]
            =  \E\left[ \sqrt{ \sum_{ v \in \mcalD{} } \left( f_\mcalU{}_i [ v ] -  \frac{ f_\mcalU{} [ v ] }{ k} \right)^2 } \right]
            \le \sqrt{ \sum_{ v \in \mcalD{} } \E\left[  \left( f_\mcalU{}_i [ v ] -  \frac{ f_\mcalU{} [ v ] }{ k} \right)^2 \right] }\,.
    \end{align*}
    
    \noindent 
    For a fixed $v \in \mcalD{}$, define $\mcalU{}[v ] \doteq \{ u \in \mcalU{} : v^{ (u) } = v \}$ as the subset of users in~$\mcalU{}$ holding element $v$. 
    It holds that $| \mcalU{}[v ] | = f_\mcalU{} [ v ]$.
    For each $u \in \mcalU{} [v]$, define the indicator random variable~$Z^{ (u) } \doteq \dsone{} [ X_u = i ]$ for the event~$u \in \mcalU{}_i$. 
    Then $\Pr[ Z^{ (u) } = 1 ] = 1 / k$ and $\Pr[ Z^{ (u) } = 0 ] = 1 - 1 / k$. 
    Then 
    $
        f_\mcalU{}_i [ v ] = \sum_{ u \in \mcalU{} [ v ] } Z^{ (u) },
    $
    is a sum of $f_\mcalU{} [ v ]$ independent random variables with expectation $f_\mcalU{} [ v ] / k$. Hence, 
    \begin{align*}
        \E\left[  \left( f_\mcalU{}_i [ v ] -  \frac{ f_\mcalU{} [ v ] }{ k} \right)^2 \right]
            = \Var{} \left[  f_\mcalU{}_i [ v ] \right] 
            &=  \sum_{ u \in \mcalU{} [ v ] } \Var{} \left[ Z^{ (u) } \right] 
            \le \frac{ f_\mcalU{} [ v ] }{ k }\,.
    \end{align*}
    Therefore, 
    $$
        \E\left[ \left\Vert f_\mcalU{}_i - \frac{ f_\mcalU{} }{ k }  \right\Vert_2 \right] \le \sqrt{ \sum_{ v \in \mcalD{} } \E\left[  \left( f_\mcalU{}_i [ v ] -  \frac{ f_\mcalU{} [ v ] }{ k} \right)^2 \right] } \le \sqrt{ \sum_{ v \in \mcalD{} } \frac{ f_\mcalU{} [ v ] }{ k } } = \sqrt{ \frac{| \mcalU{} |}{k} }\,.
    $$
    Finally, summing over all~$i \in [k]$, we obtain
    $$
        \E[ \Delta_1  ] = \sum_{i \in [k] } 
        \E\left[ \left\Vert f_\mcalU{}_i - \frac{ f_\mcalU{} }{ k }  \right\Vert_2 \right] \le \sqrt{ | \mcalU{} | k }\,. 
    $$
    
    $\square{}$

\vspace{8mm}
\subsubsection{Proof of Lemma~\ref{lem: bounds of Delta 1} for {\bf Permutation Partitioning}}

We need to bound~$\Delta_1 = \big( \Delta_1 - \E[ \Delta_1 ] \big) + \E[ \Delta_1 ]$; we bound $\Delta_1 - \E[ \Delta_1 ]$ by $2 \sqrt{  | \mcalU{} | \ln \frac{2}{\beta'} }$ and $\E[ \Delta_1 ]$ by $\sqrt{ k | \mcalU{} | }$.

Without loss of generality, assume that $\mcalU{} = \{ 1, 2, .., | \mcalU{} | \}$. 
Let~$\mbfX{} = (X_1, X_2, \ldots, X_{ | \mcalU{} | })$ be a random permutation of~$\mcalU{}$, i.e., one chosen uniformly at random from the set of all possible permutation of~$\mcalU{}$. 
Note that~$X_1, \ldots, X_{ | \mcalU{} | }$ are dependent random variables. 

For each~$j \in [k]$, by the way we generate~$\mcalU{}_j$, it has size~$| \mcalU{} | / k$ and $\mcalU{}_j = \left\{ X_i :  (j - 1) \cdot | \mcalU{} | / k + 1 \le i \le j \cdot | \mcalU{} | / k \right\}$. For every $v \in \mcalD{}$, we can write 
$$
    f_\mcalU{}_j [ v ] = \sum_{ X_i \in \mcalU{}_j } \dsone{} \left[ v^{ ( X_i ) } = v \right]\,,
$$
where $\dsone{} \left[ v^{ ( X_i ) } = v \right]$ is the indicator random variable for the event $v^{ ( X_i ) } = v$. 
Therefore, 
$$
    \left\Vert f_\mcalU{}_j - \frac{ f_\mcalU{} }{ k }  \right\Vert_2 
        = \sqrt{ \sum_{ v \in \mcalD{} } \left( f_\mcalU{}_j [ v ] -  \frac{ f_\mcalU{} [ v ] }{ k} \right)^2 }
        = \sqrt{ \sum_{ v \in \mcalD{} } \left( \sum_{ X_i \in \mcalU{}_j } \dsone{} \left[ v^{ ( X_i ) } = v \right] -  \frac{ f_\mcalU{} [ v ] }{ k} \right)^2 }\,.
$$

Now, $\Delta_1 = \sum_{j \in [k] } \left\Vert f_\mcalU{}_j - \frac{ f_\mcalU{} }{ k }  \right\Vert_2$ is a function that depends on $\mbfX{}$;
we write $\Delta_1$ explicitly as~$\Delta_1(\mbfX{})$ or~$\Delta_1(X_1, \ldots, X_{ | \mcalU{} | })$ when necessary. 
For a sequence of values $\pmbx{} = \{ x_1, \ldots, x_{ | \mcalU{} | } \}$, we use $\Delta_1 (x_1, \ldots, x_{ | \mcalU{} | })$ or $\Delta_1 (\pmbx{})$ to denote the value of $\Delta_1$, when $\mbfX{} = \pmbx{}$. 

\textbf{Martingale Construction.} 

We will apply a martingale concentration inequality  (Fact~\ref{fact:Azuma}) for the proof.
First, to construct a martingale that satisfies Definition~\ref{def: martingale}, we introduce a dummy variable~$X_0 \equiv 0$. For each $0 \le i \le |\mcalU{} |$, let~$\mbfS{}_i$ be shorthand for $(X_0, \ldots, X_i)$, and define
$$
    Y_i \doteq \E[ \Delta_1 \mid \mbfS{}_i ]\,.
$$
Clearly $Y_i$ is a function of~$X_0, \ldots, X_i$ and $\E[ |Y_i| ] \le \infty$. Moreover, 
$$
    \begin{aligned}
        \E[ Y_{i + 1} \mid \mbfS{}_i ] 
        = \E \left[ \E \left[ \Delta_1 \,\middle\vert\, \mbfS{}_i, X_{i + 1} \right] \,\middle\vert\, \mbfS{}_i \right] 
        = \E \left[ \Delta_1 \,\middle\vert\, \mbfS{}_i \right] 
        = Y_i. 
    \end{aligned}
$$    
The sequence $Y_0, \ldots, Y_n$ satisfies all conditions specified in Definition~\ref{def: martingale} and is a martingale. 
By definition, $Y_{ | \mcalU{} | }  = \E [ \Delta_1 \mid \mbfS{}_{ | \mcalU{} |  } ] = \Delta_1$, as once the values of $X_1, \ldots, X_{ | \mcalU{} | }$ are determined, so is $\Delta_1$. 
And we have $Y_0 = \E \left[ \Delta_1 \,\middle\vert\, X_0 \right] = \E \left[ \Delta_1 \right]$, as $X_0 \equiv 0$.

Observe that 
$$
    \Delta_1 = \Delta_1 - \E[ \Delta_1] + \E[ \Delta_1] = Y_{ | \mcalU{} | }  - Y_0 + Y_0\,, 
$$
In order to upper bound $\Delta_1$, we can upper bound both $Y_{ | \mcalU{} | }  - Y_0$ and~$Y_0$. 

\vspace{1mm}
\noindent \textbf{Step 1: Bounding $\Delta_1 - \E[ \Delta_1]$.}

We upper bound it via Azuma's Inequality (Fact~\ref{fact:Azuma}). 
We prove that,
\begin{equation} \label{ineq: martingale diff 1}
    A_i \le Y_i - Y_{i - 1}  \le A_i + 2 \sqrt{2}\,, 
\end{equation}
for some random variables~$\{A_i\}$ that are functions of $X_0, \ldots, X_{i - 1}$. 
By Azuma's inequality, 
$$
    \Pr\left[ |Y_{ | \mcalU{} | }  - Y_0| \ge \left( 2 \sqrt{| \mcalU{} | \ln \frac{2}{\beta'} } \right) \right] \le 2 \exp \left( - \frac{  2 \left( 2 \sqrt{| \mcalU{} | \ln ( 2 / \beta' ) } \right)^2  }{ \sum_{i \in \mcalU{} } (2 \sqrt{2} )^2 } \right) \le \beta'\,.
$$

\textbf{Proof of Inequality~(\ref{ineq: martingale diff 1})}. We prove that the gap between the upper and lower bounds on $Y_i - Y_{i - 1}$ is at most $2 \sqrt{2}$. By the definitions of $Y_i$ and $Y_{i - 1}$,
$$
    Y_i - Y_{i - 1}  = \E[ \Delta_1 \mid \mbfS{}_i ] - \E[ \Delta_1 \mid \mbfS{}_{i - 1} ]\,.
$$
Let $\mcalU{} \setminus \mbfS{}_{i - 1}$ be the set of integers in~$\mcalU{}$ that are distinct from $X_1, \ldots, X_{i - 1}$. 
Define 
\begin{align*}
    A_i = \inf_{ x \in \mcalU{} \setminus \mbfS{}_{i - 1} } \E[ \Delta_1 \mid \mbfS{}_{i - 1}, X_i = x ] - \E[ \Delta_1 \mid \mbfS{}_{i - 1} ]\,,\, \text{and}\quad
    B_i = \sup_{ x' \in \mcalU{} \setminus \mbfS{}_{i - 1} } \E[ \Delta_1 \mid \mbfS{}_{i - 1}, X_i = x' ] - \E[ \Delta_1 \mid \mbfS{}_{i - 1} ]\,.
\end{align*}
Clearly $A_i \le Y_i - Y_{i - 1} \le B_i$. 
We  prove that $B_i - A_i \le 2 \sqrt{2}$. 
Let $x_0 \equiv 0$, and for each $j \ge 0$, $\pmbx{}_j = (x_0, x_1, \ldots, x_j )$ be the sequence that consists of a starting~$0$, and the first $j$ entries of a possible permutation of~$\mcalU{}$. 
For each~$i \ge 1$, let $\mcalU{} \setminus \pmbx{}_{i - 1}$ be the set of integers in $\mcalU{}$ that are distinct from $x_1, \ldots, x_{i - 1}$. 
Conditioned on $\mbfS{}_{i - 1} = \pmbx{}_{i - 1}$, 
\begin{align*} 
    B_i - A_i 
    &= \sup_{ x' \in \mcalU{} \setminus \pmbx{}_{i - 1} } \E[ \Delta_1 \mid \mbfS{}_{i - 1} = \pmbx{}_{i - 1}, X_i = x' ] - \inf_{ x \in \mcalU{} \setminus \pmbx{}_{i - 1} }\E[ \Delta_1 \mid \mbfS{}_{i - 1} = \pmbx{}_{i - 1}, X_i = x ] \\
    &= \sup_{ x', x \in \mcalU{} \setminus \pmbx{}_{i - 1} } \left( \E[ \Delta_1 \mid \mbfS{}_{i - 1} = \pmbx{}_{i - 1}, X_i = x' ] - \E[ \Delta_1 \mid \mbfS{}_{i - 1} = \pmbx{}_{i - 1},  X_i = x ] \right)\,.
\end{align*}
It suffices to bound this for every possible sequence of $\pmbx{}_{i - 1}$. 

Consider a fixed $i \in \mcalU{}$ and $\pmbx{}_{i - 1}$. Define
$$
    \gamma_{x', x} \doteq \E[ \Delta_1 \mid \mbfS{}_{i - 1} = \pmbx{}_{i - 1}, X_i = x' ] - \E[ \Delta_1 \mid \mbfS{}_{i - 1} = \pmbx{}_{i - 1}, X_i = x ] \,.
$$
Our goal is to prove for all $x', x \in \mcalU{} \setminus \pmbx{}_{i - 1}$, $\gamma_{x', x'} \le 2 \sqrt{2}$. 
It follows that $\sup_{ x', x \in \mcalU{} \setminus \pmbx{}_{i - 1} } \gamma_{x', x'} \le 2 \sqrt{2}$.
If $x' = x$, then $\gamma_{x', x'} = 0$. 
Suppose $x' \neq x$. 
As $\mbfX{}$ is a random permutation of $\mcalU{}$, conditioned on $\mbfS{}_{i - 1} = \pmbx{}_{i - 1}$ and $X_i = x'$, with equal probability, one of the elements $X_{i + 1}, \ldots, X_{ | \mcalU{} | }$ equals $x$. 
Hence, 
$$
    \E[ \Delta_1 \mid \mbfS{}_{i - 1} = \pmbx{}_{i - 1}, X_i = x' ] 
    = \frac{1}{ | \mcalU{} | - i} \sum_{ \ell = i + 1 }^{ | \mcalU{} | } \E[ \Delta_1 \mid \mbfS{}_{i - 1} = \pmbx{}_{i - 1}, X_i = x', X_\ell  = x ]\,. 
$$
Similarly, it holds that 
$$
    \E[ \Delta_1 \mid \mbfS{}_{i - 1} = \pmbx{}_{i - 1}, X_i = x ] 
    = \frac{1}{ | \mcalU{} | - i} \sum_{ \ell = i + 1 }^{ | \mcalU{} | } \E[ \Delta_1 \mid \mbfS{}_{i - 1} = \pmbx{}_{i - 1}, X_i = x, X_\ell  = x' ]\,. 
$$

\noindent By triangle inequality,
$$
     \gamma_{x', x} 
        = \frac{1}{ | \mcalU{} | - i} \sum_{ \ell = i + 1 }^{ | \mcalU{} | } 
        \left[ 
            \E[ \Delta_1 \mid \mbfS{}_{i - 1} = \pmbx{}_{i - 1}, X_i = x', X_\ell  = x ] 
            - \E[ \Delta_1 \mid \mbfS{}_{i - 1} = \pmbx{}_{i - 1}, X_i = x, X_\ell  = x' ] 
        \right]
$$

For all permutation sequences $\pmbx{} = (x_1, \ldots, x_{ | \mcalU{} | })$ and for all $i \neq \ell \in \mcalU{}$, define $\pmbx{}_{i, \ell}$ to be the sequence with the values of~$x_i$ and~$x_\ell$ being swapped. We claim it holds that 
\begin{equation} \label{ineq:swap-diff}
    | \Delta_1( \pmbx{} ) - \Delta_1( \pmbx{}_{i, \ell} ) | \le 2 \sqrt{2}, 
\end{equation}
This proves that 
$$
    \E[ \Delta_1 \mid \mbfS{}_{i - 1} = \pmbx{}_{i - 1}, X_i = x', X_\ell  = x ] 
    - \E[ \Delta_1 \mid \mbfS{}_{i - 1} = \pmbx{}_{i - 1}, X_i = x, X_\ell  = x' ]  \le 2 \sqrt{2},\,
    \quad
    \text{and}\quad
    \gamma_{x', x} \le 2 \sqrt{2}.
$$ 

To prove Inequality~(\ref{ineq:swap-diff}), recall that $\Delta_1 = \sum_{j \in [k] } \left\Vert f_\mcalU{}_j - \frac{ f_\mcalU{} }{ k }  \right\Vert_2$. 

If $v^{ (x_i) } = v^{ ( x_\ell )}$ or there exists some $j \in [k]$, s.t. both $x_i, x_\ell \in \mcalU{}_j$, then the swap does not change~$\Delta_1$, and $| \Delta_1( \pmbx{} ) - \Delta_1( \pmbx{}_{i, \ell} ) | = 0$.
    
Otherwise, $v^{ (x_i) } \neq v^{ (x_\ell )}$ and $x_i \in \mcalU{}_j$, $x_\ell \in \mcalU{}_{j'}$ for different $j, j' \in [k]$. The swap affects only $\left\Vert f_\mcalU{}_j - \frac{ f_\mcalU{} }{ k }  \right\Vert_2$ and $\left\Vert f_\mcalU{}_{j'} - \frac{ f_\mcalU{} }{ k }  \right\Vert_2$. 
    Let $f_{\mcalU{}_j, \pmbx{} }$ and $f_{\mcalU{}_j, \pmbx{}_{i, \ell} }$ be the frequency vectors when $\mbfX{} = \pmbx{}$ and $\mbfX{} = \pmbx{}_{i, \ell}$, respectively. 
    They differ in both the $( v^{ (x_i ) } )^\text{th}$ and $( v^{ (x_\ell ) } )^\text{th}$ coordinates, each by 1. 
    
    If we view $f_{\mcalU{}_j, \pmbx{} } - \frac{ f_\mcalU{} }{ k }$ and $f_{\mcalU{}_j, \pmbx{}_{i, \ell} } - \frac{ f_\mcalU{} }{ k }$ as $|\mcalD{}|$-dimensional vectors, it holds that 
    $$
        \left( f_{\mcalU{}_j, \pmbx{} } - \frac{ f_\mcalU{} }{ k } \right) - \left( f_{\mcalU{}_j, \pmbx{}_{i, \ell} } - \frac{ f_\mcalU{} }{ k } \right) = - \pmbe{}_{ v^{ (x_{i} ) } } + \pmbe{}_{ v^{ (x_\ell ) } },  
    $$
    where $\pmbe{}_{ v^{ (x_{i} ) } }$ and $\pmbe{}_{ v^{ (x_\ell ) } }$ are the $v^{ (x_{i} ) }$-th and the $v^{ (x_\ell ) }$-th standard basis vectors in $\R^{ |\mcalD{} | }$ respectively. By the triangle inequality, 
    $$
        \left\Vert  f_{\mcalU{}_j, \pmbx{} } - \frac{ f_\mcalU{} }{ k } \right\Vert_2 - \left\Vert f_{\mcalU{}_j, \pmbx{}_{i, \ell} } - \frac{ f_\mcalU{} }{ k } \right\Vert_2 \le \Vert - \pmbe{}_{ v^{ (x_{i} ) } } + \pmbe{}_{ v^{ (x_\ell ) } } \Vert_2 = \sqrt{2}. 
    $$

    Similarly, we can prove that the change of $\left\Vert f_\mcalU{}_{j'} - \frac{ f_\mcalU{} }{ k }  \right\Vert_2$ is bounded by $\sqrt{2}$. Therefore, $| \Delta_1( \pmbx{} ) - \Delta_1( \pmbx{}_{i, \ell} ) | \le 2 \sqrt{2}$.

\noindent \textbf{Step 2: Bounding $\E[ \Delta_1]$.} 

By linearity of expectation, 
$$
    Y_0 = \E[ \Delta_1 ] = \sum_{j \in [k] } 
    \E\left[ \left\Vert f_\mcalU{}_j - \frac{ f_\mcalU{} }{ k }  \right\Vert_2 \right].
$$
For a fixed $j \in [k]$, by Jensen's inequality, it holds that 
\begin{align*}
    \E\left[ \left\Vert f_\mcalU{}_j - \frac{ f_\mcalU{} }{ k }  \right\Vert_2 \right]
        =  \E\left[ \sqrt{ \sum_{ v \in \mcalD{} } \left( f_\mcalU{}_j [ v ] -  \frac{ f_\mcalU{} [ v ] }{ k} \right)^2 } \right]
        \le \sqrt{ \sum_{ v \in \mcalD{} } \E\left[  \left( f_\mcalU{}_j [ v ] -  \frac{ f_\mcalU{} [ v ] }{ k} \right)^2 \right] }\,.
\end{align*}

Consider a fixed $v \in \mcalD{}$, define $\mcalU{} [ v ] \doteq \{ u \in \mcalU{} : v^{ (u) } = v \}$ as the set of users holding element $v$. It holds that $| \mcalU{}[ v ] | = f_\mcalU [ v ]$.
For each $u \in \mcalU{} [v]$, define the indicator random variable~$Z^{ (u) }$ for the event $u \in { \mcalU{}_j }$. 
Then $\Pr[ Z^{ (u) } = 1 ] = 1 / k$ and $\Pr[ Z^{ (u) } = 0 ] = 1 - 1 / k$. 
Then 
$$
    f_{ \mcalU{}_j } [ v ] = \sum_{ u \in \mcalU{} [ v ] } Z^{ (u) },
$$
is a sum of~$f_\mcalU[ v ]$ dependent random variables with expectation~$f_\mcalU[ v ] / k$. For a pair of users $u, u' \in \mcalU{} [ v ], u \neq u'$, due the permutation, if~$u$ belongs to~$\mcalU{}_j$, it is less likely that~$u'$ belongs to~$\mcalU{}_j$. In particular, 
\begin{align*}
   \Cov{} \left[ Z^{ (u) }, Z^{ (u') } \right] 
    &= \E\left[ Z^{ (u) } \cdot Z^{ (u') } \right] - \E\left[ Z^{ (u) } \right] \E \left[ Z^{ (u') } \right] \\
    &= \binom{ | \mcalU{} | - 2}{ | \mcalU{} | / k - 2} / \binom{ | \mcalU{} | }{ | \mcalU{} | / k} - \left( \frac{1}{k} \right)^2 \\
    &= \frac{ | \mcalU{} | / k \cdot ( | \mcalU{} | / k - 1)}{ | \mcalU{} | ( | \mcalU{} | - 1)} - \left( \frac{1}{k} \right)^2 \\
    &\le 0\,.
\end{align*}
Hence, 
\begin{align*}
    \E\left[  \left( f_\mcalU{}_j [ v ] -  \frac{ f_\mcalU{} [ v ] }{ k} \right)^2 \right]
        &= \Var{} \left[ f_\mcalU{}_j [ v ] \right] 
        =  \sum_{ u \in \mcalU{} [ v ] } \Var{} \left[ Z^{ (u) } \right] +
            \sum_{ u \neq u' \in \mcalU{} [ v ] } \Cov{} \left[ Z^{ (u) }, Z^{ (u') } \right]
        \\
        &\le \sum_{ u \in \mcalU{} [ v ] } \Var{} \left[ Z^{ (u) } \right]
        = f_\mcalU [ v ] \cdot \frac{1}{k} \cdot \left( 1 - \frac{1}{k} \right) \le \frac{ f_\mcalU [ v ] }{ k }\,.
\end{align*}
Therefore, 
$$
    \E\left[ \left\Vert f_\mcalU{}_j - \frac{ f_\mcalU{} }{ k }  \right\Vert_2 \right] \le \sqrt{ \sum_{ v \in \mcalD{} } \E\left[  \left( f_\mcalU{}_j [ v ] -  \frac{ f_\mcalU{} [ v ] }{ k} \right)^2 \right] } \le \sqrt{ \sum_{ v \in \mcalD{} } \frac{ f_\mcalU{} [ v ] }{ k } } = \sqrt{ \frac{ | \mcalU{} |}{k} }\,,
$$
and 
$$
    Y_0 = \E[ \Delta_1 ] = \sum_{j \in [k] } 
    \E\left[ \left\Vert f_\mcalU{}_j - \frac{ f_\mcalU{} }{ k }  \right\Vert_2 \right] \le \sqrt{  | \mcalU{} | k }\,. 
$$

\newpage
\section{Proofs For Section~\ref{sec: succinct histogram}} \label{appendix: proof for sec: succinct histogram}

This section is organized as follows: 
\begin{enumerate}
    \item In Section~\ref{appendix: proof for thm: hadaheavy estiamtion error of single element}, we provide the detailed proof for Theorem~\ref{thm: hadaheavy estiamtion error of single element}. 
    \item In Section~\ref{appendix: proof for thm: elements in P_tau}, we provide the detailed proof for Theorem~\ref{thm: elements in P_tau}.
\end{enumerate}

\subsection{Theorem~\ref{thm: hadaheavy estiamtion error of single element}} \label{appendix: proof for thm: hadaheavy estiamtion error of single element}

\textbf{Theorem~\ref{thm: hadaheavy estiamtion error of single element}.}
    For each $\tau \in [L]$, fix some query string~$\pmbs{} \in \Lambda^\tau$ for the frequency estimate. 
    It holds that, with probability~$1 - \beta'$,
    $$
        | \hat f_\mcalU{} [ \pmbs{} ] - f_\mcalU{} [ \pmbs{} ] | \in O( ( 1 / \eps) \sqrt{ n \cdot (\log d) \cdot (\ln (1 / \beta') ) / \ln n } )\,.
    $$

\begin{proof}[{\bf Proof of Theorem~\ref{thm: hadaheavy estiamtion error of single element}}]

    Recall that for each $\tau \in [L]$ and each $\pmbs{} \in \Lambda^\tau$, $f_{\U_\tau} [ \pmbs{} ] \doteq | \{ u \in \U_\tau : v^{ (u) }[1 : \tau] = \pmbs{} \} |$ is the frequency of~$\pmbs{}$ in $\U_\tau$, and $\hat f_{\U_\tau} [ \pmbs{} ]$ is its estimate by \oracle. 
    
    For each $\tau \in [L]$, define frequency vector~$f_{\U_\tau } \doteq \big( f_{\U_\tau} [ \pmbs{} ] : \pmbs{} \in \Lambda^\tau \big)$.
    Denote the~$\ell_2$ distance between the frequency vector~$f_{\U_\tau }$ and its expectation~$f_\mcalU{} / L$ as 
    $$
        \left\Vert f_{\U_\tau } - \frac{ f_\mcalU{} }{ L }  \right\Vert_2 \doteq \sqrt{ \sum_{ \pmbs{} \in \Lambda^\tau } \left( f_{\U_\tau } [ \pmbs{} ] -  \frac{ f_\mcalU{} [ \pmbs{} ] }{ L } \right)^2 }\,. 
    $$

    To prove Theorem~\ref{thm: hadaheavy estiamtion error of single element}, we need the following lemmas. 
    
    \begin{lemma} \label{lem: concentration of U_tau}
        For each $\tau \in [L]$, with probability $1 - \beta'$, it holds that $| \U_\tau | \in O \left( n / L \right)$. 
    \end{lemma}
    
    \begin{lemma} \label{lem: concentration of f_U_tau}
        For each $\tau \in [L]$, with probability $1 - \beta'$, it holds that 
        $$
                \left\Vert f_{\U_\tau } - \frac{ f_\mcalU{} }{ L }  \right\Vert_2 \in O \left( \sqrt{ \frac{ n }{L} \ln \frac{1}{\beta'} } \right)\,.
        $$
    \end{lemma}

    We need to prove the lemmas for both {\it independent partitioning} and {\it permutation partitioning}.
    The proofs are technical, so we defer them to the end of the proof.
    For now, we show how to put them together to complete the proof of Theorem~\ref{thm: hadaheavy estiamtion error of single element}. 
    
    For each $\tau \in [L]$, each $\pmbs{} \in \Lambda^\tau$, we regard $\hat f_\mcalU{} [ \pmbs{} ] = L \cdot \hat f_{\U_\tau} [ \pmbs{} ]$ as an estimate of $f_\mcalU{} [ \pmbs{} ]$. By triangle inequality, 
        $$
            \left| L \cdot \hat f_{\U_\tau} [ \pmbs{} ] - f_\mcalU{} [ \pmbs{} ] \right| 
            \le \left| L \cdot \hat f_{\U_\tau} [ \pmbs{} ] - L \cdot f_{\U_\tau} [ \pmbs{} ] \right| 
                + \left| L \cdot f_{\U_\tau} [ \pmbs{} ] - f_\mcalU{} [ \pmbs{} ] \right|,
        $$
    where $\hat f_{\U_\tau} [ \pmbs{} ]$ is the estimate of $f_{\U_\tau} [ \pmbs{} ]$ returned by \oracle. 
    By Corollary~\ref{thm: our oracle} and Lemma~\ref{lem: concentration of U_tau}, it holds with probability at least $1 - \beta' / 2$, 
        $$
            \left| \hat f_{\U_\tau} [ \pmbs{} ] - f_{\U_\tau } [ \pmbs{} ] \right| \in O \left( \frac{1}{\eps} \sqrt{ | \U_\tau| \ln \frac{ 1 }{\beta'} } \right) \subseteq O \left( \frac{1}{\eps} \sqrt{ \frac{ n }{ L } \ln \frac{ 1 }{\beta'} } \right). 
        $$
    
    By Lemma~\ref{lem: concentration of f_U_tau}, with probability at least $1 - \beta' / 2$, 
        $$
            \left| f_{\U_\tau} [ \pmbs{} ] - f_\mcalU{} [ \pmbs{} ] / L \right| \le \left\Vert f_{\U_\tau } - f_\mcalU{} / L \right\Vert_2 
            \in O \left( \sqrt{ \frac{ n }{ L } \cdot \ln \frac{L}{\beta'} } \right)\,.
        $$
    By a union bound, with probability at least $1 -  \beta'$, it holds that 
    $$
        \left| L \cdot \hat f_{\U_\tau} [ \pmbs{} ] - f_\mcalU{} [ \pmbs{} ] \right| \in O \left( \frac{1}{\eps} \sqrt{ n \cdot L \cdot \ln \frac{ 1 }{\beta'} } \right)\,.
    $$
    
    We finish the proof by substituting $L = 2 \cdot (\log d) / \log n$. 

\end{proof}

In what follows, we need to prove Lemma~\ref{lem: concentration of U_tau} and Lemma~\ref{lem: concentration of f_U_tau}.

First, we prove Lemma~\ref{lem: concentration of U_tau}.
The lemma holds trivially for {\bf permutation partitioning}, as in such case it holds that~$|\U_\tau| = |\mcalU{}| / L$ for all $\tau \in [L]$.  
We need to prove the lemma for \textbf{independent partitioning}.

\subsubsection{Proof of Lemma~\ref{lem: concentration of U_tau} for {\bf Independent Partitioning}}

    Let us fix a $\tau \in [L]$, for each $u \in [n]$, define the indicator random variable $X_u$ that equals $1$ if $u \in \U_\tau$ and $0$ otherwise. Then $| \U_\tau | = \sum_{u \in [n] } X_u$, $\E[ | \U_\tau | ] = n / L$. By Chernoff bound (Fact~\ref{fact: prototype chernoff bound}), it holds that 
    $$
        \Pr[ | \U_\tau | > e n / L ] \le \left( \frac{ e^{e - 1} }{ e^e } \right)^{n / L} = \frac{1}{ e^{n / L} }\,.
    $$
    
    Recall that $L = 2 \cdot (\log d) / \log n$. Further, in order that the bound $O \left( \frac{1}{\eps} \sqrt{ n \cdot \frac{\log d}{\log n} \cdot \ln \frac{1}{\beta'}} \right)$ in Theorem~\ref{thm: hadaheavy estiamtion error of single element} to be meaningful, we need the assumption that $n \ge \frac{1}{\eps^2} \cdot L \cdot \ln \frac{1}{\beta'}$. Therefore, 
    $$
        \frac{1}{ e^{n / L} } \le e^{ - \ln (1 / \beta') } = \beta'\,.
    $$
    It concludes that  $|\U_\tau| \in O( n / L)$ with probability at least~$1 - \beta'$.

\noindent $\blacksquare$

\vspace{8mm}

This finishes the proof of Lemma~\ref{lem: concentration of U_tau}. 
Next, we prove Lemma~\ref{lem: concentration of f_U_tau}.
We need to prove the lemma for both {\bf permutation partitioning} and {\bf independent partitioning}.

\subsubsection{Proof of Lemma~\ref{lem: concentration of f_U_tau} for \textbf{permutation partitioning}}

    Consider a fixed $\tau \in [L]$. 
    We have
    $$
        \left\Vert f_{\U_\tau } - \frac{ f_\mcalU }{ L } \right\Vert_2 - \E \left[ \left\Vert f_{\U_\tau } - \frac{ f_\mcalU }{ L } \right\Vert_2 \right] + \E \left[ \left\Vert f_{\U_\tau } - \frac{ f_\mcalU }{ L } \right\Vert_2 \right].
    $$
    We will bound each term separately. 
    
    \noindent \textbf{Step 1: Bounding $\E \left[ \left\Vert f_{\U_\tau } - { f_\mcalU } / { L } \right\Vert_2 \right]$.} 
    
    By Jensen's inequality, 
    \begin{align*}
        \E \left[ \left\Vert f_{\U_\tau } - \frac{ f_\mcalU }{ L }  \right\Vert_2 \right] \le \sqrt{ \sum_{ \pmbs{} \in \Lambda^\tau } \E \left[ \left( f_{\U_\tau } [ \pmbs{} ] - \frac{ 1 }{ L } f_\mcalU [ \pmbs{} ]  \right)_2^2 \right] } = \sqrt{ \sum_{ \pmbs{} \in \Lambda^\tau } \Var \left[  f_{\U_\tau } [ \pmbs{} ] \right] }\,.
    \end{align*}
    
    For each user $u \in \mcalU{}$, define $\pmbs{}^{ (u) } \doteq v^{ (u) } [ 1 : \tau ]$, the prefix of $v^{ (u) }$ with length $\tau$. 
    Consider a fixed $\pmbs{} \in \Lambda^\tau$, define $\U[\pmbs{} ] \doteq \{ u \in \mcalU{} : \pmbs{}^{ (u) } = \pmbs{} \}$ as the set of users holding element $\pmbs{}$. 
    It holds that $| \U{}[\pmbs{} ] | = f_\mcalU [ \pmbs{} ]$.
    For all $u \in \U [\pmbs{}]$, define the indicator random variable $Z^{ (u) }$ to represent the event $u \in  \U_\tau $.
    Then $\Pr[ Z^{ (u) } = 1 ] = 1 / L$ and $\Pr[ Z^{ (u) } = 0 ] = 1 - 1 / L$. 
    Then 
    $$
        f_{ \U_\tau } [ \pmbs{} ] = \sum_{ u \in \U [ \pmbs{} ] } Z^{ (u) }\,,
    $$
    is a sum of~$f_\mcalU[ \pmbs{} ]$ dependent random variables with expectation $f_\mcalU[ \pmbs{} ] / L$. For a pair of users $u, u' \in \U{} [ \pmbs{} ], u \neq u'$, due to the permutation, if~$u$ belongs to~$\U_\tau$, it is less likely that~$u'$ belongs to~$\U_\tau$. 
    In particular, 
    \begin{align*}
       \Cov{} \left[ Z^{ (u) }, Z^{ (u') } \right] 
        &= \E\left[ Z^{ (u) } \cdot Z^{ (u') } \right] - \E\left[ Z^{ (u) } \right] \E \left[ Z^{ (u') } \right] \\
        &= \binom{n - 2}{n / L - 2} / \binom{n}{n / L} - \left( \frac{1}{L} \right)^2 
        = \frac{ n / L \cdot (n / L - 1)}{n (n - 1)} - \left( \frac{1}{L} \right)^2 
        \le 0\,.
    \end{align*}
    Hence, 
    \begin{align*}
        \Var{} \left[  f_{ \U_\tau } [ \pmbs{} ] \right] 
            &=  \sum_{ u \in \U{} [ \pmbs{} ] } \Var{} \left[ Z^{ (u) } \right] +
                \sum_{ u \neq u' \in \U{} [ \pmbs{} ] } \Cov{} \left[ Z^{ (u) }, Z^{ (u') } \right] 
            &\le \sum_{ u \in \U{} [ \pmbs{} ] } \Var{} \left[ Z^{ (u) } \right]
            = f_\mcalU [ \pmbs{} ] \cdot \frac{1}{L} \cdot \left( 1 - \frac{1}{L} \right) \le \frac{ f_\mcalU [ \pmbs{} ] }{ L }\,.
    \end{align*}
    Therefore, 
    $$
        \E\left[ \left\Vert f_{ \U_\tau } - \frac{ f_\mcalU }{ L }  \right\Vert_2 \right] = \sqrt{ \sum_{ \pmbs{} \in \Lambda^\tau } \Var \left[  f_{\U_\tau } [ \pmbs{} ] \right] } \le \sqrt{ \sum_{ \pmbs{} \in \Lambda^\tau  } \frac{ f_\mcalU [ \pmbs{} ] }{ L } } = \sqrt{ \frac{n}{L} }\,.
    $$

    \vspace{5mm}
    \noindent \textbf{Step 2: Bounding $\left\Vert f_{\U_\tau } - { f_\mcalU } / { L } \right\Vert_2 - \E \left[ \left\Vert f_{\U_\tau } - { f_\mcalU } / { L } \right\Vert_2 \right]$. }
    
    By symmetry, we prove just the case when $\tau = 1$.
    
    Without loss of generality, assume that the~$n$ users in~$\mcalU{}$ are indexed by~$[n] = \{1, 2, \ldots, n \}$.
    Let $\mbfX{} = \{ X_1, \ldots, X_n \}$ be a random permutation of~$[n]$.
    As $\left\Vert f_{\U_1 } - { f_\mcalU } / { L }  \right\Vert_2$ is a function that depends on $\mbfX{}$, we write it explicitly as $\left\Vert f_{\U_1 } - { f_\mcalU } / { L }  \right\Vert_2(X_1, \ldots, X_n)$ or $\left\Vert f_{\U_1 } - { f_\mcalU } / { L }  \right\Vert_2(\mbfX{})$ when necessary. 
    For a possible permutation $\pmbx{} = \{ x_1, \ldots, x_n \}$ of $[n]$, we use $\left\Vert f_{\U_1 } - { f_\mcalU } / { L }  \right\Vert_2(x_1, \ldots, x_n)$ or $\left\Vert f_{\U_1 } - { f_\mcalU } / { L }  \right\Vert_2(\pmbx{})$ to denote the value of $\left\Vert f_{\U_1 } - { f_\mcalU } / { L }  \right\Vert_2$ when $\mbfX{} = \pmbx{}$.
    Let $\gsize = n / L$. For each possible permutation $\pmbx{}$, $\left\Vert f_{\U_1 } - { f_\mcalU } / { L }  \right\Vert_2 ( \pmbx{} )$ does not change its value under the change of order of the first~$\gsize$ and/or last $n - \gsize$ coordinates~$\pmbx{}$. Hence, $\left\Vert f_{\U_1 } - { f_\mcalU } / { L }  \right\Vert_2$ is $(\gsize, n - \gsize)$-symmetric (Definition~\ref{def: symmetric function}). 

    We prove that for each possible permutation~$\pmbx{}$, for all $i \in \{1, \ldots, \gsize \}, j \in \{ \gsize + 1, \ldots, n \}$, it holds that 
    \begin{equation} \label{ineq: f tau 1 permutaion swap change}
        \left| \left\Vert f_{\U_1 } - \frac{ f_\mcalU }{ L }  \right\Vert_2 ( \pmbx{} ) - \left\Vert f_{\U_1 } - \frac{ f_\mcalU }{ L }  \right\Vert_2 ( \pmbx{}_{i, j} ) \right| \le \sqrt{2}\,,
    \end{equation}
    where the permutation $\pmbx{}_{i,j}$ is obtained from $\pmbx{}$ by transposition of its $i^{\text{th}}$ and $j^{\text{th}}$ coordinates. 
    Then by the McDiarmid Inequality with respect to permutation (Fact~\ref{fact: McDiarmid Inequality Permutation}), we have that for all~$\eta > 0$, 
    $$
        \Pr \left[ \left\Vert f_{\U_\tau } - \frac{ f_\mcalU }{ L } \right\Vert_2 - \E \left[ \left\Vert f_{\U_\tau } - \frac{ f_\mcalU }{ L } \right\Vert_2 \right] \ge \eta \right] \le \exp \left( - \frac{2 \eta^2 }{ ( {n} / {L} ) \cdot 2 } \left( \frac{n - 1 / 2 }{n - ( {n} / {L} ) } \right) \left( 1 - \frac{1}{2 \max \{ ( {n} / {L} ) , n - ( {n} / {L} ) \} } \right) \right)\,. 
    $$
    
    As $L = 2 (\log d ) / \log n \ge 2$, it holds that $n / L \le n - n / L$. It follows that 
    \begin{align*}
        \left( \frac{n - 1 / 2 }{n - ( {n} / {L} ) } \right) \left( 1 - \frac{1}{2 \max \{ ( {n} / {L} ) , n - ( {n} / {L} ) \} } \right) 
        \ge \left( 1 - \frac{ 1 / 2 }{ ( n - ( {n} / {L} ) ) } \right) 
        \ge \frac{1}{2}\,.
    \end{align*}
    Substituting~$\eta$ with $\sqrt{ 2 \cdot ( {n} / {L} ) \cdot  \ln ( {1} / {\beta'} ) }$, we get
    $$
        \Pr \left[ \left\Vert f_{\U_\tau } - \frac{ f_\mcalU }{ L } \right\Vert_2 - \E \left[ \left\Vert f_{\U_\tau } - \frac{ f_\mcalU }{ L } \right\Vert_2 \right] \ge \eta \right] 
        \le \exp \left( - \frac{2 \eta^2 }{ ( {n} / {L} ) \cdot 2 } \cdot \frac{1}{2} \right)
        = \beta'\,.
    $$
    
    \vspace{2mm}
    \textbf{Proof of Inequality~(\ref{ineq: f tau 1 permutaion swap change}).}
    
    For each~$\pmbs{} \in  \Lambda$, define $f_{\U{}_1, \pmbx{} } [ \pmbs{} ] \doteq | \{ u \in \U{}_1 : v^{ (u) } [1 : 1] = \pmbs{} \} |$ to be the frequency of $\pmbs{}$ in the set $\{ v^{ (u) } [1:1]  : u \in \U{}_1 \}$, when the random permutation~$\mbfX{} = \pmbx{}$.
    Let $f_{\U{}_1, \pmbx{} } \doteq \big( f_{\U{}_1, \pmbx{} } [ \pmbs{} ] : \pmbs{} \in  \Lambda \big)$ be the corresponding frequency vector when~$\mbfX{} = \pmbx{}$.
    Similarly, define~$f_{\U{}_1, \pmbx{}_{i, j} }$ to be the frequency vector when~$\mbfX{} = \pmbx{}_{i, j}$.

    Let~$\pmbs{}^{ (x_{i} ) } = v^{ (x_i) }[1:1]$ be the prefix of user~$x_i$, and~$\pmbs{}^{ (x_j ) } = v^{ (x_j) } [1:1]$ be the prefix of user~$x_j$.
    When the permutation of~$\mbfX{}$ changes from~$\pmbx{}$ to~$\pmbx{}_{i, j}$, user~$x_i$ is removed from~$\U_1$ and user~$x_j$ is added into~$\U_1$.
    The frequency of~$\pmbs{}^{ (x_{i} ) }$ in~$\U_1$ decreases by~$1$, and the frequency of~$\pmbs{}^{ (x_j ) }$ increases by~$1$.
    
    It holds that 
    \begin{align*}
        f_{\U{}_1, \pmbx{}_{i, j} } [ \pmbs{} ]
            = \begin{cases}
                f_{\U{}_1, \pmbx{} } [ \pmbs{} ], &
                \forall \pmbs{} \in  \Lambda \setminus \{ \pmbs{}^{ (x_i) },  \pmbs{}^{ (x_j ) } \}\,, \\
                f_{\U{}_1, \pmbx{} } [ \pmbs{} ] - 1, & \pmbs{} =  \pmbs{}^{ (x_{i} ) }\,,   \\
                f_{\U{}_1, \pmbx{} } [ \pmbs{} ] + 1, &
                \pmbs{} = \pmbs{}^{ (x_j ) }\,.
            \end{cases}
    \end{align*}
    
     If we view $f_{\U_1, \pmbx{} } - { f_\mcalU} / { L }$ and $f_{\U_1, \pmbx{}_{i, j} } - { f_\mcalU } / { L }$ as $|\Lambda|$-dimensional vectors, it holds that 
    $$
        \left( f_{\U_1, \pmbx{} } - \frac{ f_\mcalU }{ L } \right) - \left( f_{\U_1, \pmbx{}_{i, j} } - \frac{ f_\mcalU }{ L } \right) = - \pmbe{}_{ \pmbs{}^{ (x_{i} ) } } + \pmbe{}_{ \pmbs{}^{ (x_j ) } }\,,
    $$
    where $\pmbe{}_{ \pmbs{}^{ (x_{i} ) } }$ and $\pmbe{}_{ \pmbs{}^{ (x_j ) } }$ are the $\pmbs{}^{ (x_{i} ) }$-th and the $\pmbs{}^{ (x_j ) }$-th standard basis vectors in $\R^{ |\Lambda | }$ respectively. By the triangle inequality, 
    $$
        \left\Vert  f_{\U_1, \pmbx{} } - \frac{ f_\mcalU }{ L } \right\Vert_2 - \left\Vert f_{\U_1, \pmbx{}_{i, j} } - \frac{ f_\mcalU }{ L } \right\Vert_2 \le \Vert - \pmbe{}_{ \pmbs{}^{ (x_{i} ) } } + \pmbe{}_{ \pmbs{}^{ (x_j ) } } \Vert_2 = \sqrt{2}\,.
    $$
    
\noindent $\blacksquare$

\vspace{10mm}

This finishes the proof of Lemma~\ref{lem: concentration of f_U_tau} for {\bf permutation partitioning}.
Next, we prove Lemma~\ref{lem: concentration of f_U_tau} for {\bf independent partitioning}.

\subsubsection{Proof of Lemma~\ref{lem: concentration of f_U_tau} for \textbf{independent partitioning}}

    Without loss of generality, we prove this lemma for a fixed $\tau \in [L]$. 
    To simplify the notation, for each user $u \in \mcalU{}$, we write $\pmbs{}^{ (u) } \doteq v^{ (u) } [ 1 : \tau ]$ as the prefix of $v^{ (u) }$ with length $\tau$. 
    For each $u \in \mcalU{}$, define the indicator random variable~$X_u$ for the event~$u \in \U_\tau$.
    Let~$\mbfX{}$ be shorthand for $\{ X_1, \ldots, X_n \}$.
    As $\left\Vert f_{\U_\tau } - { f_\mcalU{} } / { L }  \right\Vert_2$ is a function that depends on $\mbfX{}$, we write it explicitly as $\left\Vert f_{\U_\tau } - { f_\mcalU } / { L }  \right\Vert_2(X_1, \ldots, X_n)$ or $\left\Vert f_{\U_\tau } - { f_\mcalU } / { L }  \right\Vert_2(\mbfX{})$ when necessary. 
    For a sequence of values $\pmbx{} = \{ x_1, \ldots, x_n \} \in \{0, 1\}^n$, we use $\left\Vert f_{\U_\tau } - { f_\mcalU } / { L }  \right\Vert_2(x_1, \ldots, x_n)$ or $\left\Vert f_{\U_\tau } - { f_\mcalU } / { L }  \right\Vert_2(\pmbx{})$ to denote the value of $\left\Vert f_{\U_\tau } - { f_\mcalU } / { L }  \right\Vert_2 ( \mbfX{} )$, when $\mbfX{} = \pmbx{}$. Since
    $$
        \left\Vert f_{\U_\tau } - \frac{ f_\mcalU{} }{ L } \right\Vert_2 
        = \left\Vert f_{\U_\tau } - \frac{ f_\mcalU{} }{ L }  \right\Vert_2 
        - \E \left[ \left\Vert f_{\U_\tau } - \frac{ f_\mcalU{} }{ L }  \right\Vert_2 \right] 
        + \E \left[ \left\Vert f_{\U_\tau } - \frac{ f_\mcalU{} }{ L }  \right\Vert_2 \right]\,,
    $$
    we can bound $\left\Vert f_{\U_\tau } - { f_\mcalU } / { L } \right\Vert_2 - \E \left[ \left\Vert f_{\U_\tau } - { f_\mcalU } / { L }  \right\Vert_2 \right]$ and $\E \left[ \left\Vert f_{\U_\tau } - { f_\mcalU } / { L }  \right\Vert_2 \right]$ separately. 
    
    \noindent \textbf{Step 1: Bounding $\E \left[ \left\Vert f_{\U_\tau } - { f_\mcalU } / { L } \right\Vert_2 \right]$.} 
   
    First, similar to the proof for Lemma~\ref{lem: bounds of Delta 1}, we have that 
    $$
        \E \left[ \left\Vert f_{\U_\tau } - \frac{ f_\mcalU{} }{ L }  \right\Vert_2^2 \right] 
        = \sum_{ \pmbs{} \in \Lambda^\tau } \E \left[ \left( f_{\U_\tau } [ \pmbs{} ] - \frac{ 1 }{ L } f_\mcalU{} [ \pmbs{} ]  \right)^2 \right] 
        = \sum_{ \pmbs{} \in \Lambda^\tau } \Var \left[  f_{\U_\tau } [ \pmbs{} ] \right] \le \sum_{ \pmbs{} \in \Lambda^\tau } \frac{1}{L} f_\mcalU{} [ \pmbs{} ] = \frac{n}{L}\,.
    $$
    
    Hence, by Jensen's inequality, $\E \left[ \left\Vert f_{\U_\tau } - { f_\mcalU } / { L }  \right\Vert_2 \right] \le \sqrt{ \E \left[ \left\Vert f_{\U_\tau } - { f_\mcalU } / { L }  \right\Vert_2^2 \right]  } \in O \left( \sqrt{ {n} / {L} } \right)$\,.

    \vspace{5mm}
    \noindent \textbf{Step 2: Bounding $\left\Vert f_{\U_\tau } - { f_\mcalU } / { L } \right\Vert_2 - \E \left[ \left\Vert f_{\U_\tau } - { f_\mcalU } / { L } \right\Vert_2 \right]$. }

    Bounding this is more nuanced than the equivalent term in the proof of Lemma~\ref{lem: bounds of Delta 1}. 
    We could show that $\left\Vert f_{\U_\tau } - { f_\mcalU{} } / { L } \right\Vert_2$ satisfies the Lipschitz condition (Definition~\ref{def:Lipschitz-Condition}) with bound~$1$ and then apply McDiarmid’s Inequality (Fact~\ref{fact:McDiarmid-Inequality}). 
    But this will give us an inferior bound of $O \left( \sqrt{ n \ln ( {1} / {\beta'} ) } \right)$.
    The Lipschitz condition states that for each $u \in \mcalU{}$, if the value of $X_u$ changes, it affects $\left\Vert f_{\U_\tau } - { f_\mcalU } / { L } \right\Vert_2$ by at most~$1$. 
    This completely ignores the variance of $\left\Vert f_{\U_\tau } - { f_\mcalU } / { L } \right\Vert_2$. 
    
    We apply a martingale concentration inequality that incorporates variances (Fact~\ref{fact: martingale-bernstein}). 
    First, we introduce a dummy variable~$X_0 \equiv 0$. For each $0 \le i \le n$, let~$\mbfS{}_i$ be shorthand for $(X_0, \ldots, X_i)$. Define
    $$
        Y_i \doteq \E \left[ \left\Vert f_{\U_\tau } - \frac{ f_\mcalU{} }{ L } \right\Vert_2 \,\middle\vert\, \mbfS{}_i \right]. 
    $$
    Clearly $Y_i$ is a function of~$X_0, \ldots, X_i$ and $\E[ |Y_i| ] \le \infty$. Moreover, 
    $$
        \begin{aligned}
            \E[ Y_{i + 1} \mid \mbfS{}_i ] 
            &= \E \left[ \E \left[ \left\Vert f_{\U_\tau } - \frac{ f_\mcalU{} }{ L } \right\Vert_2 \,\middle\vert\, \mbfS{}_i, X_{i + 1} \right] \,\middle\vert\, \mbfS{}_i \right] 
            &= \E \left[ \left\Vert f_{\U_\tau } - \frac{ f_\mcalU{} }{ L } \right\Vert_2 \,\middle\vert\, \mbfS{}_i \right] 
            &= Y_i. 
        \end{aligned}
    $$    
    The sequence $Y_0, \ldots, Y_n$ satisfies all conditions specified in Definition~\ref{def: martingale} and is a martingale. By definition, $Y_n = \E \left[ \left\Vert f_{\U_\tau } - { f_\mcalU } / { L } \right\Vert_2 \,\middle\vert\, \mbfS{}_n \right] = \left\Vert f_{\U_\tau } - { f_\mcalU } / { L } \right\Vert_2$ as, once the values of $X_1, \ldots, X_n$ are determined, so is $\left\Vert f_{\U_\tau } - { f_\mcalU } / { L } \right\Vert_2$. And we have $Y_0 = \E \left[ \left\Vert f_{\U_\tau } - { f_\mcalU } / { L } \right\Vert_2 \,\middle\vert\, X_0 \right] = \E \left[ \left\Vert f_{\U_\tau } - { f_\mcalU } / { L } \right\Vert_2 \right]$, as $X_0 \equiv 0$.
    
    Next, we show that for all~$i \in [n]$,
    \begin{enumerate}
        \item $| Y_i - Y_{i - 1} | \le 1$; and
        \item $\Var{}[ Y_i \mid X_0, \ldots, X_{i - 1} ] \le {1} / {L}$.
    \end{enumerate}
    Then by the concentration inequality~(Fact~\ref{fact: martingale-bernstein}), it holds that for all $\eta > 0$, 
    $$
        \Pr[ Y_n - Y_0 \ge \eta ] \le \exp \left( - \frac{\eta^2}{ 2 \left( n / L + \eta / 3 \right) } \right)\,.
    $$
    We want to find an~$\eta$, s.t., $\exp \left( - \frac{\eta^2}{ 2 \left( n / L + \eta / 3 \right) } \right) = \beta'$. 
    This leads to an equation 
    $$
        \eta^2 - \frac{2 \cdot \ln ( 1  /  \beta' ) }{ 3 } \eta - \frac{2 \cdot n \cdot \ln ( 1  /  \beta' ) }{L} = 0\,,
    $$ 
    whose solution gives 
    $$
        \eta = \frac{ \ln ( 1  /  \beta' ) }{ 3 } + \frac{1}{2} \sqrt{ \frac{ 2^2 \cdot \ln^2 ( 1  /  \beta' ) }{ 3^2 } + \frac{8 \cdot n \cdot \ln ( 1  /  \beta' ) }{L}} \le \sqrt{ \frac{2 \cdot n \cdot \ln ( 1  /  \beta' ) }{L} } + \frac{ 2 \cdot \ln ( 1  /  \beta' ) }{ 3 }\,.
    $$ 
    We conclude that, with probability at least $1 - \beta'$, it holds that 
    $$
        \left\Vert f_{\U_\tau } - \frac{ f_\mcalU{} }{ L } \right\Vert_2 - \E \left[ \left\Vert f_{\U_\tau } - \frac{ f_\mcalU{} }{ L }  \right\Vert_2 \right] \le \sqrt{ 2 \cdot \frac{n}{L} \cdot \ln \frac{1}{\beta'}  } + \frac{ 2 \cdot \ln ( 1  /  \beta' ) }{ 3 } \le \left( \sqrt{2} + \frac{2}{3} \right) \cdot \sqrt{ \frac{n}{L} \cdot \ln \frac{1}{\beta'}  }\,. 
    $$
    
    \noindent \textbf{Proving that for all $i \in [n], | Y_i - Y_{i - 1} | \le 1$.}  
    
    For all sequences $\pmbx{}_{i - 1} = (0, x_1, \ldots, x_{i - 1} ) \in \{0, 1\}^i$ and for all $x_i \in \{0, 1\}$, we prove that 
    \begin{equation*} 
         \left| \E \left[ \left\Vert f_{\U_\tau } - \frac{ f_\mcalU{} }{ L } \right\Vert_2 \,\middle\vert\, \mbfS{}_{i - 1} = \pmbx{}_{i - 1} , X_i = x_i \right] - \E \left[ \left\Vert f_{\U_\tau } - \frac{ f_\mcalU{} }{ L } \right\Vert_2 \,\middle\vert\, \mbfS{}_{i - 1} = \pmbx{}_{i - 1} \right] \right| \le 1\,. 
    \end{equation*}
    Note that
    $$
        \E \left[ \left\Vert f_{\U_\tau } - \frac{ f_\mcalU{} }{ L } \right\Vert_2 \,\middle\vert\, \mbfS{}_{i - 1} = \pmbx{}_{i - 1} \right] = \underset{ x \in \{0, 1\} }{\E } \left[ \E \left[ \left\Vert f_{\U_\tau } - \frac{ f_\mcalU{} }{ L } \right\Vert_2 \,\middle\vert\, \mbfS{}_{i - 1} = \pmbx{}_{i - 1} , X_i = x \right] \right].
    $$
    Define 
    \begin{equation*} 
         \gamma \doteq \left| \E \left[ \left\Vert f_{\U_\tau } - \frac{ f_\mcalU{} }{ L } \right\Vert_2 \,\middle\vert\, \mbfS{}_{i - 1} = \pmbx{}_{i - 1} , X_i = 0 \right] - \E \left[ \left\Vert f_{\U_\tau } - \frac{ f_\mcalU{} }{ L } \right\Vert_2 \,\middle\vert\, \mbfS{}_{i - 1} = \pmbx{}_{i - 1} , X_i = 1  \right] \right|\,. 
    \end{equation*}
    
    It suffices to prove that $\gamma \le 1$. 
    Let $\pmbx{}_i^+ = (x_{i + 1}, \ldots, x_n )$ denote a possible sequence in $\{0, 1 \}^{ n - i}$. 
    For $x_i \in \{0, 1 \}$, we write 
    $$
        \left\Vert f_{\U_\tau } - \frac{ f_\mcalU{} }{ L } \right\Vert_2 ( \pmbx{}_{i - 1} , x_i, \pmbx{}_i^+ ) 
    \qquad
    \text{for}
    \qquad
        \left\Vert f_{\U_\tau } - \frac{ f_\mcalU{} }{ L } \right\Vert_2 ( x_0, \ldots, x_{i - 1} , x_i, x_{i + 1}, \ldots, x_n ). 
    $$
    Let $\mbfS{}_i^+$ be shorthand for $\{ X_{i + 1}, \ldots, X_n\}$.
    Since~$\mbfS{}_i^+$ is independent of~$\mbfS{}_i$, we have 
    \begin{align*}
        \gamma 
            &= \left| \sum_{ \pmbx{}_i^+ \in \{0, 1\}^{n - i} } \Pr[ \mbfS{}_i^+ = \pmbx{}_i^+ ] \cdot \left( \left\Vert f_{\U_\tau } - \frac{ f_\mcalU{} }{ L } \right\Vert_2 ( \pmbx{}_{i - 1} , 0, \pmbx{}_i^+ ) - \left\Vert f_{\U_\tau } - \frac{ f_\mcalU{} }{ L } \right\Vert_2 ( \pmbx{}_{i - 1} , 1, \pmbx{}_i^+ ) \right) \right|.
    \end{align*}
    
    Consider a fixed $\pmbx{}_i^+ \in \{0, 1\}^{n - i}$. 
    Let $\pmbx{} = (\pmbx{}_{i - 1}, 0, \pmbx{}_i^+ )$ and~$\pmbx{}' = (\pmbx{}_{i - 1}, 1, \pmbx{}_i^+ )$. Let~$f_{\U_\tau, \pmbx{} }$ and~$f_{\U_\tau, \pmbx{}' }$ be the frequency vectors defined on $\U_\tau$, when~$\mbfX{} = \pmbx{}$ and~$\mbfX{} = \pmbx{}'$, respectively. 
    Let $\pmbs{}^{ (i) } = v^{ (i) } [ 1 : \tau]$ be the prefix with length $\tau$ of user~$i$. 
    For $\pmbs{} \in \Lambda^\tau \setminus \{ \pmbs{}^{ (i) } \}$, it holds that $f_{\U_\tau, \pmbx{} } [ \pmbs{} ] = f_{\U_\tau, \pmbx{}' } [ \pmbs{} ]$. Further, $f_{\U_\tau, \pmbx{} } [ \pmbs{}^{ (i) } ] + 1 = f_{\U_\tau, \pmbx{} '} [ \pmbs{}^{ (i) } ]$. 
    
    If we view $f_{\U_\tau, \pmbx{} } - { f_\mcalU{} } / { L }$ and $f_{\U_\tau, \pmbx{} ' } - { f_\mcalU{} } / { L }$ as $|\Lambda^\tau|$-dimensional vectors, it holds that 
    $$
        \left( f_{\U_\tau, \pmbx{} } - \frac{ f_\mcalU{} }{ L } \right) - \left( f_{\U_\tau, \pmbx{} ' } - \frac{ f_\mcalU{} }{ L } \right) = - \pmbe{}_{ \pmbs{}^{ (i) } }\,, 
    $$
    where $\pmbe{}_{ \pmbs{}^{ (i) } }$ is the $\pmbs{}^{ (i) }$-th standard basis vector. By the triangle inequality, 
    $$
        \left\Vert  f_{\U_\tau, \pmbx{} } - \frac{ f_\mcalU{} }{ L } \right\Vert_2 - \left\Vert f_{\U_\tau, \pmbx{} ' } - \frac{ f_\mcalU{} }{ L } \right\Vert_2 \le \Vert \pmbe{}_{ \pmbs{}^{ (i) } } \Vert_2 = 1\,.
    $$
    
    \vspace{3mm}
    \noindent \textbf{Proving that for all~$i \in [n], \Var{}[ Y_i \mid X_0, \ldots, X_{i - 1} ] \le 1 / L$.} 
    
    For each sequence $\pmbx{}_{i - 1}$, 
    \begin{align*}
        \E \left[ \left\Vert f_{\U_\tau } - \frac{ f_\mcalU{} }{ L } \right\Vert_2 \,\middle\vert\, \mbfX{}_{i - 1} = \pmbx{}_{i - 1} \right]
        &= 
        \begin{cases} 
            \E \left[ \left\Vert f_{\U_\tau } - \frac{ f_\mcalU{} }{ L } \right\Vert_2 \,\middle\vert\, \mbfX{}_{i - 1} = \pmbx{}_{i - 1} , X_i = 1 \right]\,, 
            & \text{w.p.}\, 1 / L\,, \\
            \E \left[ \left\Vert f_{\U_\tau } - \frac{ f_\mcalU{} }{ L } \right\Vert_2 \,\middle\vert\, \mbfX{}_{i - 1} = \pmbx{}_{i - 1} , X_i = 0 \right]\,,
            & \text{w.p.}\, 1 - 1 / L \,. \\
        \end{cases}
    \end{align*}
    
    We have also proven that 
    $$
        \left| \E \left[ \left\Vert f_{\U_\tau } - \frac{ f_\mcalU{} }{ L } \right\Vert_2 \,\middle\vert\, \mbfX{}_{i - 1} = \pmbx{}_{i - 1} , X_i = 0 \right] - \E \left[ \left\Vert f_{\U_\tau } - \frac{ f_\mcalU{} }{ L } \right\Vert_2 \,\middle\vert\, \mbfX{}_{i - 1} = \pmbx{}_{i - 1} , X_i = 1  \right] \right| \le 1.
    $$
    It follows that $\Var{} \left[ Y_i \mid \mbfX{}_{i - 1} = \pmbx{}_{i - 1} \right] \le 1 / L \cdot (1 - 1 / L) \le 1 / L$.

\noindent $\blacksquare$

\vspace{10mm}

This finishes the proof of Theorem~\ref{thm: hadaheavy estiamtion error of single element}.
Next, we prove Theorem~\ref{thm: elements in P_tau}.

\subsection{Theorem~\ref{thm: elements in P_tau}} \label{appendix: proof for thm: elements in P_tau}

\textbf{Theorem~\ref{thm: elements in P_tau}.}
{\it 
    Let $\lambda \doteq 3 \cdot \lambda'$.
    With probability at least~$1 - \beta$, it is guaranteed that for each $\tau \in [L]$, and each $ \pmb{s} \in \Lambda^\tau$:\, (1) if $f_\mcalU{} [ \pmb{s} ] \ge \lambda$, then $\pmbs{} \in \mcalP{}_\tau$;\, (2) and for each $\pmbs{} \in \mcalP{}_\tau$, the frequency estimate~$\hat f_\mcalU{} [ \pmb{s} ]$ satisfies $|\hat f_\mcalU{} [ \pmb{s} ] - f_\mcalU{} [ \pmb{s} ] | \le \lambda'$.
    Constructing the $\mcalP{}_\tau$ for all $\tau \in [L]$ has~$\tilde{O}(n)$ running time and~$\tilde{O}(\sqrt{n})$ memory usage.
}

\begin{proof}
    We focus on the estimation errors of prefixes from a fixed set. 
    
    {\bf Definition~\ref{def: candidate set}}
        Define $\Gamma_0 \doteq \{ \bot \}$ to be the set of the empty string, and for $\tau \in [L]$, $\Gamma_\tau \doteq \{ \pmbs{} \in \Lambda^\tau : f_\mcalU{} [ \pmbs{} ] \ge \lambda' \}$, the set of prefixes of length~$\tau$ whose frequency is at least~$\lambda'$. 
        For $\tau < L$, the child set of~$\Gamma_\tau$ is defined as $ \Gamma{}_\tau \times \Lambda  \doteq \{ \pmbs{} = \pmbs{}_1 \circ \pmbs{}_2 : \pmbs{}_1 \in \Gamma{}_\tau, \pmbs{}_2 \in \Lambda \}$, where $\pmbs{}_1 \circ \pmbs{}_2$ is the concatenation of $\pmbs{}_1$ and $\pmbs{}_2$. 
        The \emph{candidate} set is defined as $\Gamma \doteq \cup_{0 \le \tau < L } \left( \Gamma{}_\tau \times \Lambda \right)$. 
    \vspace{3mm}
    
    Note that for each $\tau \in [L]$, we have $|\Gamma_\tau | \le n / \lambda' \le \sqrt{n}$. 
    Hence, $|\Gamma| = \sum_{0 \le \tau < L } \left| \Gamma{}_\tau \times \Lambda \right| \le L\sqrt{n} \cdot \sqrt{n} \in \tilde O (n)$. 
    By applying Theorem~\ref{thm: hadaheavy estiamtion error of single element} with $\beta' = \beta / (nL)$ and  the union bound over all~$\pmbs{} \in \Gamma$, we have:
    
    \vspace{1mm}
    {\bf Corollary~\ref{cor: error of candidate set}}
    {\it
        There exists some constant $C_\lambda$, such that with probability at least $1 - \beta$, it holds that
        {\small
            $
                \max_{ \pmbs{} \in \Gamma } | \hat f_\mcalU{} [ \pmbs{} ] - f_\mcalU{} [ \pmbs{} ] | 
                \le \lambda'
            $
        }, where 
        \vspace{-2mm}
        $$
            \lambda' = ( C_{\lambda} / \eps) \sqrt{ n \cdot (\log d) \cdot (\ln (n / \beta) ) / \ln n }\,.
        $$
    }
    \vspace{-2mm}
    
    \begin{lemma} \label{lem: elements in P_tau}
        Suppose that all strings in~$\Gamma$ having estimation error~$\lambda'$, i.e., for each $\pmbs{} \in \Gamma$, it holds that $| \hat f_\mcalU{} [ \pmbs{} ] - f_\mcalU{} [ \pmbs{} ] | \le \lambda'$. 
        It is follows that for each $\tau \in [L]$, and for each $\pmb{s} \in \Lambda^\tau$:\, i) if $f_\mcalU{} [ \pmb{s} ] \ge 3 \cdot \lambda'$, then $\pmbs{} \in \mcalP{}_\tau$;\,
        ii) if $f_\mcalU{} [ \pmb{s} ] < \lambda'$, then $\pmbs{} \notin \mcalP{}_\tau$.
    \end{lemma}
    
    The proof of the lemma is by induction, and is rather technical.
    We defer it to the end of the proof.
    By now, we finish the proof of the theorem based on this lemma. 
    
    Conditioned all strings in~$\Gamma$ having estimation error~$\lambda'$, 
    via the lemma, we see that for each $\tau \in [L]$, and each $ \pmb{s} \in \Lambda^\tau$, if $f_\mcalU{} [ \pmb{s} ] \ge \lambda = 3 \lambda'$, then $\pmbs{}$ is guaranteed to be added to~$\mcalP{}_\tau$.
    
    Moreover, the second condition of the lemma implies that for each~$\tau \in [L]$, and each~$\pmbs{} \in \mcalP{}_\tau$, it holds that~$f_\mcalU{} [ \pmb{s} ] \ge \lambda'$.
    By the definition of~$\Gamma_\tau$, we have~$\pmbs{} \in \Gamma_\tau$. 
    Hence~$\mcalP{}_\tau \subset \Gamma_\tau$.
    By the assumption that all string in~$\Gamma$ have estimation error bounded by~$\lambda'$, we the see that strings in $\mcalP{}_\tau$ have estimation errors bounded by~$\lambda'$.
    Lastly, note that for $0 \le \tau < L$, the {\it modified search strategy} constructs~$\mcalP{}_{\tau + 1}$ based on~$\mcalP{}_\tau$, by checking strings in~$\mcalP{}_\tau \times \Lambda$.
    It follows that all such strings belong to~$\Gamma_\tau \times \Lambda \subset \Gamma$.
    Therefore, the {\it modified search strategy} only inspects the frequencies of the strings from~$\Gamma$, in order to construct the $\mcalP{}_\tau, \tau \in [L]$.
    
    To analyze the running time and memory usage, observe that the {\it modified search strategy} invokes~$L$ frequency oracles.
    By Theorem~\ref{thm: our oracle}, they have total construction time~$\tilde{O}(n)$ and memory usage~$\tilde{O}( \sqrt{n} )$.
    Since each frequency query takes~$\tilde O(1)$ time, and all strings queried belong to~$\Gamma$, the total query time is bounded by $|\Gamma | \in \tilde O(n)$, which finishes the proof.
\end{proof}

\begin{proof}[Proof for Lemma~\ref{lem: elements in P_tau}]
    Suppose that all strings in~$\Gamma$ having estimation error~$\lambda'$, i.e., for each $\pmbs{} \in \Gamma$, it holds that $| \hat f_\mcalU{} [ \pmbs{} ] - f_\mcalU{} [ \pmbs{} ] | \le \lambda'$. 
    We prove the claims by induction.
    
    When $\tau = 1$, $\mcalP{}_{\tau - 1} = \Gamma_{\tau - 1} = \{ \bot \}$. 
    Then $\mcalP{}_0 \times \Lambda = \Gamma_0 \times \Lambda = \Lambda$. 
    For all $\pmbs{} \in \Lambda$, if $f_\mcalU{} [ \pmb{s} ] \ge 3 \cdot \lambda'$, then it holds that 
    $$
        \hat f_\mcalU{} [ \pmbs{} ] \ge f_\mcalU{} [ \pmbs{} ] - \lambda' \ge 2 \lambda'\,. 
    $$
    According to the definition of~$\mcalP{}_1$, the element~$\pmbs{}$ belongs to~$\mcalP{}_1$. 
    On the other hand, if $f_\mcalU{} [ \pmb{s} ] < \lambda'$, then
    $$
        \hat f_\mcalU{} [ \pmbs{} ] \le f_\mcalU{} [ \pmbs{} ] + \lambda' < 2 \lambda'\,.
    $$
    It is guaranteed that $\pmbs{} \notin \mcalP{}_1$. Consequently, for all $\pmbs{} \in \mcalP{}_1$, we have $f_\mcalU{} [ \pmbs{} ] \ge \lambda'$, which implies that $\mcalP{}_1 \subseteq \Gamma_1$.
    
    Let $\tau > 1$ and assume that the claims holds for $\tau -1$.
    For all $\pmbs{} \in \Lambda^\tau$, let $\pmbs{}[1 : \tau - 1] \in \Lambda^{\tau - 1}$ be the prefix of $\pmbs{}$ of length $\tau - 1$. If $f_\mcalU{} [ \pmb{s} ] \ge 3 \cdot \lambda'$, then it holds that 
    $$
            f_\mcalU{} [ \pmbs{} [1: \tau - 1] ] \ge f_\mcalU{} [ \pmbs{} ] \ge 3 \lambda'\,.
    $$
    By the induction hypothesis, $\pmbs{} [1 : \tau - 1] \in \mcalP{}_{\tau - 1} \subseteq \Gamma_{\tau - 1}$. 
    Hence, 
    $\pmbs{} \in \mcalP{}_{\tau - 1} \times \Lambda \subseteq \Gamma_{\tau - 1} \times \Lambda$, and 
    $$
            \hat f_\mcalU{} [ \pmbs{} ] \ge f_\mcalU{} [ \pmbs{} ] - \lambda' \ge 2 \lambda'\,.
    $$
    According to the definition of $\mcalP{}_\tau$, the element $\pmbs{}$ belongs to $\mcalP{}_\tau$. 
    
    On the other hand, if $f_\mcalU{} [ \pmb{s} ] < \lambda'$, there are two possible cases. First, if $\pmbs{} [ 1 : \tau - 1] \notin \mcalP{}_{\tau - 1}$, then by the definition of $\mcalP{}_\tau$, it is guaranteed that $\pmbs{} \notin \mcalP{}_\tau$. Second, if $\pmbs{} [ 1 : \tau - 1] \in \mcalP{}_{\tau - 1}$, by the inductive hypothesis that $\mcalP{}_{\tau - 1} \subset \Gamma_{\tau - 1}$, we have $\pmbs{} \subset \mcalP{}_{\tau - 1} \times \Lambda \subset \Gamma_{\tau - 1} \times \Lambda$. 
    Then
    $$
        \hat f_\mcalU{} [ \pmbs{} ] \le f_\mcalU{} [ \pmbs{} ] + \lambda' < 2 \lambda'\,.
    $$
    It is guaranteed that $\pmbs{} \notin \mcalP{}_\tau$. Consequently, for all $\pmbs{} \in \mcalP{}_\tau$, we have $f_\mcalU{} [ \pmbs{} ] \ge \lambda'$, which implies that $\mcalP{}_\tau \subseteq \Gamma_\tau$.  
     
\end{proof}

%% file: aistats22.bbl
\begin{thebibliography}{30}
\providecommand{\natexlab}[1]{#1}
\providecommand{\url}[1]{#1}
\csname url@samestyle\endcsname
\providecommand{\newblock}{\relax}
\providecommand{\bibinfo}[2]{#2}
\providecommand{\BIBentrySTDinterwordspacing}{\spaceskip=0pt\relax}
\providecommand{\BIBentryALTinterwordstretchfactor}{4}
\providecommand{\BIBentryALTinterwordspacing}{\spaceskip=\fontdimen2\font plus
\BIBentryALTinterwordstretchfactor\fontdimen3\font minus
  \fontdimen4\font\relax}
\providecommand{\BIBforeignlanguage}[2]{{%
\expandafter\ifx\csname l@#1\endcsname\relax
\typeout{** WARNING: IEEEtranN.bst: No hyphenation pattern has been}%
\typeout{** loaded for the language `#1'. Using the pattern for}%
\typeout{** the default language instead.}%
\else
\language=\csname l@#1\endcsname
\fi
#2}}
\providecommand{\BIBdecl}{\relax}
\BIBdecl

\bibitem[Bassily et~al.(2017)Bassily, Nissim, Stemmer, and Thakurta]{BNST17}
R.~Bassily, K.~Nissim, U.~Stemmer, and A.~G. Thakurta, ``Practical locally
  private heavy hitters,'' in \emph{Advances in Neural Information Processing
  Systems 30: Annual Conference on Neural Information Processing Systems 2017,
  December 4-9, 2017, Long Beach, CA, {USA}}, I.~Guyon, U.~von Luxburg,
  S.~Bengio, H.~M. Wallach, R.~Fergus, S.~V.~N. Vishwanathan, and R.~Garnett,
  Eds., 2017, pp. 2288--2296.

\bibitem[Erlingsson et~al.(2014)Erlingsson, Pihur, and Korolova]{EPK14}
{\'{U}}.~Erlingsson, V.~Pihur, and A.~Korolova, ``{RAPPOR:} randomized
  aggregatable privacy-preserving ordinal response,'' in \emph{Proceedings of
  the 2014 {ACM} {SIGSAC} Conference on Computer and Communications Security,
  Scottsdale, AZ, USA, November 3-7, 2014}, G.~Ahn, M.~Yung, and N.~Li,
  Eds.\hskip 1em plus 0.5em minus 0.4em\relax {ACM}, 2014, pp. 1054--1067.

\bibitem[Apple(2017)]{D.P.Apple17}
Apple, ``Learning with privacy at scale,'' \emph{Apple Machine Learning
  Journal}, vol. 2017, no.~1, pp. 1--25, 2017.

\bibitem[Fanti et~al.(2016)Fanti, Pihur, and Erlingsson]{FPE16}
G.~C. Fanti, V.~Pihur, and {\'{U}}.~Erlingsson, ``Building a {RAPPOR} with the
  unknown: Privacy-preserving learning of associations and data dictionaries,''
  \emph{Proc. Priv. Enhancing Technol.}, vol. 2016, no.~3, pp. 41--61, 2016.

\bibitem[Voigt and Von~dem Bussche(2017)]{voigt2017eu}
P.~Voigt and A.~Von~dem Bussche, ``The {EU} general data protection regulation
  ({GDPR}),'' \emph{A Practical Guide, 1st Ed., Cham: Springer International
  Publishing}, vol.~10, p. 3152676, 2017.

\bibitem[Tang et~al.(2017)Tang, Korolova, Bai, Wang, and Wang]{TKBWW17}
J.~Tang, A.~Korolova, X.~Bai, X.~Wang, and X.~Wang, ``Privacy loss in {Apple}'s
  implementation of differential privacy on {MacOS 10.12},'' \emph{CoRR}, vol.
  abs/1709.02753, 2017.

\bibitem[Ding et~al.(2017)Ding, Kulkarni, and Yekhanin]{DKY17}
B.~Ding, J.~Kulkarni, and S.~Yekhanin, ``Collecting telemetry data privately,''
  in \emph{Advances in Neural Information Processing Systems 30: Annual
  Conference on Neural Information Processing Systems 2017, December 4-9, 2017,
  Long Beach, CA, {USA}}, I.~Guyon, U.~von Luxburg, S.~Bengio, H.~M. Wallach,
  R.~Fergus, S.~V.~N. Vishwanathan, and R.~Garnett, Eds., 2017, pp. 3571--3580.

\bibitem[Dwork and Roth(2014)]{DR14}
C.~Dwork and A.~Roth, ``The algorithmic foundations of differential privacy,''
  \emph{Found. Trends Theor. Comput. Sci.}, vol.~9, no. 3-4, pp. 211--407,
  2014.

\bibitem[Bassily and Smith(2015)]{BS15}
R.~Bassily and A.~D. Smith, ``Local, private, efficient protocols for succinct
  histograms,'' in \emph{Proceedings of the Forty-Seventh Annual {ACM} on
  Symposium on Theory of Computing, {STOC} 2015, Portland, OR, USA, June 14-17,
  2015}, R.~A. Servedio and R.~Rubinfeld, Eds.\hskip 1em plus 0.5em minus
  0.4em\relax {ACM}, 2015, pp. 127--135.

\bibitem[Bun et~al.(2019)Bun, Nelson, and Stemmer]{BNS19}
M.~Bun, J.~Nelson, and U.~Stemmer, ``Heavy hitters and the structure of local
  privacy,'' \emph{{ACM} Trans. Algorithms}, vol.~15, no.~4, pp. 51:1--51:40,
  2019.

\bibitem[Cormode et~al.(2021)Cormode, Maddock, and Maple]{cormode2021frequency}
G.~Cormode, S.~Maddock, and C.~Maple, ``Frequency estimation under local
  differential privacy [experiments, analysis and benchmarks],'' \emph{CoRR},
  vol. abs/2103.16640, 2021.

\bibitem[Warner(1965)]{Warner65}
S.~L. Warner, ``Randomized response: A survey technique for eliminating evasive
  answer bias,'' \emph{Journal of the American Statistical Association},
  vol.~60, no. 309, pp. 63--69, 1965.

\bibitem[Wang et~al.(2017)Wang, Blocki, Li, and Jha]{WangBLJ17}
T.~Wang, J.~Blocki, N.~Li, and S.~Jha, ``Locally differentially private
  protocols for frequency estimation,'' in \emph{26th {USENIX} Security
  Symposium, {USENIX} Security 2017, Vancouver, BC, Canada, August 16-18,
  2017}, E.~Kirda and T.~Ristenpart, Eds.\hskip 1em plus 0.5em minus
  0.4em\relax {USENIX} Association, 2017, pp. 729--745.

\bibitem[Bassily et~al.(2020)Bassily, Nissim, Stemmer, and Thakurta]{BNST20}
R.~Bassily, K.~Nissim, U.~Stemmer, and A.~Thakurta, ``Practical locally private
  heavy hitters,'' \emph{J. Mach. Learn. Res.}, vol.~21, pp. 16:1--16:42, 2020.

\bibitem[Wang et~al.(2018)Wang, Xiao, Yang, Hoang, Shin, Shin, and
  Yu]{WangXYHSS018}
N.~Wang, X.~Xiao, Y.~Yang, T.~D. Hoang, H.~Shin, J.~Shin, and G.~Yu,
  ``Privtrie: Effective frequent term discovery under local differential
  privacy,'' in \emph{34th {IEEE} International Conference on Data Engineering,
  {ICDE} 2018, Paris, France, April 16-19, 2018}.\hskip 1em plus 0.5em minus
  0.4em\relax {IEEE} Computer Society, 2018, pp. 821--832.

\bibitem[Nguy{\^{e}}n et~al.(2016)Nguy{\^{e}}n, Xiao, Yang, Hui, Shin, and
  Shin]{NXYHSS16}
T.~T. Nguy{\^{e}}n, X.~Xiao, Y.~Yang, S.~C. Hui, H.~Shin, and J.~Shin,
  ``Collecting and analyzing data from smart device users with local
  differential privacy,'' \emph{CoRR}, vol. abs/1606.05053, 2016.

\bibitem[Cormode et~al.(2019)Cormode, Kulkarni, and Srivastava]{CKS19}
G.~Cormode, T.~Kulkarni, and D.~Srivastava, ``Answering range queries under
  local differential privacy,'' \emph{Proc. {VLDB} Endow.}, vol.~12, no.~10,
  pp. 1126--1138, 2019.

\bibitem[Cormode and Yi(2020)]{cormode2020small}
G.~Cormode and K.~Yi, \emph{Small Summaries for Big Data}.\hskip 1em plus 0.5em
  minus 0.4em\relax Cambridge University Press, 2020.

\bibitem[Mitzenmacher and Upfal(2017)]{MU17}
M.~Mitzenmacher and E.~Upfal, \emph{Probability and computing: Randomization
  and probabilistic techniques in algorithms and data analysis}.\hskip 1em plus
  0.5em minus 0.4em\relax Cambridge University Press, 2017.

\bibitem[Chung and Lu(2006)]{FL06}
F.~Chung and L.~Lu, ``Concentration inequalities and martingale inequalities: a
  survey,'' \emph{Internet Math.}, vol.~3, no.~1, pp. 79--127, 2006.

\bibitem[Wang et~al.(2021)Wang, Li, and Jha]{WLJ21}
T.~Wang, N.~Li, and S.~Jha, ``Locally differentially private heavy hitter
  identification,'' \emph{{IEEE} Trans. Dependable Secur. Comput.}, vol.~18,
  no.~2, pp. 982--993, 2021.

\bibitem[Acharya et~al.(2019)Acharya, Sun, and Zhang]{AcharyaSZ19}
J.~Acharya, Z.~Sun, and H.~Zhang, ``Hadamard response: Estimating distributions
  privately, efficiently, and with little communication,'' in \emph{The 22nd
  International Conference on Artificial Intelligence and Statistics, {AISTATS}
  2019, 16-18 April 2019, Naha, Okinawa, Japan}, ser. Proceedings of Machine
  Learning Research, K.~Chaudhuri and M.~Sugiyama, Eds., vol.~89.\hskip 1em
  plus 0.5em minus 0.4em\relax {PMLR}, 2019, pp. 1120--1129.

\bibitem[Ghazi et~al.(2021)Ghazi, Golowich, Kumar, Pagh, and
  Velingker]{GhaziG0PV21}
B.~Ghazi, N.~Golowich, R.~Kumar, R.~Pagh, and A.~Velingker, ``On the power of
  multiple anonymous messages: Frequency estimation and selection in the
  shuffle model of differential privacy,'' in \emph{Advances in Cryptology -
  {EUROCRYPT} 2021 - 40th Annual International Conference on the Theory and
  Applications of Cryptographic Techniques, Zagreb, Croatia, October 17-21,
  2021, Proceedings, Part {III}}, ser. Lecture Notes in Computer Science,
  A.~Canteaut and F.~Standaert, Eds., vol. 12698.\hskip 1em plus 0.5em minus
  0.4em\relax Springer, 2021, pp. 463--488.

\bibitem[Charikar et~al.(2002)Charikar, Chen, and Farach{-}Colton]{CCF02}
M.~Charikar, K.~C. Chen, and M.~Farach{-}Colton, ``Finding frequent items in
  data streams,'' in \emph{Automata, Languages and Programming, 29th
  International Colloquium, {ICALP} 2002, Malaga, Spain, July 8-13, 2002,
  Proceedings}, ser. Lecture Notes in Computer Science, P.~Widmayer, F.~T.
  Ruiz, R.~M. Bueno, M.~Hennessy, S.~J. Eidenbenz, and R.~Conejo, Eds., vol.
  2380.\hskip 1em plus 0.5em minus 0.4em\relax Springer, 2002, pp. 693--703.

\bibitem[Luo et~al.(2021)Luo, Wang, and Yi]{LWY21}
Q.~Luo, Y.~Wang, and K.~Yi, ``Frequency estimation in the shuffle model with
  (almost) a single message,'' \emph{CoRR}, vol. abs/2111.06833, 2021.

\bibitem[Balle et~al.(2019)Balle, Bell, Gasc{\'{o}}n, and Nissim]{BalleBGN19}
B.~Balle, J.~Bell, A.~Gasc{\'{o}}n, and K.~Nissim, ``The privacy blanket of the
  shuffle model,'' in \emph{Advances in Cryptology - {CRYPTO} 2019 - 39th
  Annual International Cryptology Conference, Santa Barbara, CA, USA, August
  18-22, 2019, Proceedings, Part {II}}, ser. Lecture Notes in Computer Science,
  A.~Boldyreva and D.~Micciancio, Eds., vol. 11693.\hskip 1em plus 0.5em minus
  0.4em\relax Springer, 2019, pp. 638--667.

\bibitem[Audibert et~al.(2009)Audibert, Munos, and Szepesv{\'{a}}ri]{AMS09}
J.~Audibert, R.~Munos, and C.~Szepesv{\'{a}}ri, ``Exploration-exploitation
  tradeoff using variance estimates in multi-armed bandits,'' \emph{Theor.
  Comput. Sci.}, vol. 410, no.~19, pp. 1876--1902, 2009.

\bibitem[Devroye and Lugosi(2001)]{DL01}
L.~Devroye and G.~Lugosi, \emph{Combinatorial methods in density estimation},
  ser. Springer series in statistics.\hskip 1em plus 0.5em minus 0.4em\relax
  Springer, 2001.

\bibitem[Motwani and Raghavan(1995)]{MR95}
R.~Motwani and P.~Raghavan, \emph{Randomized Algorithms}.\hskip 1em plus 0.5em
  minus 0.4em\relax Cambridge University Press, 1995.

\bibitem[Tolstikhin(2017)]{TIO17}
I.~O. Tolstikhin, ``Concentration inequalities for samples without
  replacement,'' \emph{Theory of Probability \& Its Applications}, vol.~61,
  no.~3, pp. 462--481, 2017.

\end{thebibliography}
